\documentclass[journal=jpcafh,manuscript=article]{achemso}

\usepackage[version=3]{mhchem} 
\usepackage{amsmath,amsthm} 
\usepackage{amssymb,mathrsfs} 
\usepackage{a4wide} 
\usepackage{graphicx}
\usepackage{physics}
\usepackage{color,subcaption} 
\usepackage{enumerate}
\usepackage{comment}
\usepackage{makecell}
\usepackage{todonotes}

\newtheorem{lemma}{Lemma}

\newtheorem{remark}{Remark}

\SectionNumbersOn
\DeclareMathAlphabet{\mathpzc}{OT1}{pzc}{m}{it}
\newcommand{\dps}{\displaystyle } 
\newcommand{\rme}{\mathrm{e}}

\renewcommand{\leq}{\leqslant}
\renewcommand{\geq}{\geqslant}
\renewcommand{\le}{\leqslant}
\renewcommand{\ge}{\geqslant}
\newcommand{\fenc}{f_\mathrm{enc}}
\newcommand{\fass}{f_\mathrm{ass}}
\newcommand{\fred}{f_\mathrm{LOP}}
\newcommand{\wfenc}{\widetilde{f}_\mathrm{enc}}
\newcommand{\fdec}{f_\mathrm{dec}}
\newcommand{\cX}{\mathcal{X}}
\newcommand{\cZ}{\mathcal{Z}}
\newcommand{\sL}{\mathscr{L}}
\newcommand{\sLh}{\widehat{\mathscr{L}}}
\newcommand{\Nd}{N_\mathrm{data}}
\newcommand{\wNd}{\widetilde{N}_\mathrm{data}}
\newcommand{\E}{\mathbb{E}}
\newcommand{\cF}{\mathcal{F}}
\newcommand{\sS}{\mathscr{S}}
\newcommand{\sI}{\mathscr{I}}
\newcommand{\sC}{\mathscr{C}}
\newcommand{\sD}{\mathscr{D}}
\newcommand{\cD}{\mathcal{D}}
\newcommand{\chis}{\chi_\varepsilon}
\newcommand{\wx}{\widetilde{x}}
\newcommand{\cFenc}{\mathcal{F}_\mathrm{enc}}
\newcommand{\cFred}{\mathcal{F}_\mathrm{LOP}}
\newcommand{\cFdec}{\mathcal{F}_\mathrm{dec}}
\renewcommand{\div}{\mathrm{div}}

\allowdisplaybreaks

\author{Tony Lelièvre}
\affiliation[Cermics]
{CERMICS, École des Ponts ParisTech, 6-8 Avenue Blaise Pascal, 77455,Marne-la-Vallée, France}
\alsoaffiliation[Matherials]
{MATHERIALS team-project, Inria Paris, 2 Rue Simone Iff, 75012 Paris, France}
\email{tony.lelievre@enpc.fr}

\author{Thomas Pigeon}
\affiliation[Matherials]
{MATHERIALS team-project, Inria Paris, 2 Rue Simone Iff, 75012 Paris, France}
\email{thomas.pigeon@ifpen.fr}
\alsoaffiliation[Cermics]
{CERMICS, École des Ponts ParisTech, 6-8 Avenue Blaise Pascal, 77455,Marne-la-Vallée, France}
\alsoaffiliation[IFPEN Solaize]
{IFP Energies Nouvelles, Rond-Point de l’Echangeur de Solaize, BP 3, 69360 Solaize, France}

\author{Gabriel Stoltz}
\affiliation[Cermics]
{CERMICS, École des Ponts ParisTech, 6-8 Avenue Blaise Pascal, 77455,Marne-la-Vallée, France}
\email{gabriel.stoltz@enpc.fr}
\alsoaffiliation[Matherials]
{MATHERIALS team-project, Inria Paris, 2 Rue Simone Iff, 75012 Paris, France}

\author{Wei Zhang}
\affiliation[DMCS]
{Department of Mathematics and Computer Science, Freie Universit{\"a}t Berlin,  Arnimallee 14, 14195 Berlin, Germany}
\email{wei.zhang@fu-berlin.de}
\alsoaffiliation[ZIB]
{Zuse Institute Berlin, Takustraße 7, 14195 Berlin, Germany}

\title{Analyzing multimodal probability measures with autoencoders}

\abbreviations{AE, PCA, MEP}

\begin{document}

\begin{tocentry}
\includegraphics[width=1.0\textwidth]{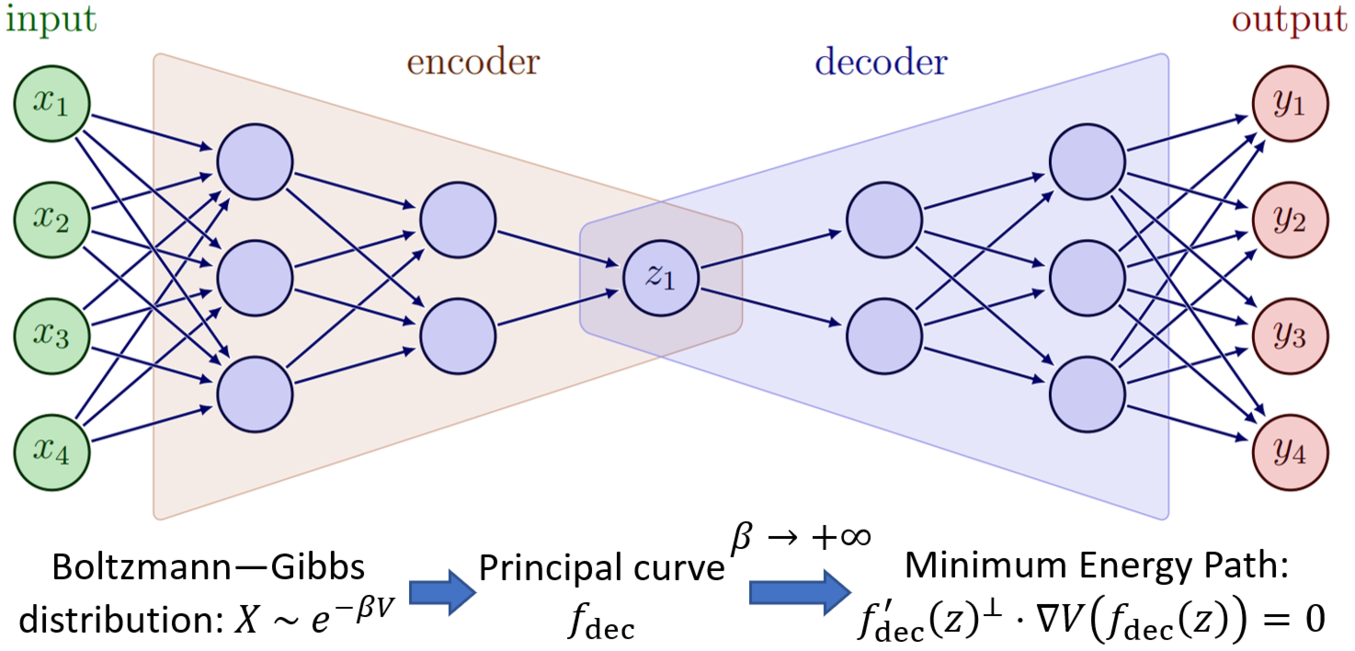}





\end{tocentry}

\begin{abstract}
  Finding collective variables to describe some important coarse-grained information on physical systems, in particular metastable states, remains a key issue in molecular dynamics. Recently, machine learning techniques have been intensively used to complement and possibly bypass expert knowledge in order to construct collective variables. Our focus here is on neural network approaches based on autoencoders. We study some relevant mathematical properties of the loss function considered for training autoencoders, and provide physical interpretations based on conditional variances and minimum energy paths. We also consider various extensions in order to better describe physical systems, by incorporating more information on transition states at saddle points, and/or allowing for multiple decoders in order to describe several transition paths. Our results are illustrated on toy two dimensional systems and on alanine dipeptide.
\end{abstract}

\section{Introduction}

Molecular simulation has proved to be an effective computational approach to understand complex systems in biophysics, chemistry, material science, etc.~\cite{HOLLINGSWORTH20181129-md-for-all,Durrant2011-md-drug}. However, sampling complex physical systems in molecular dynamics remains a computational challenge due to the metastability of the system's dynamics -- by which we mean that the system is often stuck for a long time in some state (called ``metastable state''), before switching to another state where it again remains stuck for a long time. Many efforts have been, and still are pursued to enhance the sampling of molecular systems. One successful strategy to this end relies on free energy biasing where (a fraction of) the free energy is used as an importance sampling function~\cite{colvar}. This paradigm underpins adaptive methods to compute the free energy such as the celebrated metadynamics and its variants~\cite{metadynamics-2002,using-metadynamics-2020}, and the adaptive biasing force method~\cite{abf-darve}, to quote only two options (see Ref.~\citenum{enhanced-sampling-for-md-review} for a more extensive review). A key element in these free energy biasing methods is the choice of collective variables (CVs), which summarize some coarse-grained information on the system, and describe in a lower dimensional space the various metastable states. 

CVs were chosen until recently based on expert knowledge and chemical intuition. Following the many successes of machine learning approaches in data science and computer vision, the last decade witnessed many propositions to improve, complement, or even replace expert knowledge in molecular dynamics. There are by now various reviews on the use of machine learning techniques in molecular dynamics, including also the construction of empirical force fields (see Ref.~\citenum{WGSM21} which summarizes various other review papers). Dedicated review articles are devoted to finding CVs~\cite{Ferguson_2018,SCF20_review,NTMC20,Gkeka_CECAM,GHRCNL21,Chen21}; see also Ref.~\citenum{SKH23} for more mathematically oriented elements. These methods mostly correspond to unsupervised dimensionality reduction and manifold learning techniques, although some of them are based on (semi)supervised approaches.

In order to make the best use of machine learning methods, one needs to specify criteria which make CVs appropriate. Desirable requirements include:
\begin{itemize}
\item free energy biasing based on these CVs efficiently suppresses or reduces metastability. From a mathematical viewpoint, this can be quantified through improved constants in functional inequalities (see for instance Ref.~\citenum{lelievre2013two} for a pedagogical introduction);
\item the CVs provide a parametrization of the eigenfunctions associated with the dominant eigenvalues of the transfer operators (or the generator of the dynamics), as these eigenfunctions allow to characterize metastable states~\cite{msm_generation};  
\item the transition path in the collective variable space from one metastable state to another should correspond to some form of minimum energy or free energy path.~\cite{VEV09}
\end{itemize}
Some CVs meet only some of these requirements. For instance, eigenfunctions have sharp transitions between metastable states, and may therefore not be used as such as CVs since otherwise the biasing forces used in free energy biasing dynamics, which involve the gradient of the collective variable, would be large and lead to unstable dynamics. The same remark applies to the committor function.

Most machine learning methods to find CVs are unsupervised (see however Refs.~\citenum{Bonati2020,DeepTDA} for supervised methods based on Fisher's linear discriminant analysis and some variation of it, Ref.~\citenum{automated-design-cv-supervised-ml} for a method utilizing decision functions such as support vector machines in a supervised setting, and Ref.~\citenum{BBMLSG23} for a semi-supervised approach). There are two main classes of unsupervised methods: those seeking high variances CVs, which aim at reproducing overall features of the Boltzmann--Gibbs distribution at hand; and those seeking slowly evolving CVs (e.g. based on transition path sampling~\cite{cv-by-likelihood-maximization} or committor analysis~\cite{cv-by-cross-entropy-minimization}; see also Ref.~\citenum{zhang2023understanding} for related discussions). Both classes can be separated into linear and nonlinear methods. Our focus in this work is on methods that find high variance CVs, for which no notion of dynamics is a priori used. It is useful to further distinguish between (i) linear methods, e.g. principal component analysis (PCA) or factor analysis, and (ii) non-linear methods, e.g. kernel methods, autoencoders, decision trees and random forests (see for instance Refs.~\citenum{Mehta2019,Murphy22} for introductory references on these classes of methods). In the first case, CVs are interpretable, but often too simple to give precise results. An effective way to find CVs with good intepretability and expressivity is to make a selection from a pool of candidate physical variables~\cite{automatic-identify-rc-ma-and-dinner,cv-genetic-algorithm,cv-by-cross-entropy-minimization}. 

\paragraph{Aims and scope.}
We focus in this work on autoencoders, which have been successfully used in molecular dynamics (see Ref.~\citenum{CC23} and Section~\ref{sec:presentation_AEs} for reviews). Our aim is to better characterize the learning problem associated with autoencoders, and relate it to various frameworks relevant for statistical methods (principal curves and manifolds) and/or molecular dynamics (use of conditional expectations as for free energy computations).

We also aim at better understanding what is learned when minimizing the reconstruction error, and how/whether the so-obtained information can be leveraged. Let us indeed emphasize here that the CVs obtained from maximizing the variance have a priori no dynamical relevance to describe the dynamics of the system, in particular transitions from one metastable mode to another. It is however empirically observed in various works that the CVs found by autoencoders may have some relevance to describe such transitions, possibly upon filtering out some degrees of freedom of the encoder which account for the overall variability of the system under consideration (see for instance Ref.~\citenum{BBMLSG23}). 

In order to further improve the dynamical relevance of encoders, we explore several options in this work: (i) adding extra terms in the loss function, to describe transitions from one metastable state to another, following up on Ref.~\citenum{RBGMM22}; (ii) allowing for multiple decoders in order to represent multiple transition paths connecting two metastable states; or (iii) requiring the encoder to parametrize dominant eigenfunctions.

\paragraph{Outline.}
In the Methods section, we first analyze and provide some theoretical understanding of autoencoders in Sections~\ref{sec:presentation_AEs} to~\ref{sec:cond_exp_others}, and then propose various extensions in Sections~\ref{sec:changing_ref_measure} to \ref{sec:parametrzing_eigenfunc_multiple_states}. More precisely, we briefly present in Section~\ref{sec:presentation_AEs} autoencoders (in particular, bottleneck autoencoders) and some of their properties, as well as some of their applications in molecular dynamics. After a general discussion on the loss function to train autoencoders in Section~\ref{sec:interpretation_loss}, we provide in Section~\ref{sec:cond_exp} a focus on their interpretation in terms of conditional expectations, and use this reformulation to draw a link between AE and PCA or clustering in Section~\ref{sec:cond_exp_others}. We then first consider in Section~\ref{sec:changing_ref_measure} a supervised setting where the probability measure on the data has only two modes which are known a priori, discussing options to construct a single path relating the two modes, and incorporating useful information on the saddle points and transition states observed during the transition. The situation when the modes can be related via several paths is then addressed in a second stage in Section~\ref{sec:multiple_path_2_state}. An approach to regularize autoencoders using the leading eigenfunctions of the transfer operator in the case of multiple metastable states is presented in Section~\ref{sec:parametrzing_eigenfunc_multiple_states}. In the Results section, applications to toy two dimensional systems and alanine dipeptide are reported. More precisely, we first illustrate in Section~\ref{sec:interpretation_numerics} that autoencoders allow to compute conditional expectations. We next demonstrate in Section~\ref{sec:numerics_modifying_proba_dist} that the quality of the autoencoder model can be improved by incorporating more information on saddle-points in the dataset. We then present in Section~\ref{sec:numerics_multiple_paths} results for situations where there are several transition paths, with an autoencoder model made of a single encoder and several decoders. We finally discuss in Section~\ref{sec:ad} the use of regularization terms based on transfer operators for alanine dipeptide. We conclude with some perspectives of this work. 

Let us emphasize that our baseline assumption in all this work is that the dataset at hand is sufficiently rich to correctly describe all the modes of the probability measure, and sometimes maybe also the transition zones. The latter situation is somewhat unrealistic for actual systems of interest, but can be achieved by an iterative procedure where one cycles between exploration phases using a dynamics biased by the free energy, and a learning phase to update the collective variables (as done in various works; see references in Section~\ref{sec:presentation_AEs}).

\section{Methods}

%

\subsection{Presentation of autoencoders}
\label{sec:presentation_AEs}

Autoencoders have been considered early on in the neural network literature, where they were also called auto-associative neural networks~\cite{Kramer91}. The models considered in these early works correspond to what is currently known as bottleneck autoencoders, and were rather shallow. Deep autoencoders were used later on with the advent of modern computing architectures~\cite{HS06}. Bottleneck autoencoders were initially introduced to provide a nonlinear generalization of PCA, as it was shown that the linear neural networks obtained by minimizing the mean-square error were essentially equivalent to PCA~\cite{BK88,BH89}. We refer for instance to Section~12.4.2 of Ref.~\citenum{Bishop06}, Chapter~14 of Ref.~\citenum{GBC16} and Section~20.3 of Ref.~\citenum{Murphy22} for textbook presentations of autoencoders, which include some historical perspectives, and discuss the many variations and extensions that were considered.

Autoencoders fall into the class of unsupervised machine learning methods. For a given input data point~$x \in \cX \subset \mathbb{R}^D$, we denote by $f_\theta(x)$ the prediction of the neural network. The parameters~$\theta \in \Theta$ are chosen to minimize the loss function
\begin{equation}
\label{eq:population_loss}
\sL(\theta) = \E[\ell(X,f_\theta(X))],
\end{equation}
where $\ell$ is a given elementary loss function, and the expectation is over the realizations of the input data~$X$ distributed according to some probability measure denoted by~$\mu$. By default, in the sequel, expectations are always taken with respect to the distribution of the data. The typical choice for the latter elementary loss function is the square loss~$\ell(x,y) = \|x-y\|^2$, although other choices, such as the mean absolute loss~$\ell(x,y) = \|x-y\|$ could also be considered in order to give less weight to outliers. In practice, the population loss~$\sL$ is replaced by the empirical loss over a training set of~$\Nd$ given input data points~$\{x^1,\dots,x^{\Nd}\}$:
\[
\sLh(\theta) = \frac{1}{\Nd} \sum_{n=1}^{\Nd} \ell(x^n,f_\theta(x^n)).
\]

\paragraph{Families of autoencoders.}
There are various classes of autoencoders. It is useful to distinguish between undercomplete and overcomplete models. Undercomplete models have a limited capacity that prevents them from achieving zero training loss. The most prominent example is provided by bottleneck autoencoders for which
\begin{equation}
\label{eq:f_AE}
f_\theta = f_{{\rm dec},\theta_2} \circ f_{{\rm enc},\theta_1},
\end{equation}
where the parameters~$\theta = (\theta_1,\theta_2)$ have been decomposed into parameters used in the encoder and decoder parts, respectively (see Figure~\ref{AE_notation} below), and~$\circ$ is the composition operator, namely
\[
h_2 \circ h_1(x) = h_2(h_1(x)).
\]

\begin{figure}[!ht]
	\centering
	\includegraphics[width=0.8\textwidth]{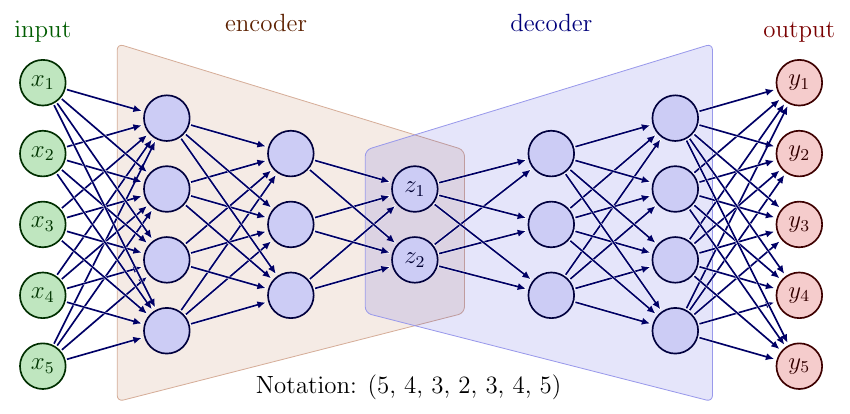}
	\caption{Schematic representation of a symmetric autoencoder. Blue neurons correspond to hidden layers, while the input and output layers are respectively in green and red. The encoder and decoder are respectively referred to as~(5, 4, 3, 2) and~(2, 3, 4, 5). Activation functions are hyperbolic tangents, except for the bottleneck and output layers, for which linear activation functions are considered in order not to restrict the range of values.}
	\label{AE_notation}
\end{figure}

The limitation in the capacity of the autoencoder arises from the fact that the encoding function~$f_{{\rm enc},\theta_1}$ has values in a latent space~$\cZ \subset \mathbb{R}^d$ of dimension~$d$ strictly smaller than the dimension~$D$ of the input/output space~$\cX$, usually much smaller in fact. Overcomplete models can on the other hand achieve zero training error. These models are of course useless as such since they would simply copy the input without extracting the salient features explaining the data at hand. This is why some regularization process should be considered to limit the capacity of the neural network. Regularization is however also useful for undercomplete models. Standard examples of regularization mechanisms include:
\begin{itemize}
	\item resorting to regularization strategies commonly used for neural networks in general, in particular early stopping, dropout, and standard weight decay to name a few options (see for instance Section~13.5 of Ref.~\citenum{Murphy22} and Chapter~7 of Ref.~\citenum{GBC16} for a more precise presentation of these options);
	\item sparse autoencoders where a penalization term is added to the loss function to prevent too many neurons to be active~\cite{Ng11};
	\item denoising autoencoders~\cite{VLLBM10}, where outputs should be predicted from inputs corrupted by some noise, which forces networks to learn the structure of the distribution of the data~\cite{AB14};
	\item contractive autoencoders, where the Jacobian of the encoder is penalized in order to ensure smoother variations of this function, and limit its sensitivity to small variations in the input~\cite{RVMGB11};
	\item variational autoencoders (VAEs)~\cite{KW14,KW19,GSVBS20} can also be interpreted as some regularized version of the usual autoencoder framework.
\end{itemize}
Other variations/extensions of autoencoders were also considered for more specific purposes, and are still studied, in particular for manifold learning, where the preservation of neighborhood relationships and/or geometric information in the dimensionality reduction process are important~\cite{JSGSS15,DMWM20,LYSP22}.

In this work, we will consider bottleneck autoencoders only, as these are the most relevant for molecular dynamics. Overcomplete models in particular are not directly interesting in terms of dimensionality reduction.

\paragraph{Autoencoders architectures.}
Autoencoders are often symmetric in their structures. In some cases, tied weights are being used, \emph{i.e.} the weights~$\theta_2$ are the transpose of the weights~$\theta_1$ when writing the prediction function as~\eqref{eq:f_AE}. This choice reproduces at the level of autoencoders the symmetric structure of PCA, where the decoder matrix is the transpose of the encoder matrix (see~\eqref{pca-objective} below and the discussion around this equation). In this case, a regularization on the encoder part only, as in contractive autoencoders, in fact also regularizes the decoder part. However, there is no particular motivation to use symmetric architectures, and various works, such as Ref.~\citenum{SMGY16}, studied the impact of an asymmetric architecture on the performance of the model.

Another important consideration, more specific to bottleneck autoencoders, is the choice of the bottleneck dimension. In PCA, the number of dimensions to retain usually corresponds to some ``knee'' in the plot of eigenvalues of the empirical covariance matrix. Similar plots can be performed for bottleneck autoencoders, for which the reconstruction error can be reported as a function of the bottleneck dimension (the remainder of the architecture being fixed). The optimal dimension should ideally coincide with the intrinsic dimension of the data set in manifold learning~\cite{DW19}, which can be quantified using the Frechet inception distance (see for instance the method described in Ref.~\citenum{IBE22} in the context of VAEs).

\paragraph{Autoencoders in molecular dynamics.}
In molecular dynamics, only bottleneck autoencoders are considered. In the remainder of this work, we simply refer to them as "autoencoders". They have been used for a few years now to find CVs, by using the encoder part as a CV for dimensionality reduction. In this situation, the decoder is not useful as such, but is still required to learn a good encoder.  

Autoencoders were first considered in a static framework, to extract information from samples distributed according to the Boltzmann--Gibbs distribution of the system. In fact, since the systems under consideration are metastable, it is often difficult to directly obtain a good sampling of the target distribution, and one should therefore turn to an iterative procedure alternating between a free energy biased sampling to obtain new data points, and a subsequent update of the collective variable given by the encoder. Such strategies were considered for applications in biophysics and materials science in unsupervised~\cite{CF18,CTF18,BGLS22,CLFFCSC22,BGSMNM22} and semi-supervised~\cite{BBMLSG23} settings, possibly also with variational autoencoders~\cite{rave,rave2019}. Autoencoders are also useful for coarse-graining, both to find the coarse-grained description and the effective forces.\cite{WGB19}

More recently, autoencoders were also considered to learn the dynamics of the systems under consideration. Besides methods for generic dynamical systems~\cite{CLKB19}, dedicated approaches were introduced for the dynamics considered in molecular simulation. Examples include time-lagged autoencoders~\cite{WN18,CSF19_TAE} and time-lagged variational autoencoders~\cite{HWSSHP18}; VAMPnets and their subsequent extensions and modifications~\cite{MPWN18,WMPN18,CSF19,SCF20}, which can be seen as versions of variational approaches to conformational dynamics~\cite{NKPMN14,KNKWKSN18} using a basis of functions represented by neural networks (in fact, only their encoder parts), instead of using linear functions as in tICA~\cite{MS94,NF13}; simultaneous reconstruction of the probability measure and learning of the committor function~\cite{FAB21,machine-guided-path-sampling}; approaches based on recurrent neural networks, in particular long short term memory networks~\cite{HS97}, used for instance in Ref.~\citenum{TKT20}, or mixture density networks~\cite{VZPK22}. All these methods differ in the details of the architecture and choice of loss function (e.g. time-lagged reconstruction loss for time-lagged autoencoders vs. autocorrelation of slow modes for VAMPnets and their extensions/modifications). 

It is often the case that no explicit regularization is introduced as the networks under consideration are rather shallow and narrow. Some works however explicitly introduce such regularization terms: for instance, an elastic net penalization is considered in Refs.~\citenum{KCL22,KL22}, and dropout plus weight decay in Ref.~\citenum{MPWN18}. In addition, it is sometimes suggested to add physically motivated terms to the loss function, for instance to favor the exploration of new metastable modes~\cite{KCL22,KL22}, or find minimum energy paths~\cite{RBGMM22}.

\begin{remark}
  \label{rmk:measure_is_known}
  In contrast to applications of ML for which the distribution of the data is typically unknown (see for instance Ref.~\citenum{AB14} where it is shown how to estimate the gradient of the log-density using autoencoders), an important difference in the use of autoencoders in MD is that the distribution of the training data generally has a density known up to a multiplicative constant (typically, the Boltzmann--Gibbs distribution). Although standard losses to train autoencoders do not make use of the expression of the density of the data, it might be possible to leverage this information in the MD context to find better encodings/decodings than in standard ML contexts. 
\end{remark}

\subsection{Interpretations of the loss function}
\label{sec:interpretation_loss}

We discuss in this section various reformulations and reinterpretations of the loss function~\eqref{eq:population_loss} for bottleneck autoencoders~\eqref{eq:f_AE} when the loss function is the square loss, and discuss in particular the relationship with principal curves/manifolds~\cite{HS89,Tib92}.  

\paragraph{Three viewpoints on the loss function.}
We consider an ideal setting where we minimize upon all measurable functions~$\fenc: \cX \to \cZ$ and~$\fdec: \cZ \to \cX$. We denote by~$\cFenc$ and~$\cFdec$ the sets of measurable functions from~$\cX$ to~$\cZ$ and from~$\cZ$ to~$\cX$, respectively; and by~$\cF$ the set of measurable functions obtained by composing functions of~$\cFenc$ with functions of~$\cFdec$:
\[
\cF = \left\{ f = \fdec \circ \fenc, \ \fenc \in \cFenc, \ \fdec \in \cFdec\right\}.
\]
Minimizing the reconstruction error over the set of functions in~$\cF$ can be then rewritten as
\begin{equation}
\label{eq:reconstruction_error_ideal}
\inf_{f \in \cF} \E\left[\left\|X-f(X)\right\|^2\right]. 
\end{equation}
Note that we do not consider a regularization term here, so that overfitting may occur in practice (for instance, even with~$\cZ$ of dimension~1, $\fdec$ can parametrize a space-filling curve).

As discussed in Ref.~\citenum{Gerber21} (which complements Ref.~\citenum{GW13} which was already hinting at autoencoders), the unsupervised least-square problem~\eqref{eq:reconstruction_error_ideal} can be thought of in various ways. In particular, there is some duality in the way the minimization over~$f \in \cF$ is performed, as one can decide to either
\begin{enumerate}[(i)]
	\item simultaneously minimize over~$\fenc$ and~$\fdec$, which is the standard way to proceed when training neural networks;
	\item minimize first over the encoder part, which allows to reformulate the minimization as the well-known problem of finding principal manifolds;
	\item minimize first over the decoder part, which is natural when thinking of the reconstruction error as some total variance to be decomposed using a conditioning on the values of the encoder.
\end{enumerate}
We discuss in the remainder of this section option~(ii), leaving the discussion of~(iii) to Section~\ref{sec:cond_exp}.

The chosen numerical approach has a natural impact on the topology of the networks which are considered: in situation~(i), encoders and decoders are treated on an equal footing, and it is therefore natural to consider them to be of a similar complexity; whereas options~(ii) and~(iii) suggest to consider asymmetric autoencoders. For instance, in option~(iii), the minimization over the decoder part, which is performed first, could be done more carefully, with more expressive networks in order to better approximate the optimal decoder for a given encoder. We numerically investigate this point in Section~\ref{sec:interpretation_numerics}.

\paragraph{Principal manifold reformulation.} We start by minimizing the reconstruction error~\eqref{eq:reconstruction_error_ideal} over the encoder part for a given decoder:
\begin{align}
\inf_{f \in \cF} \E\left[\left\|X-f(X)\right\|^2\right] & = \inf_{\fdec \in \cFdec} \left\{ \inf_{\fenc \in \cFenc} \E\left[\left\|X - \fdec \circ \fenc(X)\right\|^2\right] \right\} \notag \\
& =  \inf_{\fdec \in \cFdec} \E\left[\left\| X - \fdec \circ h^\star_{\fdec}(X)\right\|^2\right], \label{eq:encoder_first}
\end{align}
where the optimal encoder~$h^\star_{\fdec}:\cX \to \cZ$ for a given decoder~$\fdec:\cZ\to\cX$ is defined pointwise as
\[
h^\star_{\fdec}(x) \in \mathop{\mathrm{argmin}}_{z \in \cZ} \|x-\fdec(z) \|, 
\]
provided that this minimization problem admits a solution. When~$\fdec$ is smooth and has an invertible Jacobian, the principal manifold is then the set~$\fdec(\cZ) \subset \cX$. For any~$x \in \cX$, $h^\star_{\fdec}(x) \in \cZ$ gives the coordinates in the latent space of the projection of~$x$ on the principal manifold (the so-called projection index). 

The reformulation of the minimization problem in terms of the decoder function leads to a minimization problem similar to the one encountered when searching for principal curves and manifolds. These mathematical concepts generalize in some sense PCA to curves and surfaces rather than lines and hyperplanes, as already discussed in the work introducing principal curves~\cite{HS89}. The minimization problems associated with finding principal manifolds are in general difficult to solve as these manifolds correspond to saddle-points of the loss functional~\cite{DS96}. In practice, this means that it is possible to move away from a critical point without increasing the test loss. This leads to overfitting issues and prevents the use of traditional cross-validation procedures to tune regularization hyperparameters. 

In the small temperature regime, principal curves, associated with decoders starting from a one-dimensional space~$\cZ$, reduce to minimum energy paths when the measure under consideration is the Boltzmann--Gibbs measure; see the discussion in Ref.~\citenum{VEV09}. Another perspective on this statement is provided at the end of Section~\ref{sec:cond_exp}, using the interpretation relying on conditional expectations.

\subsection{Reformulating autoencoders with conditional expectations}
\label{sec:cond_exp}

We discuss here how to reformulate the training of autoencoders with conditional expectations, and provide alternative interpretations to the reconstruction error. As mentioned in Section~\ref{sec:cond_exp_others}, such reformulations can in fact be considered for other unsupervised machine learning methods such as principal component analysis and clustering.


In contrast to~\eqref{eq:encoder_first}, we minimize here the reconstruction error~\eqref{eq:reconstruction_error_ideal} by first minimizing over the decoder part for a given encoding function, as already considered in Ref.~\citenum{GTW09}. This approach is natural in molecular dynamics, as it is reminiscent of free energy computations~\cite{CP07,LRS10} where average quantities are computed for a fixed value of the collective variable~$\fenc$. From a mathematical viewpoint, it corresponds to introducing conditional averages associated with fixed values of the encoder. 

\paragraph{Rewriting the reconstruction error with conditional expectations.}
The loss function for unsupervised least-squares can be rewritten as
\begin{align}
\inf_{f \in \cF} \E\left[\left\|X-f(X)\right\|^2\right]
& = \inf_{\fenc \in \cFenc} \left\{ \inf_{\fdec \in \cFdec} \E\left[\left\|X - \fdec \circ \fenc(X)\right\|^2\right] \right\} \notag \\
& = \inf_{\fenc \in \cFenc} \E\left[\left\|X - g^\star_{\fenc} \circ \fenc(X)\right\|^2\right], \label{eq:reconstruction_error_with_Bayes_predictor}
\end{align}
where the ideal decoder~$g^\star_{\fenc}$ for a given encoder~$\fenc$ is the Bayes predictor associated with the least square regression problem (see Section~2.2.3 of Ref.~\citenum{Bach23}):
\begin{equation}
\label{eq:Bayes_predictor_dec}
g^\star_{\fenc}(z) = \E[\,X\,|\,\fenc(X) = z].
\end{equation}
Let us recall that, in all these expressions, expectations are taken with respect to the probability distribution~$\mu$ of the input data (which is not necessarily the Boltzmann--Gibbs distribution). Equations~\eqref{eq:reconstruction_error_with_Bayes_predictor}-\eqref{eq:Bayes_predictor_dec} show that the question of finding the best autoencoder can be reduced to finding the best encoding function, provided that one is able to compute good approximations of the conditional expectation.

\paragraph{Derivation of~\eqref{eq:Bayes_predictor_dec}.}
For completeness, let us recall here the derivation of~\eqref{eq:Bayes_predictor_dec}, as this derivation also allows to reformulate the minimization problem under consideration as some variance maximization for conditional expectations. We start by introducing the quantity~$g^\star_{\fenc} \circ \fenc(x)$ in the reconstruction error: for any~$\fdec \in \cFdec$ and~$\fenc \in \cFenc$,
\begin{align}
& \E\left[\left\|X - \fdec \circ \fenc(X)\right\|^2\right]
= \E\left[\left\| \left[X - g^\star_{\fenc} \circ \fenc(X)\right] + \left[g^\star_{\fenc} \circ \fenc(X) - \fdec \circ \fenc(X)\right]\right\|^2\right] \notag \\
& \qquad = \E\left[\left\|X - g^\star_{\fenc} \circ \fenc(X)\right\|^2\right] + \E\left[\left\| g^\star_{\fenc} \circ \fenc(X) - \fdec \circ \fenc(X)\right\|^2\right], \label{eq:decomposition_reconstruction_error}
\end{align}
where we used the following identity, obtained by conditioning on the values of the random variable~$Z=\fenc(X)$:
\[
\begin{aligned}
& \E\left[\left(X - g^\star_{\fenc} \circ \fenc(X)\right) \cdot \left(g^\star_{\fenc} \circ \fenc(X) - \fdec \circ \fenc(X)\right)\right] \\
& \qquad = \E\left[\left(\E\left[X \middle| Z\right]-g^\star_{\fenc}(Z)\right) \cdot \left(g^\star_{\fenc}(Z)-\fdec(Z)\right) \right] = 0,
\end{aligned}
\]
in view of the definition~\eqref{eq:Bayes_predictor_dec} of~$g^\star_{\fenc}$. It is clear from~\eqref{eq:decomposition_reconstruction_error} that the decoding function which minimizes the reconstruction error for a given encoder~$\fenc \in \cFenc$ is indeed~$g^\star_{\fenc}$, as defined by~\eqref{eq:Bayes_predictor_dec}.

\paragraph{Alternative interpretation of the reconstruction error.}
The reconstruction error~\eqref{eq:reconstruction_error_with_Bayes_predictor} can be reinterpreted in terms of variances. Indeed, on the one hand,
\begin{equation}
\label{eq:recon_error_by_condvar}
\begin{aligned}
\E\left[\left\|X - g^\star_{\fenc} \circ \fenc(X)\right\|^2\right]
& = \E\left[\left\|X - \E\left[X|\fenc(X)\right]\right\|^2\right] \\ & = \E\left[ \E\left( \left. \left\|X - \E\left[X|\fenc(X)\right]\right\|^2 \right| \fenc(X) \right) \right]  \\
& = \E\left[\mathrm{Var}(X|\fenc(X))\right].
\end{aligned}
\end{equation}
On the other hand,
\begin{equation}
\begin{aligned}
\label{eq:total_variance_formula}
\E\left[\left\|X - g^\star_{\fenc} \circ \fenc(X)\right\|^2\right]
& = \E\left[\left\|X - \E\left[X|\fenc(X)\right]\right\|^2\right]
\\ & = \E\left(\|X\|^2\right) - \E\left( \E\left[X|\fenc(X)\right]^2 \right)  \\
& = \mathrm{Var}(X) - \mathrm{Var}\left[ \mathbb{E}(X|\fenc(X)) \right]. 
\end{aligned}
\end{equation}
These two equalities yield the well-known formula for the total variance decomposition, namely~$\mathrm{Var}(X) = \E\left[\mathrm{Var}(X|\fenc(X))\right]+ \mathrm{Var}\left[ \mathbb{E}(X|\fenc(X)) \right]$.

A consequence of~\eqref{eq:total_variance_formula} is that the minimization problem~\eqref{eq:reconstruction_error_with_Bayes_predictor} can be reformulated as the following equivalent maximization problem:
\begin{equation}
\label{eq:interclass_dispersion}
\sup_{\fenc \in \cFenc} \mathrm{Var}\left[ \mathbb{E}(X|\fenc(X)) \right].
\end{equation}
In words, this reformulation translates the equivalence between (the "classes" referring here to the level sets of~$\fenc$) 
\begin{itemize}
\item minimizing the intraclass dispersion~\eqref{eq:reconstruction_error_with_Bayes_predictor}: the distribution of configurations~$x \in \cX$ for a fixed value~$z$ of~$\fenc$ should concentrate around the mean value~$g^\star_{\fenc}(z)$ by having a variance as small as possible;
\item maximizing the interclass dispersion~\eqref{eq:interclass_dispersion}: the values of the conditional averages of~$X$ for fixed values of~$\fenc$ should be as spread out as possible over the range of~$\fenc$.
\end{itemize}

\paragraph{Formal characterization of the optimal encoding function.} A key equality to establish~\eqref{eq:total_variance_formula}, namely 
\[
\E\left[\left\|X - g^\star_{\fenc} \circ \fenc(X)\right\|^2\right] = \E\left[\left\|X\right\|^2\right] - \E\left[\left\|g^\star_{\fenc} \circ \fenc(X)\right\|^2\right],
\]
shows that the minimization of the reconstruction error is equivalent to the maximization of the second moment of the conditional expectation:
\begin{equation}
\label{eq:equivalent_maximization_on_fenc}
\sup_{\fenc \in \cFenc} \E\left[\left\|g^\star_{\fenc} \circ \fenc(X)\right\|^2\right]\,.
\end{equation}
This alternative viewpoint allows to characterize the optimal encoding function~$\fenc$ by some orthogonality condition similar to the self-consistency condition of principal curves, see Section~2 of Ref.~\citenum{GW13}. In fact, it can be formally shown that critical points of~\eqref{eq:equivalent_maximization_on_fenc} satisfy
\begin{equation}
\label{eq:formal_EL_condition_optimal_fenc}
\forall j \in \{1,\dots,d\}, \quad \forall x \in \mathrm{Supp}(\mu), \qquad \left[ x - g^\star_{\fenc}(\fenc(x)) \right]^\top \partial_{z_j} g^\star_{\fenc}(\fenc(x)) = 0,
\end{equation}
where~$\mathrm{Supp}(\mu)$ is the support of the probability measure~$\mu$. The derivation of this condition, which can be read in Appendix~\ref{sec:derivation_formal_EL_condition_optimal_fenc}, can be seen as a variation of derivations of optimality conditions for principal curves, as written already in Ref.~\citenum{HS89}; see also Ref.~\citenum{GW13} where~\eqref{eq:formal_EL_condition_optimal_fenc} is used to construct a new objective function to minimize in order to find~$\fenc$. From a technical viewpoint, an originality of our approach to obtain~\eqref{eq:formal_EL_condition_optimal_fenc} is that we rely on the co-area formula~\cite{EG92,AFP00} together with the use of weak derivatives; see Appendix~\ref{sec:derivation_formal_EL_condition_optimal_fenc}.

An interesting implication of~\eqref{eq:formal_EL_condition_optimal_fenc} is that the intersection of~$\mathrm{Supp}(\mu)$ and the submanifold
\begin{equation}
\label{eq:Sigma_z}
\Sigma_z = \fenc^{-1}\{z\} = \left\{ x \in \cX \, \middle| \, \fenc(x) = z \right\}
\end{equation}
is in fact included in the $(D-d)$-dimensional hyperplane containing the point~$g^\star_{\fenc}(z)$ and orthogonal to the vectors~$\partial_{z_1} g^\star_{\fenc}(z),\dots,\partial_{z_d} g^\star_{\fenc}(z)$ (recalling that~$\cX$ and~$\cZ$ have dimensions~$D$ and~$d$, respectively). As these hyperplanes generally have a non-empty intersection, finding a regular function~$\fenc$ which satisfies~\eqref{eq:formal_EL_condition_optimal_fenc} is only possible for distributions~$\mu$ which have a support sufficiently concentrated around the principal manifold. The condition~\eqref{eq:formal_EL_condition_optimal_fenc} can be generalized when several principal manifolds are considered, see Section~\ref{sec:multiple_path_2_state} below. The issue of having hyperplanes intersecting can be seen as the counterpart in the context we consider here of the concept of ambiguity points for principal curves~\cite{HS89}. Let us mention that a condition similar to~\eqref{eq:formal_EL_condition_optimal_fenc} is derived in Ref.~\citenum{VEV09} to reinterpret the finite-temperature string method using principal curves, followed by a discussion on hyperplanes corresponding to isosurfaces of fixed values of the assignment function (the counterpart of the encoder function in our context).

\paragraph{Minimum energy paths and the small temperature limit.}
We finally study the limit~$\beta\to+\infty$ in~\eqref{eq:formal_EL_condition_optimal_fenc}, when the probability measure~$\mu$ under consideration is the Boltzmann--Gibbs measure with a density proportional to~$\rme^{-\beta V(x)}$, with~$V$ the potential energy function.

Transition paths are one-dimensional curves. It would be possible to work in a latent space~$\cZ$ of dimension~$d \geq 2$, and to extract some one dimensional path relating the initial and end configurations in this space. Here, for simplicity, we restrict ourselves to a one-dimensional latent space~$\cZ$, \emph{i.e.}~$d=1$.

We consider two local minima~$x_A$ and~$x_B$ of the potential energy function~$V$, located respectively on~$\Sigma_{z_A}$ and~$\Sigma_{z_B}$ with $z_A = \fenc(x_A)$ and~$z_B = \fenc(x_B)$ (assuming~$z_A \leq z_B$ without loss of generality). Similarly to the discussion in Ref.~\citenum{VEV09} for principal curves, it can formally be shown in the low temperature limit that the decoder path~$\{ g^\star_{\fenc}(z) \}_{z \in [z_A,z_B]}$ converges to a minimum energy path (MEP). A MEP~$\{ \gamma(z) \}_{z \in [z_A,z_B]}$ between~$x_A$ and~$x_B$ is characterized by the boundary conditions~$\gamma(z_A) = x_A$ and~$\gamma(z_B) = x_B$, and the fact that~$\gamma'(z)$ is collinear to~$\nabla V(\gamma(z))$ for~$z \in (z_A,z_B)$ (except at critical points along the path); see for instance Ref.~\citenum{HJJ02} as well as Ref.~\citenum{LCO22} for a rigorous mathematical formulation.

In view of~\eqref{eq:Bayes_predictor_dec} and~\eqref{eq:analytical_formula_g_star_fenc} in Appendix~\ref{sec:derivation_formal_EL_condition_optimal_fenc}, the conditional expectation concentrates on a minimum of~$V$ on~$\Sigma_z$ the limit~$\beta \to +\infty$. More precisely, recalling the definition~\eqref{eq:Sigma_z} of~$\Sigma_z$,
\begin{equation}
\label{eq:g_star_infty}
\forall z \in (z_A,z_B), \qquad g^\star_{\fenc}(z) \xrightarrow[\beta \to +\infty]{} g^\infty_{\fenc}(z) \in \mathop{\mathrm{argmin}}_{x \in \Sigma_z} V(x).  
\end{equation}
We assume in the sequel that the minimum in~\eqref{eq:g_star_infty} is well defined (\emph{i.e.} it exists and is unique). Then, \eqref{eq:g_star_infty} implies that
\begin{equation}
\label{eq:t_z_nabla}
\forall z \in (z_A,z_B), \quad \forall t_z \in T\Sigma_{z}\left(g^\infty_{\fenc}(z)\right), \qquad t_z^\top \nabla V\left(g^\infty_{\fenc}(z)\right) = 0,
\end{equation}
where we denote by~$T\Sigma_z(x)$ the tangent plane to~$\Sigma_z$ at~$x \in \Sigma_z$. 

Moreover, for a perfectly converged autoencoder model, we obtain, by passing in~\eqref{eq:formal_EL_condition_optimal_fenc} to the limit~$x \rightarrow g^\star_{\fenc}(z)$ for~$x \in \Sigma_z$: 
\begin{equation}
\label{eq:infinitesimal_condition_g_star_prime}
\forall z \in \cZ, \quad \forall\, t_z \in T\Sigma_z\left(g^\star_{\fenc}(z)\right), \qquad t_z^\top \left(g^\star_{\fenc}\right)'(z) = 0\,,
\end{equation}
where $(g^\star_{\fenc})'(z)$ is the derivative of the decoder. By passing to the limit~$\beta \to +\infty$ in the latter condition (assuming that the derivatives in~$z$ and the limit~$\beta\to+\infty$ can be exchanged), it follows that
\[
\forall z \in (z_A,z_B), \quad \forall t_z \in T\Sigma_{z}\left(g^\infty_{\fenc}(z)\right), \qquad t_z^\top  \left( g^\infty_{\fenc}\right)'(z) = 0.
\]
The comparison between the latter condition and~\eqref{eq:t_z_nabla} shows that~$\nabla V(g^\infty_{\fenc}(z))$ is collinear to~$(g^\infty_{\fenc})'(z)$ (provided~~$\nabla V(g^\infty_{\fenc}(z)) \neq 0$), which shows that the curve~$\{ g^\infty_{\fenc}(z) \}_{s \in [z_A,z_B]}$ is a MEP. 

\begin{remark}
  The analysis in this section shows that the autoencoder learns the principal manifold of the underlying data distribution, which collapses to the minimal energy path in the small temperature limit. Many other numerical approaches, such as the string method~\cite{VEV09}, also aim at computing the minimum energy path. Besides, let us mention that this connection between autoencoders and manifold learning methods can also be leveraged to construct generative methods to sample new configurations from the probability measure of the data. This can for instance be done with tools from manifold learning to propose stochastic perturbations along the tangent space of the learned manifold,\cite{RBDV12,BAS12} as discussed in Section 9.2 of Ref.~\citenum{BCV13}.
\end{remark}

\subsection{From autoencoders to clustering and PCA}
\label{sec:cond_exp_others}

We have seen that the unsupervised least-square reconstruction error can be reformulated in terms of conditional expectations as follows (see~\eqref{eq:reconstruction_error_with_Bayes_predictor} and~\eqref{eq:Bayes_predictor_dec}):
\begin{equation}
\label{dim-reduction-general-xi}
\inf_{\fenc \in \cFenc} \E\left[ \Big\|X-\E\left(X\,\middle|\,\fenc(X)\right)\Big\|^2\right],
\end{equation}
where the random variable~$X$ is distributed according to~$\mu$. We discuss in this section how two other unsupervised learning methods, $K$-means clustering and PCA, can be understood in terms of a similar minimization problem involving conditional expectations. Such interpretations have already been discussed in the literature (see for instance Section~4 of Ref.~\citenum{Fischer14}), but are of interest for discussions in later sections, in particular Section~\ref{sec:multiple_path_2_state} where we combine autoencoders with ideas from clustering.

\paragraph{$K$-means clustering.}
For simplicity of exposition and to be closer to traditional expositions in machine learning textbooks, we discuss here a reformulation of clustering methods for sample data, and not at the level of probability distributions (corresponding to so-called population losses, as we have considered until now). We therefore introduce a sample data set~$\cD_{\Nd} = \{ x^1,\dots,x^{\Nd} \} \subset \cX$ of cardinality~$\Nd$, and denote by~$\widehat{\mu}_{\Nd}$ the empirical probability measure
\[
\widehat{\mu}_{\Nd} = \frac{1}{\Nd} \sum_{n=1}^{\Nd} \delta_{x^n}.
\]
The data is partitioned into~$K \leq \Nd$ classes indexed by an integer~$1 \leq k \leq K$, using an assignment function~$\fass:\cD_{\Nd} \to \{1,\dots,K\}$. For a given assignment function, we denote by 
\[
\sC_k = \left\{x \in \cD_{\Nd} \, \middle| \, \fass(x) = k \right\} = \fass^{-1}\{k\}
\]
the various clusters identified by the clustering method.  The cardinality of~$\sC_k$ is denoted by~$|\sC_k|$. With this notation, the objective function of the~$K$-means algorithm reads (see for instance Section~21.3 of Ref.~\citenum{Murphy22} or Section~22.2 of Ref.~\citenum{SSBD14})
\begin{equation}
\label{k-means-loss}
\sum_{k=1}^K \sum_{i \in \sC_k} \left\| x^i - \frac{1}{|\sC_k|} \sum_{j \in \sC_k} x^j \right\|^2.
\end{equation}
It is easy to see that~\eqref{k-means-loss} is indeed of the form~\eqref{dim-reduction-general-xi} with expectations taken with respect to~$\widehat{\mu}_{\Nd}$ and~$\fenc$ replaced by~$\fass$ since 
\begin{equation*}
\forall k \in \{1,\dots,K\}, \qquad \E_{\widehat{\mu}_{\Nd}}(X \, | \, \fass(X)=k) = \frac{1}{|\sC_k|} \sum_{j \in \sC_k} x^j.
\end{equation*}
The latter two equalities show that autoencoders whose encoders have values in the discrete space~$\{1,\dots,K\}$ aim at solving the same minimization problem as~$K$-means, the assignment function being the encoder function.

\paragraph{Principal component analysis.}
We denote by~$K \in \{1,\dots,D-1\}$ the number of dimensions kept in the dimensionality reduction process performed by PCA. For a given sample data set~$\cD_{\Nd}$, we denote by
\[
\overline{x}_{\Nd} = \frac{1}{\Nd}\sum_{n=1}^{\Nd} x^n \in \mathbb{R}^D
\]
the empirical mean of the data, and by 
\[
\widehat{\Sigma}_{\Nd} = \frac{1}{\Nd} \sum_{n=1}^{\Nd} \left(x^n-\overline{x}_{\Nd}\right) \left(x^n-\overline{x}_{\Nd}\right)^\top \in \mathbb{R}^{D \times D} 
\]
the empirical covariance matrix. The PCA algorithm performs a dimensionality reduction by computing normalized eigenvectors~$\widehat{u}_1, \cdots, \widehat{u}_K$ associated with the~$K$ largest eigenvalues of the empirical covariance matrix, and proposing the following reconstruction of a test point~$x \in \mathbb{R}^D$:
\begin{equation}
\label{eq:projection_reconstruction_PCA}
\widetilde{x} = \overline{x}_{\Nd} + \sum_{k=1}^K \left[\widehat{u}_k^\top \left(x-\overline{x}_{\Nd}\right)\right]\widehat{u}_k = \overline{x}_{\Nd} + \widehat{U}\widehat{U}^\top \left(x-\overline{x}_{\Nd}\right),
\end{equation}
where the matrix~$\widehat{U} \in \mathbb{R}^{D \times K}$ has columns~$\widehat{u}_1, \cdots, \widehat{u}_K$. The maximal variance formulation of PCA is equivalent to a minimum reconstruction error formulation. In fact, as made precise in Section~23.1 of Ref.~\citenum{SSBD14} for instance, the above matrix~$\widehat{U}$ is a solution to the minimization problem
\begin{equation}
\label{pca-objective}
\min_{\substack{W\in \mathbb{R}^{D\times K} \\ W^\top W = \mathrm{I}_K}} \E_{\widehat{\mu}_{\Nd}}\left( \left\|X - WW^\top X\right\|^2 \right).
\end{equation}
This formulation is the one relevant to write PCA in the autoencoder framework~\cite{BK88,BH89}.

PCA can be reformulated as the minimization task~\eqref{dim-reduction-general-xi} in situations when the affine transformation in~\eqref{eq:projection_reconstruction_PCA} can be interpreted in terms of some conditioning. This is the case when~$X$ is a Gaussian random variable and the empirical loss in~\eqref{pca-objective} is replaced by a population loss. Without loss of generality, to simplify the notation, the Gaussian distribution can be assumed to be centered (upon subtracting the mean of the distribution to the random variable under consideration). To state the result, we introduce the set~$\cFred$ of linear orthogonal projections:
\begin{equation}
\label{function-space-pca}
\cFred = \left\{\fred:\left\{\begin{aligned}\mathbb{R}^D& \rightarrow \mathbb{R}^K \\ x & \mapsto W^\top x \end{aligned}\right. \ \mbox{for some}~W\in \mathbb{R}^{D\times K} \mbox{ such that}~ W^\top W = \mathrm{I}_K\right\}. 
\end{equation}

\begin{lemma}
  \label{lem:cond_exp_PCA}
  Assume that~$X$ is a Gaussian random variable with mean~0. Then, the PCA minimization problem~\eqref{pca-objective} is equivalent to
  \begin{equation}
    \label{eq:dim-reduction-pca}
    \min_{\fred \in \cFred} \E\left[ \Big\| X-\E\left(X\,\middle|\,\fred(X)\right)\Big\|^2\right].
  \end{equation}
\end{lemma}

The proof of this statement can be read in Appendix~\ref{sec:proof_lem:cond_exp_PCA}. PCA can thus be seen as an autoencoder method applied to a Gaussian distribution, with linear activation functions. Conversely, one can thus argue that autoencoder methods are nonlinear generalizations of PCA, which was exactly the historical motivation for their introduction.~\cite{BK88,BH89}

\subsection{Improving the dataset in order to better describe the transition}
\label{sec:changing_ref_measure}

As discussed at the end of Section~\ref{sec:cond_exp}, the decoder parametrizes the MEP in the small temperature regime, and thus gives information on the transition between local minima. We discuss here ways to improve the ability of autoencoders to describe such transitions.

The loss function of autoencoders focuses the reconstruction effort on more likely regions under the probability measure~$\mu$ of the data points. For physical systems described by the Boltzmann--Gibbs measure, this means that the reconstruction should be accurate in metastable states around local minima of the potential energy function. In contrast, transition states between metastable regions are typically only scarcely sampled, and hence a large reconstruction error on such configurations may not have a strong impact on the overall reconstruction error. However, from a physical viewpoint, transition states are very important to understand how the system can undergo a transition from one metastable state to another.

One idea to better take into account transition states is to change the reference measure from the Boltzmann--Gibbs measure to a measure putting more mass on regions between metastable states, such as the reactive trajectory measure~\cite{EVE06,LN15} (which is the distribution of configurations sampled by portions of trajectories switching from one metastable state to another). Another idea, discussed below in Section~\ref{sec:multiple_path_2_state}, is to use extra physical information encoded via additional terms in the loss functions (as considered in Ref.~\citenum{RBGMM22} for instance). A final option, not explored here, would be to reweight configurations in the data set by some factor motivated by an importance sampling approach, relying on the known expression for the distribution of the data (recall Remark~\ref{rmk:measure_is_known}). 

More specifically, we consider a family of probability measures obtained by a convex combination of the usual Boltzmann--Gibbs measure and the reactive trajectory distribution. The latter one is sampled by running adaptive multilevel splitting (AMS)~\cite{CG07,LL19,Pigeon_al_2023} to sample the transition paths for the overdamped Langevin dynamics described in Section~\ref{sec:interpretation_numerics}. More precisely, the probability distribution~$\mu$ is obtained by considering a fraction~$\lambda \in [0,1]$ of the Boltzmann--Gibbs measure, and a fraction~$1-\lambda$ of the reaction path measure. Denoting by~$\Nd$ the number of data points~$\{x^n\}_{1 \leq n \leq \Nd}$ distributed according to the Boltzmann--Gibbs measure, and by~$\wNd$ the number of data points~$\{\widetilde{x}^n\}_{1 \leq n \leq \wNd}$ distributed according to the reactive trajectory distribution, the associated empirical training loss is then (recall \eqref{eq:f_AE})
\begin{equation}
  \label{eq:sLh_lambda}
  \begin{aligned}
  \sLh_\lambda(\theta) = & \frac{\lambda}{\Nd} \sum_{n=1}^{\Nd} \left\|x^n - f_\theta(x^n)\right\|^2 \\ &  + \frac{1 - \lambda}{\wNd} \sum_{n=1}^{\wNd} \left\|\widetilde{x}^n - f_\theta(\widetilde{x}^n)\right\|^2.
  \end{aligned}
\end{equation}
The optimal value of~$\lambda$ should be determined by cross-validation, possibly on an objective function different from the above reconstruction loss. More precisely, for the numerical results reported in Section~\ref{sec:numerics_modifying_proba_dist}, we measure the quality of conditional expectations by some alignment criterion. For simplicity, we present the idea for a one-dimensional latent space~$\cZ$, \emph{i.e.} $d=1$; but our analysis can be extended to more general settings. Since $\nabla \fenc (g^\star_{\fenc}(z))$ is orthogonal to~$\Sigma_z$ at~$g^\star_{\fenc}(z) \in \Sigma_z$ (as~$\Sigma_z$ is a level set of~$\fenc$), and recalling~\eqref{eq:infinitesimal_condition_g_star_prime}, we conclude that the cosine of the angle between vectors $(g^\star_{\fenc})'(z)$ and $\nabla \fenc (g^\star_{\fenc}(z))$, namely 
\begin{equation}
\label{eq:cos_angle_alignement}
\rho(z) = \frac{\nabla \fenc (g^\star_{\fenc}(z))^\top (g^\star_{\fenc})'(z)}{\left\|\nabla \fenc (g^\star_{\fenc}(z))\right\| \left\| (g^\star_{\fenc})'(z)\right\|}
\end{equation}
should ideally be~$1$. Appropriate values for~$\lambda$ in~$\sLh_\lambda$ can thus be determined by requiring~$\rho$ to be close to~1, in addition to the reconstruction error to be small. 

\subsection{Better loss functions: multiple pathways and physically informed regularizations}
\label{sec:multiple_path_2_state}

In a situation where multiple transition paths link two metastable states, the autoencoder may fail to properly represent the system in the transition region between local minima for a one-dimensional latent space~$\cZ$, as it constructs only a single curve for the conditional expectations. An idea to address this issue is to consider multiple decoders associated with a common encoder, and to choose for a given configuration the decoder which best reconstructs the state through some assignment function reminiscent of the one considered for clustering (recall Section~\ref{sec:cond_exp_others}). We present here such a strategy. Related strategies were suggested for principal curves (see Section~II.B of Ref.~\citenum{VEV09} and Section~3.5 of Ref.~\citenum{Fischer14} for instance), and actually tested in Ref.~\citenum{KS17}.

\paragraph{Loss function for multiple decoders.} In order to take into account the presence of multiple decoders, the reconstruction loss in~\eqref{eq:population_loss} for a given encoder function~$\fenc$ and~$K$ decoder functions~$f_{{\rm dec},k}$ is modified as
\[
\E\left[\underset{k \in 1, .., K}{\min}\big\|x - f_{{\rm dec},k} \circ \fenc(x)\big\|^2\right].
\]
In practice, as in~\eqref{eq:f_AE}, the encoder~$\fenc$ is represented by a neural network~$f_{\mathrm{enc},\theta_1}$, where~$\theta_1$ gathers all the parameters of the encoder (weights and biases). This encoder is shared by the~$K$ decoders~$f_{{\rm dec},k}$, which are represented by neural networks~$f_{{\rm dec},\theta_{k+1}}$ for~$1 \leq k \leq K$, where the parameters for the~$K$ decoders (weights and biases) are denoted respectively by~$\theta_2,\dots,\theta_{K+1}$. The associated training loss for a given dataset~$\{x^n\}_{1 \leq n \leq \Nd}$ then reads
\begin{equation}
\label{eq:loss_multiple_dec}
\sLh_K(\theta) = \frac{1}{\Nd} \sum_{n=1}^{\Nd} \underset{k \in 1, .., K}{\min} \left\|x^n - f_{{\rm dec},\theta_{k+1}} \circ f_{\mathrm{enc},\theta_1}(x^n)\right\|^2.
\end{equation}
A situation where each decoder accounts for a transition path corresponds to some local minimum of the loss function. Intuitively, one can also argue that, when the number of decoders employed is equal to or less than the number of transition paths discovered by data points, the solution where each decoder accounts for a different pathway gives a smaller loss compared to the solution where multiple decoders correspond to the same pathway. When there are more decoders than transition paths, some decoders are redundant (as numerically demonstrated in Figure~\ref{fig:double_dec_mullerbrown} below in Section~\ref{sec:numerics_multiple_paths}).  

\begin{remark}
  Although we did not explore this option here, it would be possible to replace the hard assignment encoded by taking the minimum over~$k$ by a probabilistic assignment arising from a softmin function, similarly to what is considered in the probabilistic approach to classification problems.
\end{remark}

One issue with the loss~\eqref{eq:loss_multiple_dec} is that there is no mechanism to enforce that the various decoders parametrize a transition between the two metastable states under consideration (see Figure~\ref{fig:double_dec_circle} below in Section~\ref{sec:numerics_multiple_paths}). Additional, physically motivated regularization terms are needed to better describe transition paths. We present two possibilities in the remainder of this section. The relative weights of the additional terms in the loss function should ideally be fixed by cross-validation.


\paragraph{Penalizing the gradient of the encoder.}
A first option for additional terms in~\eqref{eq:loss_multiple_dec} is 
\begin{align} \sLh_{K,\lambda_0,\lambda_1,\lambda_2}(\theta) = \lambda_0 \sLh_{K}(\theta)
+ \lambda_1 
\left(\widehat{\mathrm{Var}}_{\Nd}(f_{\rm{enc},\theta_1}) - 1 \right)^2 + \frac{\lambda_2}{\Nd} \sum_{n=1}^{\Nd} \left| \nabla f_{\rm{enc},\theta_1}(x^n)\right|^2, \label{eq:loss_multiple_dec_constrained}
\end{align}
where $\lambda_0, \lambda_1, \lambda_2 \geq 0$ allow to tune the extra loss terms on the encoder
and 
\begin{equation}
\widehat{\mathrm{Var}}_{\Nd} (f)=
\frac{1}{\Nd}\sum_{n=1}^{\Nd} |f(x^n)|^2
- \left|\frac{1}{\Nd}\sum_{n=1}^{\Nd} f(x^n)\right|^2
\label{eq:empirical-var}
\end{equation}
denotes the empirical variance of a function $f$. Here we assume a one-dimensional latent space so that the encoder $f_{\rm{enc},\theta_1}$ is scalar-valued. Note that the extra terms in the loss function are insensitive to shifts in the values of the encoder function. The penalization of the variance (term proportional to~$\lambda_1$) is simply a normalization term, further discussed below. The last term, proportional to~$\lambda_2$, is the same as the one considered in the loss function of contractive autoencoders.\cite{RVMGB11} 

Let us now motivate the extra terms which are considered in~\eqref{eq:empirical-var}, by focusing on the particular case~$\lambda_0 = 0$. In the limit~$\Nd \to +\infty$ and~$\lambda_1 \to +\infty$, one recovers the loss function characterizing the eigenfunction associated with the first non-zero eigenvalue of the generator of the overdamped Langevin dynamics. Indeed, the optimization problem converges in the limit~$\Nd \to +\infty$ to the minimization of
\[
\lambda_1 \left[ \int_\cX \left(\fenc(x)-\int_\cX \fenc \, d\mu \right)^2 \, \mu(dx) - 1 \right]^2 + \lambda_2 \int_\cX \left| \nabla f_\mathrm{enc}(x)\right|^2 \mu(dx).
\]
When~$\lambda_1 \to +\infty$, the first term in the expression above enforces a normalization constraint on~$\fenc$. The minimization of~$\sLh_{K,0,\lambda_1,\lambda_2}$ in the limits~$\Nd \to +\infty$ and~$\lambda_1 \to +\infty$ can then be reformulated as
\[
\begin{aligned}
\inf_{\fenc \in \cFenc} \Bigg\{ & \int_\cX \left| \nabla \fenc(x)\right|^2 \mu(dx) \  \\ & \Bigg|\ \int_\cX \left|\fenc(x)-\int_\cX \fenc \, d\mu \right|^2 \mu(dx) = 1 \Bigg\},
\end{aligned}
\]
which characterizes (up to constant shifts) the eigenfunction associated with the first non-zero eigenvalue of the generator of the overdamped Langevin dynamics; see Ref.~\citenum{ZLS22} for a more detailed derivation and a numerical strategy based on neural networks. The reason why the eigenfunctions are relevant is that they are typically almost constant on the metastable states, and index the transition from one mode to another~\cite{Bovier2005}.

\paragraph{Normalizing the values of the encoder.}
In addition to the extra terms in~\eqref{eq:loss_multiple_dec_constrained}, one can also fix the values of the encoder at specific points, such as local minima or centers of metastable modes. For instance, when there are two metastable modes~$A$ and~$B$ with respective centers~$x_A$ and~$x_B$, one can add penalization terms $\fenc(x_A)^2$ and~$[\fenc(x_B)-1]^2$ to force the encoder to have values close to~0 in the left mode and close to~1 in the right mode, respectively. This requires the encoder to describe the transition from values~0 to~1. In essence, this can be seen as some form of semi-supervised approach where specific points are identified as belonging to the left or right mode, and the values of the encoder are softly enforced there. Of course, this approach relies on some prior information on the metastable states. 

In view of the above discussion, the loss that we will use in the experiments reported in Section~\ref{sec:numerics_multiple_paths} finally reads
\begin{align}
\sLh_{K,\lambda_0,\lambda_1,\lambda_2,\lambda_3,\lambda_4}(\theta)
=
\lambda_0 \sLh_{K}(\theta)
& + \lambda_1 
\left(\widehat{\mathrm{Var}}_{\Nd}(f_{\rm{enc},\theta_1}) - 1 \right)^2  \notag \\
& + \lambda_2 \frac{1}{\Nd} \sum_{n=1}^{\Nd} \left| \nabla f_{\rm{enc},\theta_1}(x^n)\right|^2 \notag \\
& + \lambda_3  \frac{1}{N_\mathrm{pen}} \sum_{n=1}^{N_\mathrm{pen}} \sum_{k = 1}^K \big(\wx^n - f_{{\rm dec},\theta_{k+1}} \circ f_{\mathrm{enc},\theta_1}(\wx^n)\big)^2 \notag \\
& + \lambda_4 \frac{1}{N_\mathrm{pen}} \sum_{n=1}^{N_\mathrm{pen}} \left(\widetilde{y}^n - f_{\rm{enc},\theta_1}(\wx^n) \right)^2,
\label{eq:loss_multiple_dec_constrained_2}
\end{align} 
where $\lambda_0$, $\lambda_1$, $\lambda_2$, $\lambda_3$ and $\lambda_4$ are nonnegative numbers which allow to tune the strength of the constraints, and~$\{\wx^n\}_{1 \leq n \leq N_{\rm pen}}$ is the set of points at which the value of the encoder is softly constrained. These points are typically not points of the training data set, but extra points where the value of the encoder is a priori known. The term in factor of~$\lambda_3$ ensures that the reconstruction is correct on the set~$\{\wx^n\}_{1 \leq n \leq N_{\rm pen}}$ for all~$K$ decoders, while the term in factor of~$\lambda_4$ forces the encoder values to be close to~$\widetilde{y}^n$ for the configuration~$\wx^n$. The chosen values of~$\widetilde{y}^n$ depend on the probability mode to which the configuration belongs. When there are two modes, $\widetilde{y}^n \in \{0,1\}$. In practice, for the numerical results reported in Section~\ref{sec:numerics_multiple_paths}, only two points ($N_\mathrm{pen}=2$) corresponding to the two local minima on the considered potential are penalized, with $\widetilde{y}^1 = 0$ and $\widetilde{y}^2=1$. 

Let us emphasize that the variance term in factor of~$\lambda_1$ in the above loss function is incompatible with the reconstruction term in factor of~$\lambda_4$, as both terms determine in some sense the range of values taken by the encoder, but not in the same manner. Therefore, one of the parameters~$\lambda_1$ or~$\lambda_4$ needs to be set to~0. Overall, the learning problem is semi-supervised as it is the concatenation of unsupervised learning on~$\{x^n\}_{1 \leq n \leq \Nd}$ and supervised learning on~$\{\wx^n\}_{1 \leq n \leq N_{\rm pen}}$.

\subsection{Parametrizing eigenfunctions of the transfer operator for systems with multiple metastable states}
\label{sec:parametrzing_eigenfunc_multiple_states}

In various situations, it can be beneficial to consider transfer operators in order to incorporate some information on the dynamics. Transfer operators can be more convenient than generators since the loss function of the transfer operator framework does not require knowing the equation of the underlying dynamics. This is advantageous when applied to MD applications where data comes from sampling of an underlying dynamics (via an MD package) whose potential function often involves many parameters and is too complicated to be written down explicitly. 

Our objective here is to learn an encoder that is optimized not only for configuration reconstruction but also for parametrizing the leading eigenfunctions (corresponding to the largest non-trivial eigenvalues) of the transfer operator of the underlying dynamics. This is achieved by adding to the standard autoencoder architecture another neural network, which takes the output of the encoder as input and provides the values of eigenfunctions as output; see Figure~\ref{fig-reg-ae} for an illustration of the architecture where the regularization part involves the first eigenfunction only. The extension to regularization involving multiple eigenfunctions is straightforward.

\begin{figure}[ht!]
  \centering
  \includegraphics[width=0.5\textwidth]{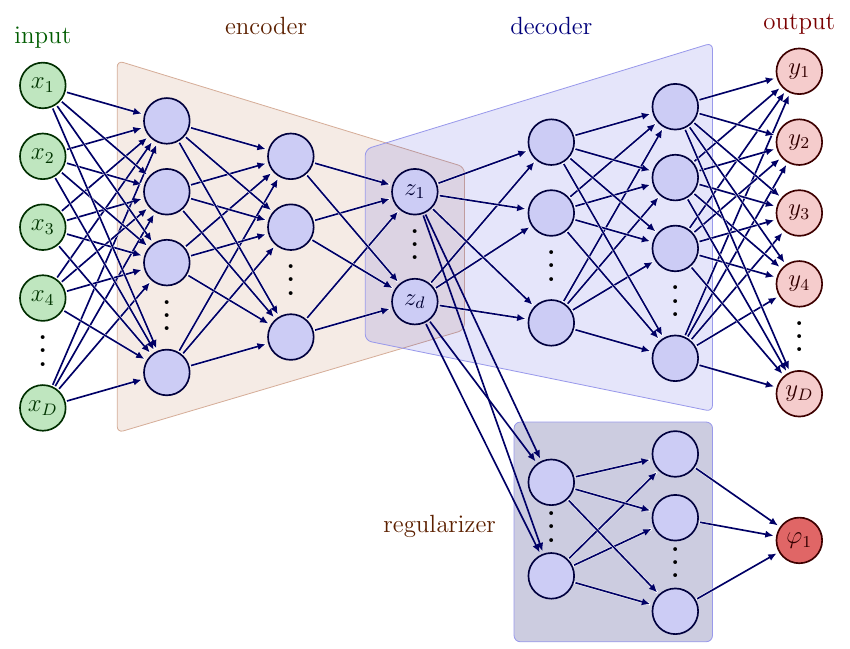}
  \caption{Illustration of an autoencoder architecture with a regularization component to represent the leading eigenfunction.}
  \label{fig-reg-ae}
\end{figure}

Recall the notation $f_{\mathrm{enc},\theta_1}$ and $f_{\rm dec,\theta_2}$ from~\eqref{eq:f_AE} for the encoder and the decoder with parameters $\theta=(\theta_1, \theta_2)$, and denote by $\{\widetilde{\varphi}_{\mathrm{reg},\widetilde{\theta}_i}: \mathbb{R}^d\rightarrow\mathbb{R}\}_{1 \le i \le K}$ the functions corresponding to the $K$ neural networks with parameters $\widetilde{\theta}=(\widetilde{\theta}_1,\dots, \widetilde{\theta}_K)$ introduced to the regularization part of the architecture. The leading eigenfunctions $\{\varphi_i:\mathbb{R}^D\rightarrow \mathbb{R}\}_{1\le i \le K}$ are represented as compositions of the encoder and regularizers, \emph{i.e.}
$\varphi_i$ is approximated by~$\widetilde{\varphi}_{\mathrm{reg},\widetilde{\theta}_i}\circ f_{\mathrm{enc},\theta_1}$ for $i=1,\dots, K$  (recall Figure~\ref{fig-reg-ae}).
We consider the training loss 
\begin{equation}
\begin{aligned}
& \sLh_{\lambda_0,\lambda_1, \lambda_2, 
	\{\omega_i\}_{1\le i\le K}, \tau}\left(\theta, \widetilde{\theta}\right) \\
& \qquad =  
\frac{\lambda_0}{\Nd} \sum_{n=1}^{\Nd} \big\|x^n - f_{\rm dec,\theta_2} \circ f_{\mathrm{enc},\theta_1}(x^n)\big\|^2
+ \frac{\lambda_1}{\tau}\sum_{i=1}^K \omega_i
\frac{\widehat{\mathscr{E}}_{\Nd,\tau}( \widetilde{\varphi}_{\mathrm{reg},\widetilde{\theta}_i}\circ f_{\mathrm{enc},\theta_1})}{\widehat{\mathrm{Var}}_{\Nd}(\widetilde{\varphi}_{\mathrm{reg},\widetilde{\theta}_i}\circ f_{\mathrm{enc},\theta_1})}  \\
& \qquad \ \ + \lambda_2 \sum_{1 \le i \le j \le K}
\left(\widehat{\mathrm{Cov}}_{\Nd}\left(\widetilde{\varphi}_{\mathrm{reg},\widetilde{\theta}_i}\circ f_{\mathrm{enc},\theta_1}, \widetilde{\varphi}_{\mathrm{reg},\widetilde{\theta}_j}\circ f_{\mathrm{enc},\theta_1}\right) - \delta_{ij}\right)^2, 
\end{aligned}
\label{regularized-loss-in-practice}
\end{equation}
where $\Delta t>0$ is the time step of the sampled data, $\tau=\ell \Delta t$ is a lag-time for some integer~$\ell \ge 1$, $\{\omega_i\}_{1\le i \le K}$ are fixed (non-increasing) weights, $\widehat{\mathrm{Var}}_{\Nd}(\cdot)$ is the empirical variance defined in \eqref{eq:empirical-var} and, for functions $f,g: \cX \rightarrow \mathbb{R}$, we have used the shorthand notation
\begin{align*}
\widehat{\mathscr{E}}_{\Nd,\tau}(g)&=\frac{1}{2(\Nd-\ell)}\sum_{n=1}^{\Nd-\ell}|g(x^{n+\ell}) - g(x^n)|^2, \\
\widehat{\mathrm{Cov}}_{\Nd} (g,h)&=
\frac{1}{\Nd}\sum_{n=1}^{\Nd} g(x^n) h(x^n)
  - \left(\frac{1}{\Nd}\sum_{n=1}^{\Nd} g(x^n)\right)
\left(\frac{1}{\Nd}\sum_{n=1}^{\Nd} h(x^n)\right)\,.
\end{align*}
Note that \eqref{regularized-loss-in-practice} is a sum of 
the reconstruction loss and the loss in Ref.~\citenum{zhang2023understanding} for learning eigenfunctions of transfer operators (also see Refs.~\citenum{MPWN18,CSF19}). The term in factor of~$\lambda_2$ in the loss function encourages the approximate eigenfunctions~$\{ \widetilde{\varphi}_{\mathrm{reg},\widetilde{\theta}_i}\circ f_{\mathrm{enc},\theta_1} \}_{1 \leq i \leq K}$ to form an orthonormal basis since
\[
\begin{aligned}
& \widehat{\mathrm{Cov}}_{\Nd}(\widetilde{\varphi}_{\mathrm{reg},\widetilde{\theta}_i}\circ f_{\mathrm{enc},\theta_1}, \widetilde{\varphi}_{\mathrm{reg},\widetilde{\theta}_j}\circ f_{\mathrm{enc},\theta_1}) \\
& \qquad \xrightarrow[\Nd \to +\infty]{} \int_\cX \Pi \left( \widetilde{\varphi}_{\mathrm{reg},\widetilde{\theta}_i}\circ f_{\mathrm{enc},\theta_1} \right) \Pi \left(\widetilde{\varphi}_{\mathrm{reg} \widetilde{\theta}_j}\circ f_{\mathrm{enc},\theta_1} \right) d\mu,
\end{aligned}
\]
where we introduced the centering operator~$\Pi$ acting as
$(\Pi g)(x) = g(x) - \int_\cX g\, d\mu$.

The various terms in factor of~$\lambda_1$ encode the property that $\{ \widetilde{\varphi}_{\mathrm{reg},\widetilde{\theta}_i}\circ f_{\mathrm{enc},\theta_1} \}_{1 \leq i \leq K}$ are approximations of the eigenfunctions of the transfer operator~$P_\tau$ since, assuming that~$P_\tau$ is self-adjoint (see Lemma~1 and Appendix~A of Ref.~\citenum{zhang2023understanding}),
\begin{equation}
\sum_{i=1}^K\omega_i\frac{\widehat{\mathscr{E}}_{\Nd,\tau}(g_i)}{\widehat{\mathrm{Var}}_{\Nd}(g_i)} \xrightarrow[\Nd \to +\infty]{} \sum_{i=1}^K\omega_i\frac{\dps \int_\cX g_i\left(g_i-P_\tau g_i\right) d\mu}{\dps \int_\cX \left(\Pi g_i\right)^2 d\mu},
\label{eqn:raylaigh-ritz-quotient}
\end{equation}
where $P_\tau g_i(x) = \mathbb{E}(g_i(x^\ell)|x^0=x)$ encodes the average evolution of the system over the time~$\tau$. It can be shown that the quantity on the right-hand side of \eqref{eqn:raylaigh-ritz-quotient} attains the minimum value $\sum_{i=1}^K \omega_i (1-\nu_i)$, when~$\{g_i\}_{1 \le i \le K}$ are the eigenfunctions associated with the largest (but smaller than $1$) eigenvalues~$\{\nu_i\}_{1\le i\le K}$ of the transfer operator~$P_\tau$~\cite{zhang2023understanding}. This allows to estimate each individual eigenvalue $\nu_i$ after training by computing the quotient on the left-hand side of \eqref{eqn:raylaigh-ritz-quotient} using the trained neural network $\widetilde{\varphi}_{\mathrm{reg},\widetilde{\theta}_i}\circ f_{\mathrm{enc},\theta_1}$. Notice that $P_\tau$ is related to the generator~$L$ of the dynamics as $P_\tau =\rme^{\tau L}$ and the eigenvalues satisfy $\nu_i=\rme^{-\rho_i \tau}$ where~$\{\rho_i\}_{1\le i\le K}$ are the smallest non-zero eigenvalues of $-L$. This relation motivates the factor $1/\tau$ in front of the second term of the loss~\eqref{regularized-loss-in-practice}, so that its minimum is of order~1 even for~$\tau$ small.

The approach described in this section is illustrated by numerical experiments in Section~\ref{sec:ad}. 

\section{Results and discussion}

We illustrate in this section the theoretical discussions of the Methods section. We mostly present results for two-dimensional model potentials (Sections~\ref{sec:interpretation_numerics} to~\ref{sec:numerics_multiple_paths}), but also present in Section~\ref{sec:ad} an application to alanine dipeptide. Although these numerical illustrations are currently limited to somewhat low dimensional systems, we believe that most of the methods we discuss can also be applied to realistic biophysical systems which were already tackled with autoencoders. 

\subsection{Decoders as conditional expectations}
\label{sec:interpretation_numerics}

We illustrate in this section the theoretical analysis performed in Section~\ref{sec:cond_exp}. In particular, we demonstrate the fact that ideal decoders predict conditional expectations by considering a concrete example in~$\mathbb{R}^2$. The formula~\eqref{eq:Bayes_predictor_dec} suggests to consider decoders which are sufficiently expressive, in order to correctly approximate the conditional expectation. This motivates working with asymmetric autoencoders, where decoders have a larger complexity than encoders. In contrast, it is important to keep relatively simple encoders (computationally not too expensive) in molecular dynamics since they are intended to be used as CVs, for instance to perform free energy computations.

\paragraph{Architecture and training of autoencoders.} In the following, the architectures of autoencoders are characterized using the formulation of Figure~\ref{AE_notation}. We systematically use hyperbolic tangents as activation functions for all layers, except the one leading to the bottleneck and the output, for which linear activation functions are considered. As discussed in Section~\ref{sec:presentation_AEs}, autoencoders with complex architectures can overfit when the dataset is not large enough, which is why various regularization strategies are considered. As the models used in this section contain a moderate number of parameters, it is sufficient to rely on early stopping for regularization. Concretely, in the present case, the training was stopped once the validation loss did not decrease after a certain number of epochs~$n_\mathrm{wait}$. The model giving the lowest validation loss was kept. Obviously, increasing $n_\mathrm{wait}$ leads to potential overfitting, and the appropriate value of this parameter, which depends on the chosen architecture, has to be chosen with some form of cross--validation.

\paragraph{M{\"u}ller--Brown potential.} We consider the M{\"u}ller-Brown potential~\cite{Muller1979} defined for~$x=(x_1,x_2)\in \mathbb{R}^2$ as
\begin{equation}
\label{eq:Muller-Brown}
V(x_1,x_2) = \sum_{i=1}^{4} A_i \exp\left( a_i \left(x_1 - u_i\right)^2 + b_i \left(x_1 - u_i\right)\left(x_2 - v_i\right) + c_i \left(x_2 - v_i\right)^2 \right),
\end{equation}
with the parameters~$A = (-200, -100, -170, 15)$, $a = (-1, -1, -6.5, 0.7)$, $b = (0, 0, 11, 0.6)$, $c = (-10, -10, -6.5, 0.7)$, $u = (1, 0, -0.5, -1)$ and~$v = (0, 0.5, 1.5, 1)$. The dataset is obtained by sampling the Boltzmann--Gibbs measure~$\mu$ associated with~$V$ for~$\beta = 0.05$ (\emph{i.e.} the density of~$\mu$ is proportional to~$\rme^{-\beta V}$), using overdamped Langevin dynamics discretized with a Euler--Maruyama scheme and a time step $\Delta t = 10^{-4}$. A dataset of $2 \times 10^7$ points covering all the relevant parts of the potential landscape was generated by running two trajectories of $10^7$ points starting from two minima of the potential, located respectively at~$(-0.56, 1.44)$ and~$(0.62, 0.03)$.

The training and validation datasets were built by first drawing randomly without replacement $\Nd = 10^4$ points for each of them from the full dataset. The hyperparameters for the training are chosen using a $K$-fold cross--validation procedure, with~$K = 6$ folds. The test set, composed of the remaining points in the initial dataset, was used to assess the quality of the results. All the models were trained using the Adam algorithm~\cite{Adam} with a learning rate of $\eta = 0.005$. The remaining training hyperparameters to select are the model architecture, the minibatch size and the early stopping criterion~$n_\mathrm{wait}$. We first use a model with a given architecture (namely (2, 5, 5, 1, 20, 20, 2)) to select the batch size and~$n_\mathrm{wait}$, and then study the impact of the model architecture in a second stage. For all these trainings, the same random seed was used to initialize the parameters of the model. We give typical orders of magnitudes for the number of epochs in the training in the last column of Table~\ref{tab:K_fold_cross_val_batchsize_nwait}. 

\begin{table}[!ht]
	\centering
	\caption{Test reconstruction error~\eqref{eq:reconstruction_error_ideal} for the training of the AE model (2, 5, 5, 1, 20, 20, 2) with fixed learning rate $\eta=0.005$, for various values of~$n_\mathrm{wait}$ and batch sizes.}
	\label{tab:K_fold_cross_val_batchsize_nwait}
	\begin{tabular}{cccc}
		\hline
		$n_\mathrm{wait}$ &  batch size & \thead{relative test \\ reconstruction error} & \thead{number of epochs \\ (order of magnitude)} \\
		\hline
		25      & $1 \times 10^{0}$ & $1.96 \times 10^{-2}$ & 40  \\
		50      & $1 \times 10^{0}$ & $1.91 \times 10^{-2}$ & 100 \\
		100     & $1 \times 10^{0}$ & $1.90 \times 10^{-2}$ & 190 \\
  
  		25      & $1 \times 10^{1}$ & $1.81 \times 10^{-2}$ & 120 \\
		50      & $1 \times 10^{1}$ & $1.77 \times 10^{-2}$ & 190 \\
		100     & $1 \times 10^{1}$ & $1.73 \times 10^{-2}$ & 300 \\
  
  		25      & $1 \times 10^{2}$ & $1.87 \times 10^{-2}$ & 100 \\
		50      & $1 \times 10^{2}$ & $1.82 \times 10^{-2}$ & 240 \\
		100     & $1 \times 10^{2}$ & $1.66 \times 10^{-2}$ & 880 \\  
  
		25      & $5 \times 10^{2}$ & $1.90 \times 10^{-2}$ & 210 \\
		50      & $5 \times 10^{2}$ & $1.83 \times 10^{-2}$ & 300 \\
		100     & $5 \times 10^{2}$ & $1.82 \times 10^{-2}$ & 530 \\
  
        25      & $1 \times 10^{3}$ & $1.99 \times 10^{-2}$ & 500 \\
		50      & $1 \times 10^{3}$ & $1.93 \times 10^{-2}$ & 700 \\
		100     & $1 \times 10^{3}$ & $1.86 \times 10^{-2}$ & 850 \\ 

  		25      & $1 \times 10^{4}$ & $1.95 \times 10^{-2}$ & 1490 \\
		50      & $1 \times 10^{4}$ & $1.86 \times 10^{-2}$ & 2790 \\
		100     & $1 \times 10^{4}$ & $1.80 \times 10^{-2}$ & 3960 \\
		\hline
	\end{tabular}
\end{table}

Table~\ref{tab:K_fold_cross_val_batchsize_nwait} shows that, on the range of hyperparameter values under consideration, the test reconstruction error~\eqref{eq:reconstruction_error_ideal} slightly decreases as~$n_\mathrm{wait}$ increases. On the other hand, it first decreases as the batchsize increases in the range $10^0$~to~$10^2$, while it increases in the range $10^2$~to~$10^4$ (which is the size of the training dataset). The variation is however not very significant as the relative test reconstruction error (test reconstruction error divided by the empirical variance of the test data) is quite small for all the values of the parameters that are considered. The choice of batch size and early stopping criterion are not very critical as the quality of the results is rather robust to changes in these parameters. 

In view of~\eqref{eq:total_variance_formula}, the fraction of unexplained variance, which corresponds to the reconstruction error divided by the total variance, is (for a perfectly converged decoder)
\[
\frac{\E\left[\mathrm{Var}(X|\fenc(X))\right]}{\mathrm{Var}(X)}
= 1 - \frac{\mathrm{Var}\left[ \mathbb{E}(X|\fenc(X)) \right]}{\mathrm{Var}(X)}.
\]
This quantity can be estimated by replacing expectations with empirical averages on the test set. For the settings considered in Table~\ref{tab:K_fold_cross_val_batchsize_nwait}, the empirical variance of the test set is~$0.86$, and the model fails to explain less than~2\% of the variance of the test set. A more careful determination of the batch size or~$n_\mathrm{wait}$ may even further decrease this percentage, but such levels of error are sufficiently low to illustrate the theoretical points discussed in the previous sections. In the following, the trainings of AE models are done with a batchsize of~$500$ and $n_\mathrm{wait} = 100$. This particular choice of batchsize instead of~$100$ is motivated by the following facts: first, the relative test reconstruction error is not significantly smaller with a batchsize of~$100$ than with one of~$500$; secondly, the training is much faster, which allows to cover more tests; finally, the alignment criterion~\eqref{eq:cos_angle_alignement} evaluated at the values of the conditional averages is worse for a batchsize of~$100$ instead of~$500$. This last point is also true for the smaller batchsize of~$10$ and will also be discussed in the following paragraphs where two neural networks with the same encoder but different decoders are compared. 

We now illustrate the impact of the decoder's architecture on its capacity to approximate conditional expectations (see~\eqref{eq:Bayes_predictor_dec}). We trained two AEs, with respective architectures (2, 5, 5, 1, 5, 5, 2) and (2, 5, 5, 1, 20, 20, 2), using the same training and validation sets. To compute the conditional averages, the range of values of the encoder in the one dimensional latent space~$\cZ$ was split into $100$ non-overlapping intervals (bins) of constant lengths. Conditional averages are then approximated as the mean values of the data points whose encoded values are in the corresponding bins. Conditional averages and values of the decoder are plotted on the potential energy heatmap in Figures~\ref{fig:cdt_avg_heatmap_a} and~\ref{fig:cdt_avg_heatmap_b} for both architectures, together with the MEP and isolevels of the encoder~$f_\mathrm{enc}$. The values of the decoder are obtained by computing the image of points that are uniformly spaced along a one-dimensional interval, under the decoder map. The boundaries of the interval are adjusted in order for the decoder values to fit the picture.

Figures~\ref{fig:cdt_avg_heatmap_a} and~\ref{fig:cdt_avg_heatmap_b} qualitatively demonstrate that an increase in the complexity of the decoder allows it to better approximate conditional averages. Note also that the decoder path seems to be orthogonal to the isolevels of the encoder, which is in agreement with the discussion before~\eqref{eq:cos_angle_alignement}. The quality of the agreement can be made more quantitative by computing the Euclidean distance between the values of the decoder and the conditional expectations for the~100 bins considered in~$\cZ$ space, and the alignment criterion~\eqref{eq:cos_angle_alignement} between the gradient of the encoder and the derivative of the decoder; see Figures~\ref{fig:dist_cdt_avg_dec} and~\ref{fig:cos_angle_enc_dec}. These pictures show respectively that the Euclidean distance is small, and that the alignment factor is close to~1. Note also that more expressive decoders lead to better values of these criteria.

\begin{figure}[!ht]
  \centering
  \begin{subfigure}[b]{0.49\textwidth}
    \centering
    \includegraphics[width=\textwidth]{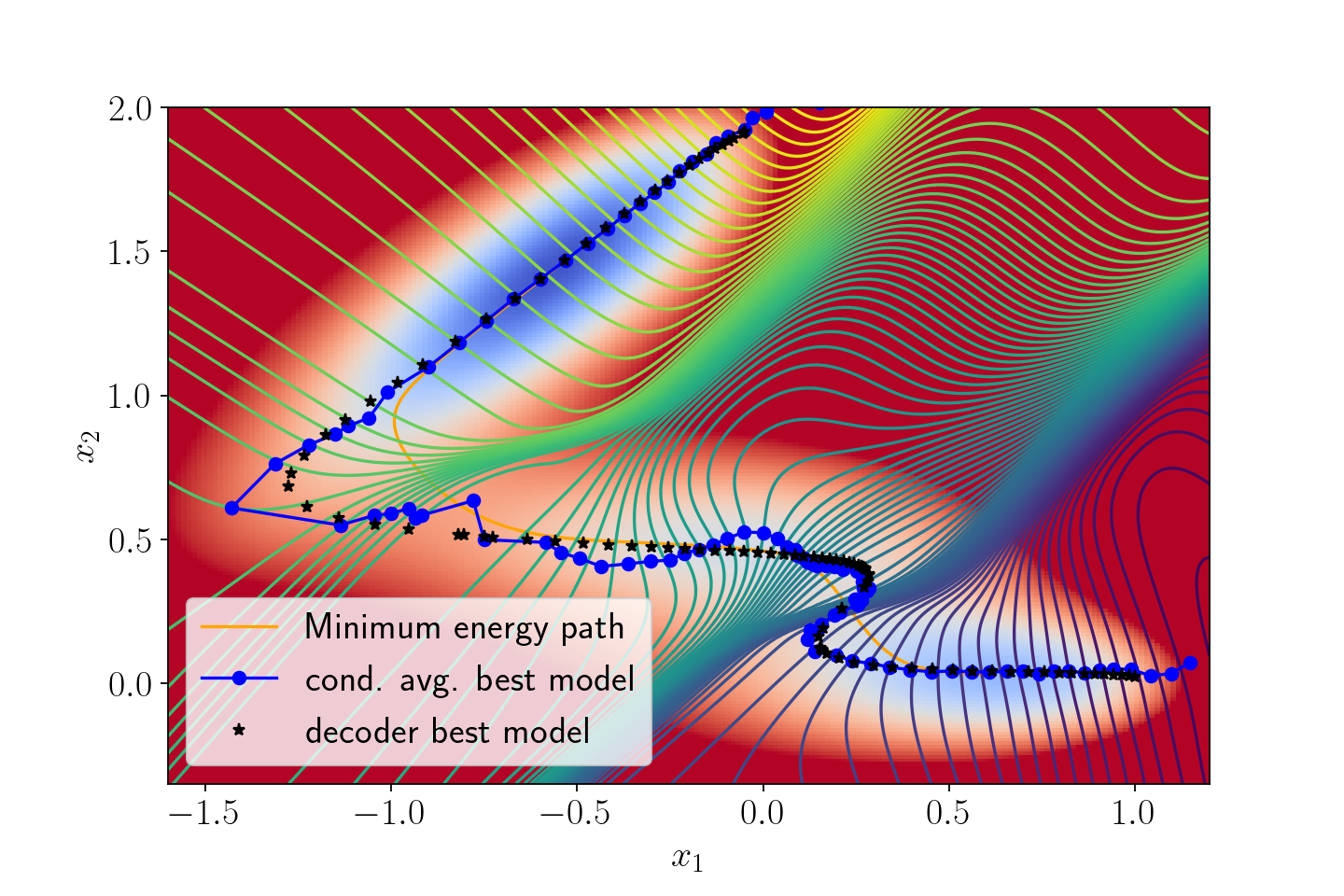}
    \caption{}
    \label{fig:cdt_avg_heatmap_a}
  \end{subfigure}
  \begin{subfigure}[b]{0.49\textwidth}
    \centering
    \includegraphics[width=\textwidth]{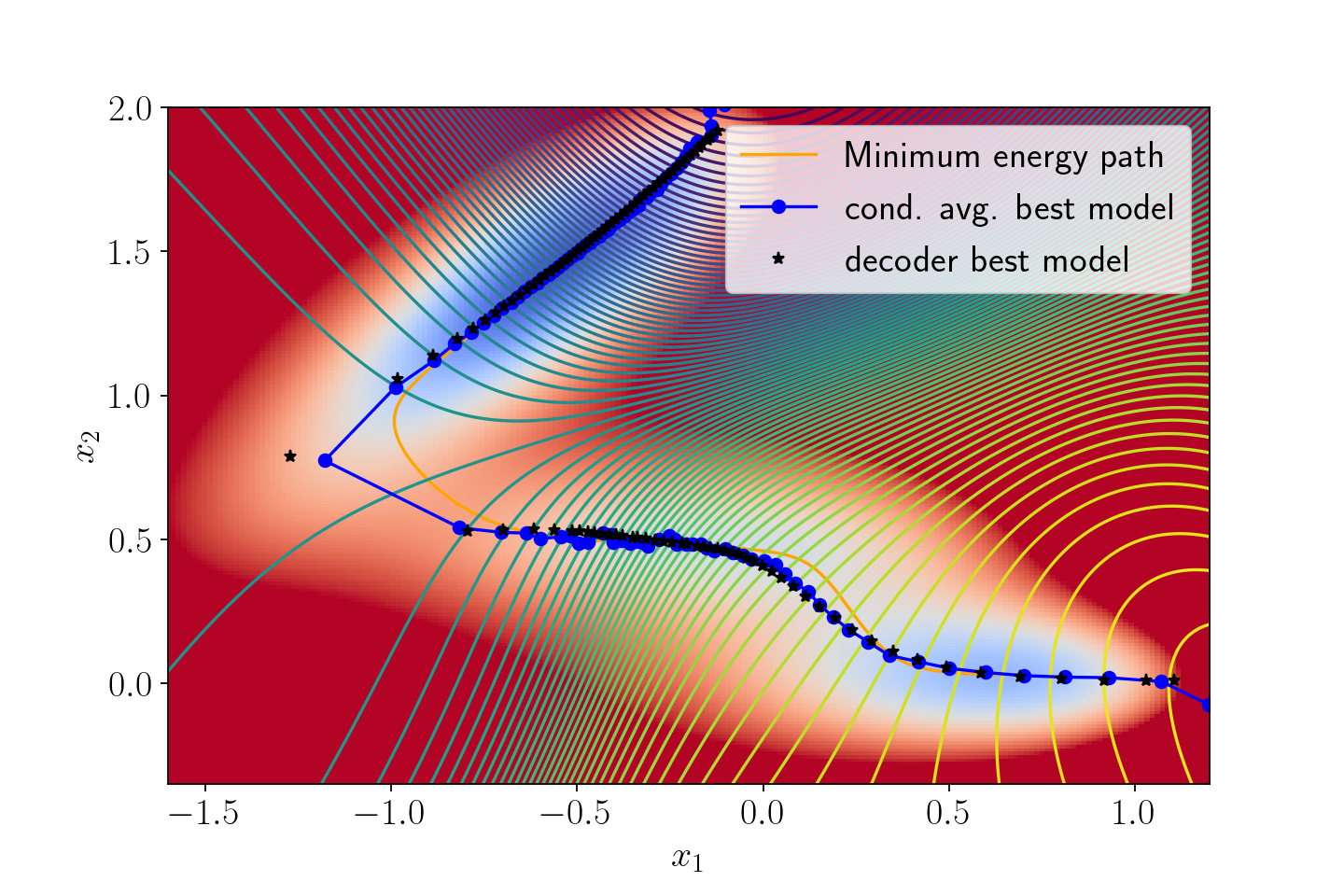}
    \caption{}
    \label{fig:cdt_avg_heatmap_b}
  \end{subfigure}
  \caption{Minimum energy path, conditional averages and decoder on the potential energy heatmap with isolevels of the encoder~$f_\mathrm{enc}$ for the AE models (a) (2, 5, 5, 1, 5, 5, 2) and (b) (2, 5, 5, 1, 20, 20, 2). Conditional averages are computed using 100 uniformly spaced bins in encoder space.}
  \label{fig:cdt_avg_heatmap}
\end{figure}

\begin{figure}[!ht]
	\centering
	\begin{subfigure}[b]{0.49\textwidth}
		\centering
		\includegraphics[width=\textwidth]{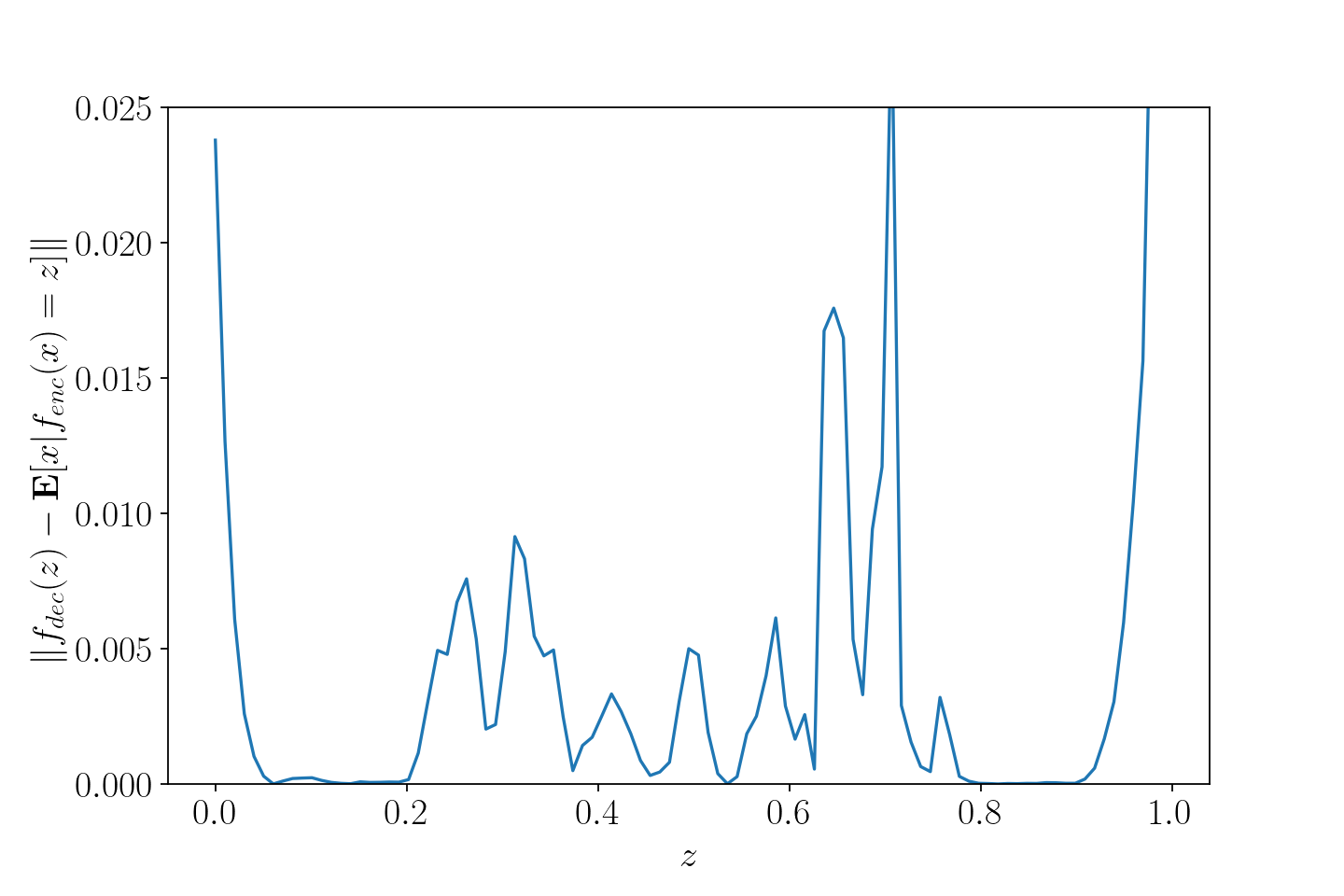}
		\caption{}
            \label{fig:dist_cdt_avg_dec_a}
	\end{subfigure}
	\begin{subfigure}[b]{0.49\textwidth}
		\centering
		\includegraphics[width=\textwidth]{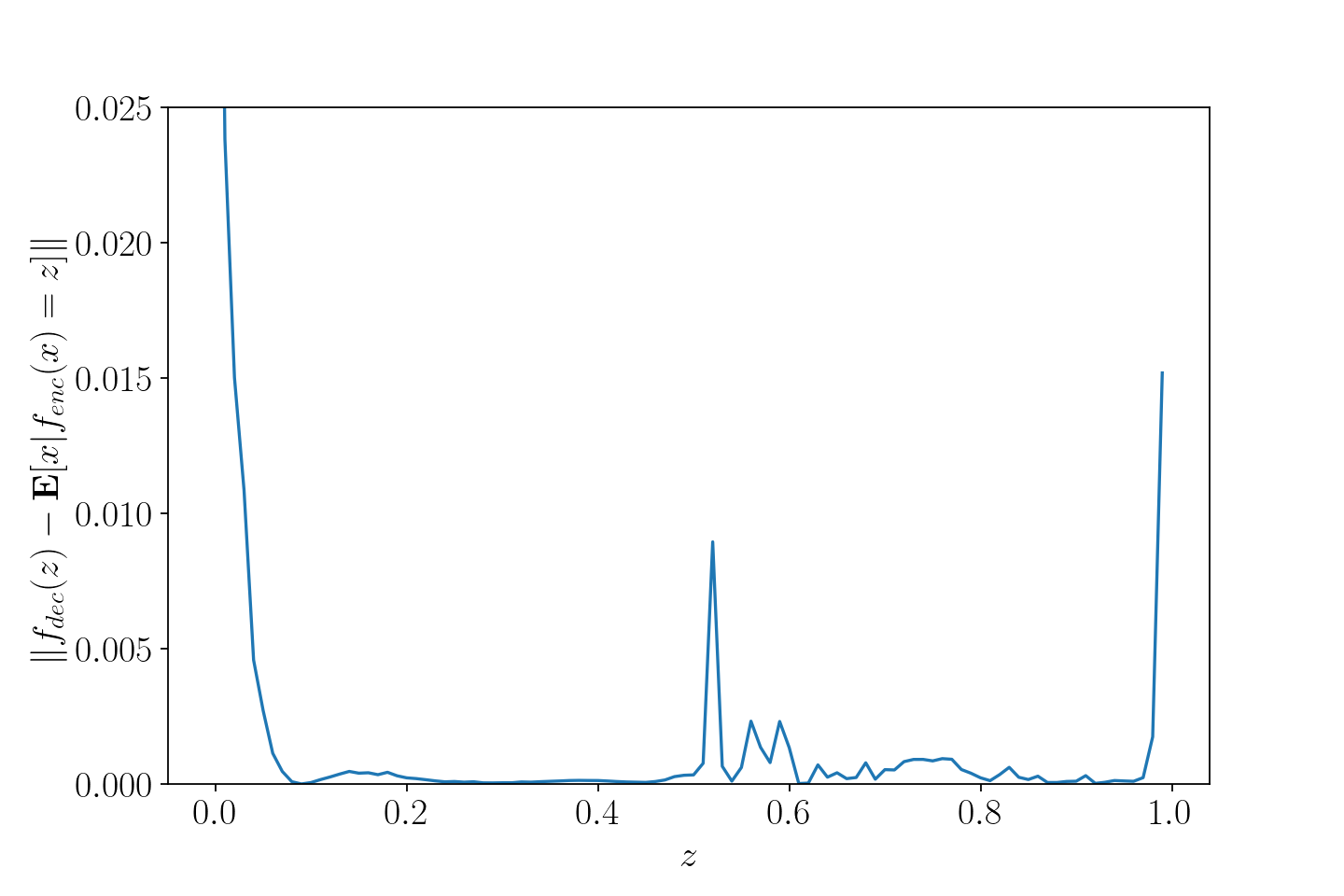}
		\caption{}
            \label{fig:dist_cdt_avg_dec_b}
	\end{subfigure}
	\caption{Euclidean distance between the conditional averages and the values of the decoder, where conditional averages are computed using 100 uniformly spaced bins in encoder space~$\cZ$ for the AE models (a) (2, 5, 5, 1, 5, 5, 2) and (b) (2, 5, 5, 1, 20, 20, 2).}
	\label{fig:dist_cdt_avg_dec}
\end{figure}

\begin{figure}[!ht]
	\centering
	\begin{subfigure}[b]{0.49\textwidth}
		\centering
		\includegraphics[width=\textwidth]{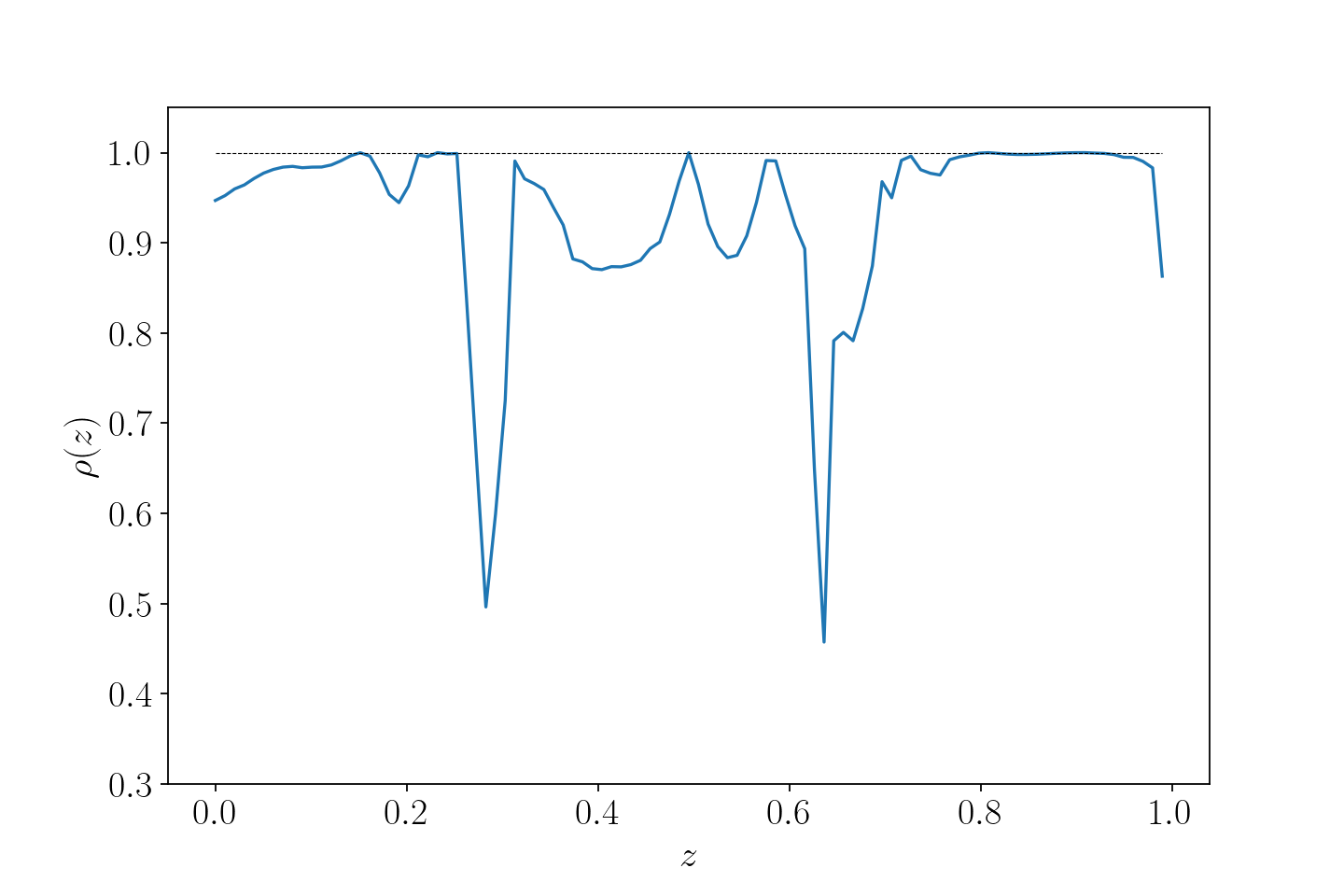}
		\caption{}
            \label{fig:cos_angle_enc_dec_a}
	\end{subfigure}
	\begin{subfigure}[b]{0.49\textwidth}
		\centering
		\includegraphics[width=\textwidth]{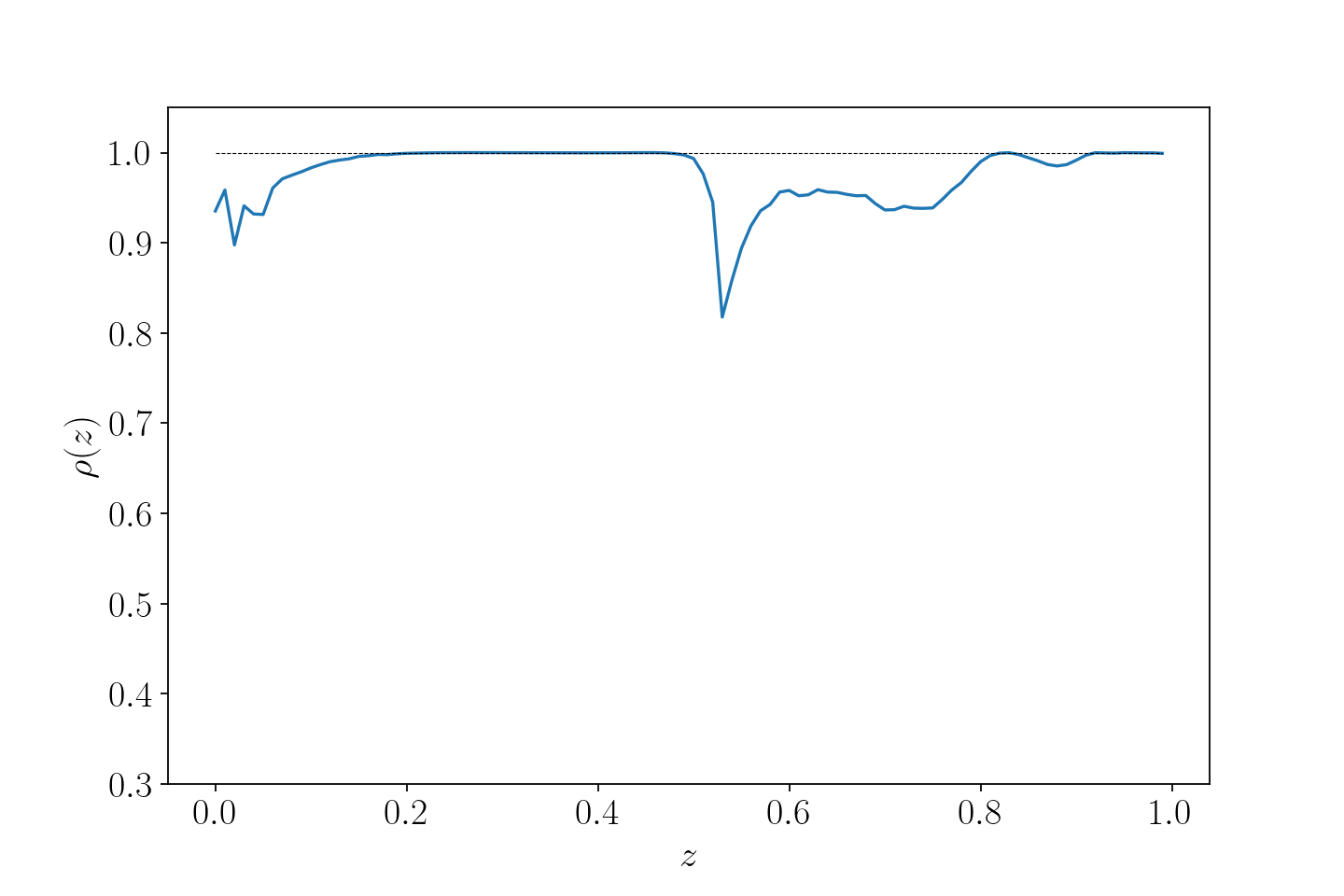}
		\caption{}
            \label{fig:cos_angle_enc_dec_b}
	\end{subfigure}
	\caption{Alignment criterion~\eqref{eq:cos_angle_alignement} evaluated at the values of the conditional averages computed using 100 uniformly spaced bins in encoder space~$\cZ$ for the AE models (a) (2, 5, 5, 1, 5, 5, 2) and (b) (2, 5, 5, 1, 20, 20, 2).}
	\label{fig:cos_angle_enc_dec}
\end{figure}

Note moreover from Figure~\ref{fig:cdt_avg_heatmap} that the decoder path with the more complex architecture is closer to the MEP, although there are some discrepancies. The agreement would presumably be better for larger values of~$\beta$. Although more complex decoders allow to better approximate conditional expectations, let us mention that the associated test losses may not be smaller than those for less complex decoders. The value of the test loss alone may therefore be misleading to assess the convergence of the decoder to the conditional average, as a trained model with a slightly smaller test loss might lead to an encoder/decoder pair less close to the conditional average (in the sense of equation~\eqref{eq:cos_angle_alignement}). Here, the relative test reconstruction error for the trained model with the smaller decoder was~$1.71 \times 10^{-2}$, which is comparable to the relative test reconstruction error~$1.83 \times 10^{-2}$ obtained with the larger decoder (the fact that this reconstruction error can be larger may seem surprising, but can be attributed to some variability in the training procedure and choice of data sets). On the other hand, the decoder path qualitatively looks nicer for the more complex model, and this can be quantified using our measure of alignment and/or the distance between the decoder path and conditional averages, as illustrated in Figures~\ref{fig:dist_cdt_avg_dec} and~\ref{fig:cos_angle_enc_dec}. 

\subsection{Improving the description by modifying the probability distribution of the data}
\label{sec:numerics_modifying_proba_dist}

Following the discussion in Section~\ref{sec:changing_ref_measure}, we illustrate the impact of changing the reference probability measure~$\mu$ of the data distribution on the quality of the autoencoder. We therefore consider in this section the loss function~\eqref{eq:sLh_lambda}, and the same two dimensional system as in Section~\ref{sec:interpretation_numerics}.

When training an autoencoder on this model with any kind of architecture, there is always a chance that the training converges to a ``wrong'' solution, depending on the random initialization of the weights and the subsequent randomness in the training procedure. Such a situation is illustrated in Figure~\ref{fig:bolz_vs_reac_cdt_avg_heatmap_a}. Although this ``wrong'' solution allows to capture the variance, the decoder path between the probability mode on the left and the one on the right is far away from the MEP, and does not have a particular meaning. Somehow, the mode on the left is described in the wrong direction. We use such spurious solutions to study the impact of the change of the reference probability measure by observing whether the wrong solution subsists after training.  

To obtain the reactive trajectory distribution with AMS simulations, the initial state~$A$ and final state~$B$ are defined as small discs of radius~$0.1$ centered on the two local minima of the potential, and the one dimensional reaction coordinate used for AMS is the finite element approximation of the committor function (as done in Ref.~\citenum{ECT22} for instance). A dataset of $\wNd = 2 \times 10^5$ configurations~$\{\widetilde{x}^n\}_{1 \leq n \leq \wNd}$ distributed according to the reactive trajectory measure is obtained by drawing at random without replacement configurations obtained by running two AMS simulations: one forward from~$A$ to~$B$, and one backward from~$B$ to~$A$, with 1000 replicas each. To compare visually this distribution with the Boltzmann--Gibbs distribution, two sets of $10^4$ points distributed according to these two distributions are plotted in Figure~\ref{fig:bolz_vs_reac_on_heatmap}.

\begin{figure}[!ht]
	\centering
	\begin{subfigure}[b]{0.49\textwidth}
		\centering
		\includegraphics[width=\textwidth]{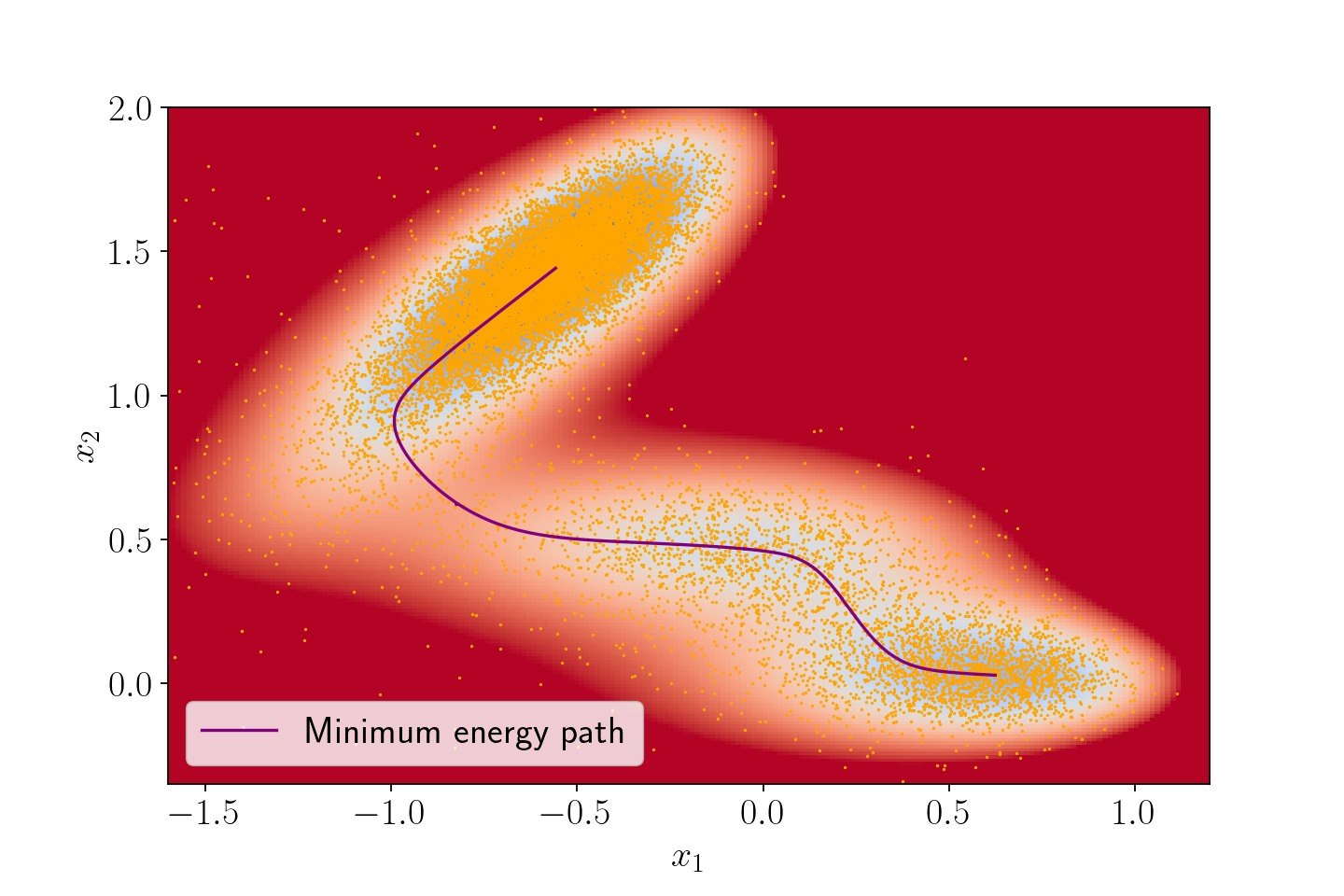}
		\caption{}
            \label{fig:bolz_vs_reac_on_heatmap_a}
	\end{subfigure}
	\begin{subfigure}[b]{0.49\textwidth}
		\centering
		\includegraphics[width=\textwidth]{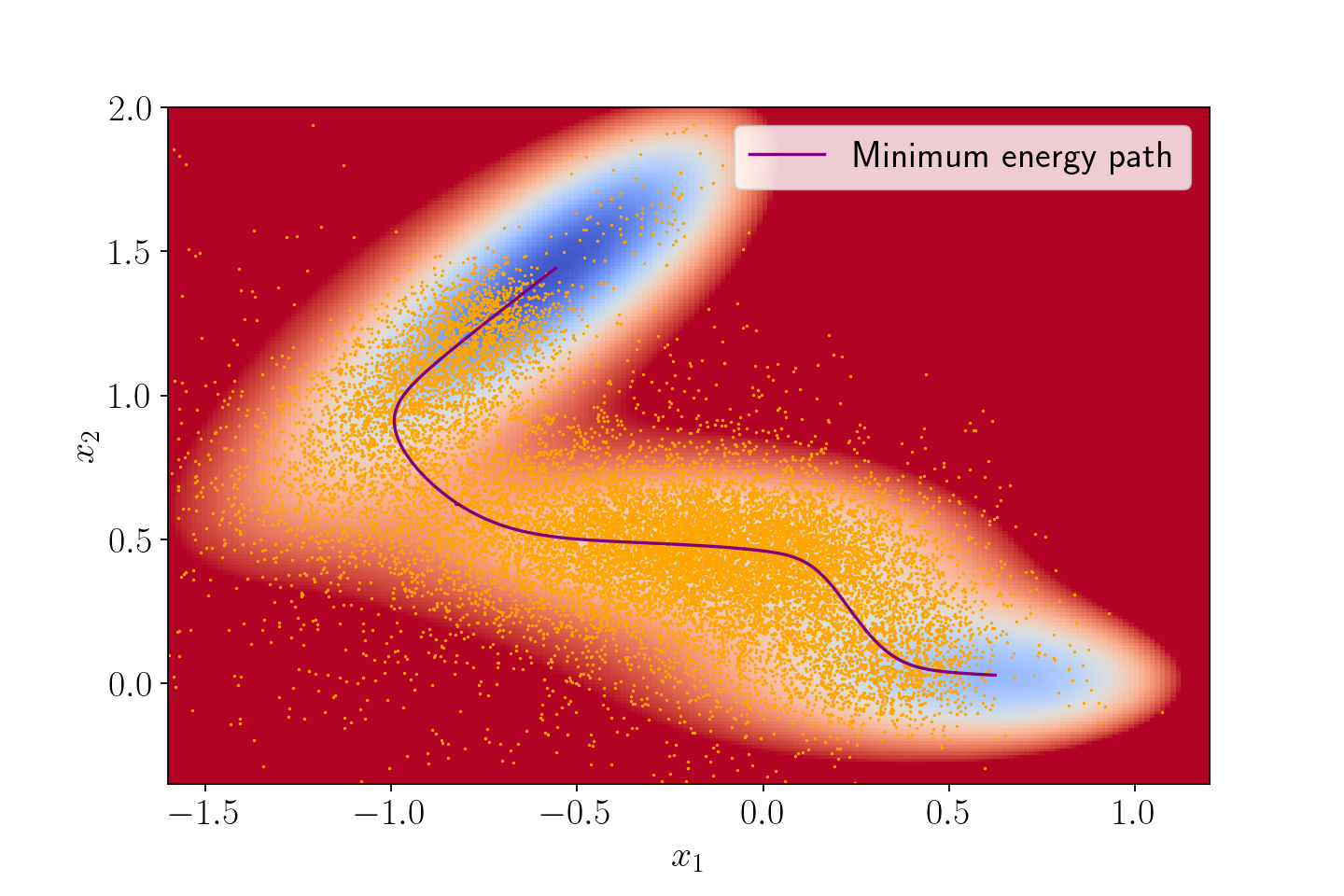}
		\caption{}
            \label{fig:bolz_vs_reac_on_heatmap_b}
	\end{subfigure}
	\caption{Minimum energy path (purple line) and configurations (orange dots) distributed according to (a) the Boltzmann--Gibbs distribution, and (b) the reactive trajectory distribution.}
	\label{fig:bolz_vs_reac_on_heatmap}
\end{figure}

The differences between the results obtained with the trained models in Figures~\ref{fig:cdt_avg_heatmap_b} and~\ref{fig:bolz_vs_reac_cdt_avg_heatmap_a} come from the values of the initial weights of the model and the changes in the chosen minibatches during the training. Here, the initial weights of the models are chosen at random, using the default initialization of PyTorch (see Section~8.4 of Ref.~\citenum{GBC16} for further background). This initialization and the choices made in the minibatching procedures are determined by the seed of the random number generators, which is set to different values for Figures~\ref{fig:cdt_avg_heatmap_b} and~\ref{fig:bolz_vs_reac_cdt_avg_heatmap_a}.

\begin{figure}[!ht]
	\centering
	\begin{subfigure}[b]{0.49\textwidth}
		\centering
		\includegraphics[width=\textwidth]{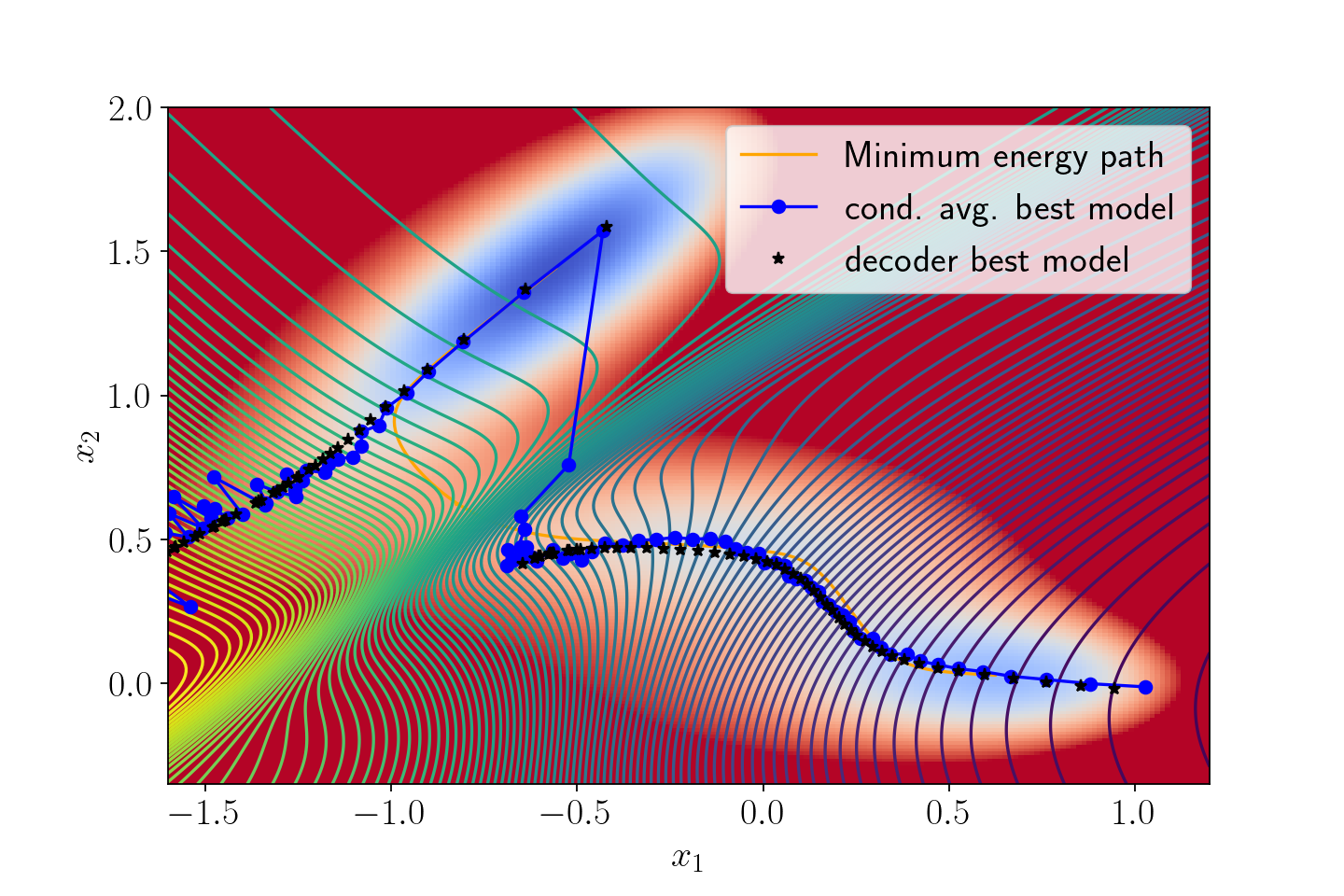}
		\caption{}
            \label{fig:bolz_vs_reac_cdt_avg_heatmap_a}
	\end{subfigure}
	\begin{subfigure}[b]{0.49\textwidth}
		\centering
		\includegraphics[width=\textwidth]{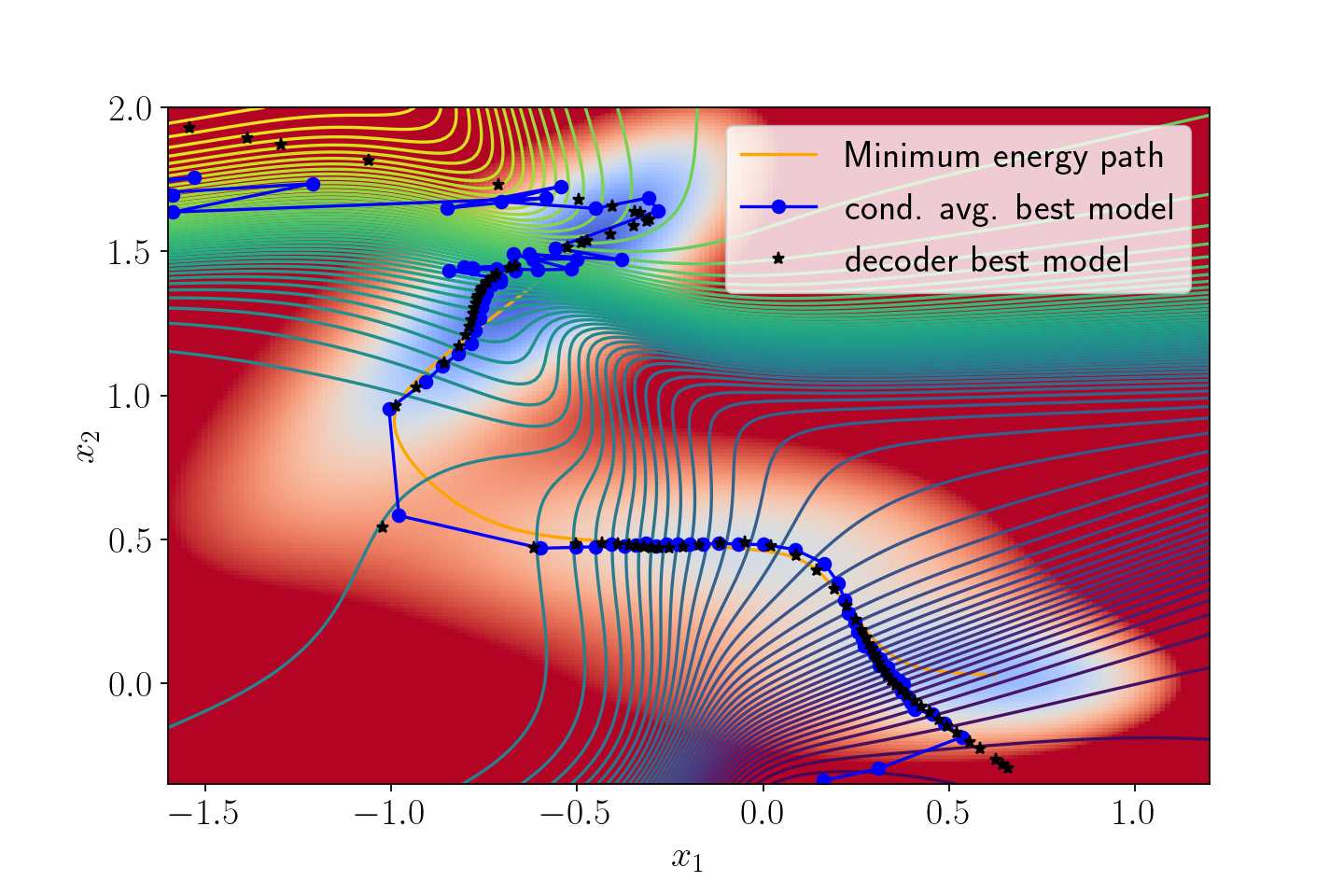}
		\caption{}
            \label{fig:bolz_vs_reac_cdt_avg_heatmap_b}
	\end{subfigure}
	\caption{Minimum energy path, conditional averages and decoder on the potential energy heatmap with isolevels of the encoder~$f_\mathrm{enc}$ for the AE models (2, 5, 5, 1, 20, 20, 2) trained using (a) only the Boltzmann--Gibbs distribution ($\lambda = 1$) and (b) only the reactive trajectory distribution ($\lambda = 0)$ as reference measures. Conditional averages are computed using (a) the Boltzmann-Gibbs distribution or (b) the reactive trajectory distribution with 100 uniformly spaced bins in encoder space~$\cZ$.}
	\label{fig:bolz_vs_reac_cdt_avg_heatmap}
\end{figure}

\begin{figure}[!ht]
	\centering
	\begin{subfigure}[b]{0.49\textwidth}
		\centering
		\includegraphics[width=\textwidth]{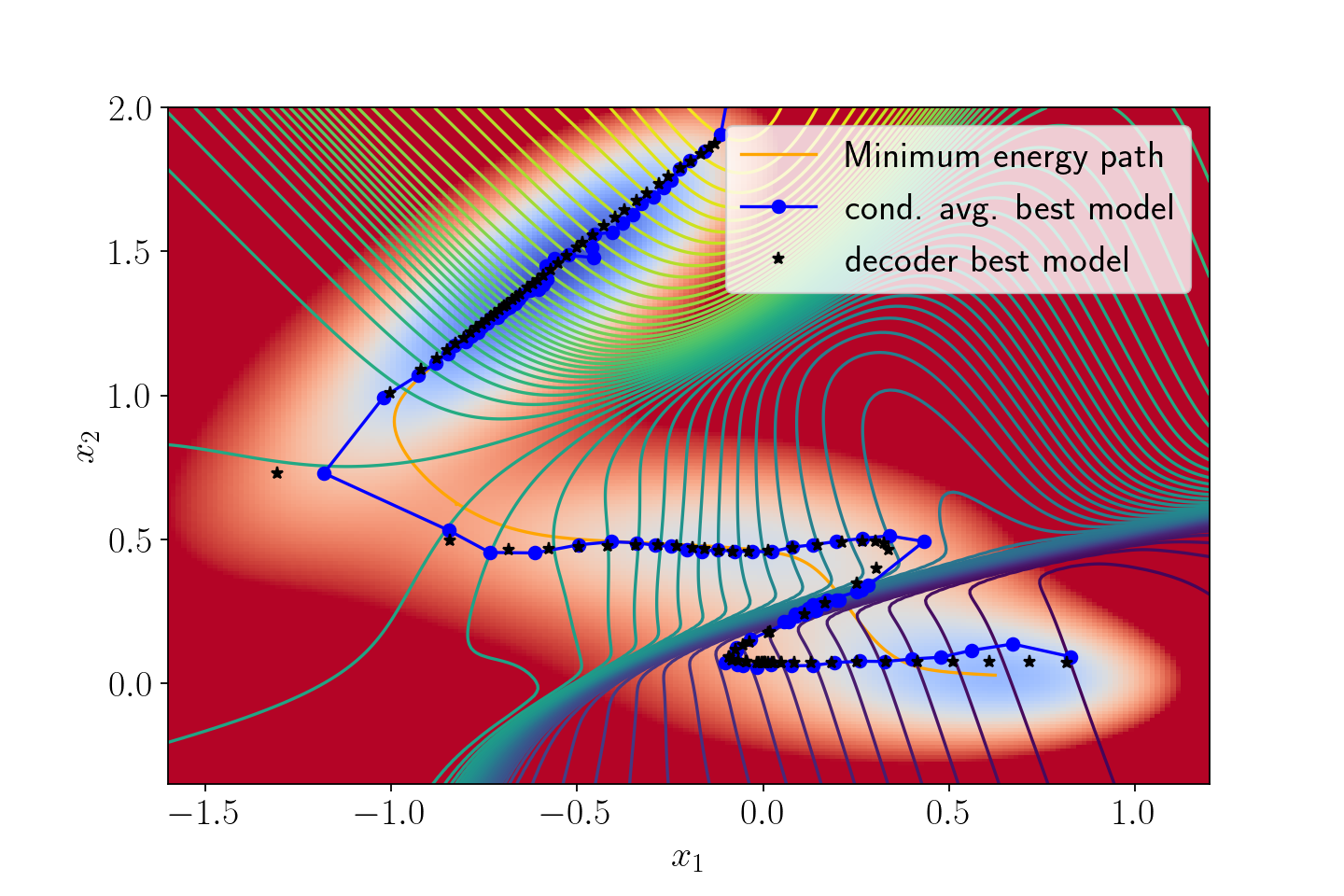}
		\caption{}
            \label{fig:bolz_plus_reac_cdt_avg_heatmap_a}
	\end{subfigure}
	\begin{subfigure}[b]{0.49\textwidth}
		\centering
		\includegraphics[width=\textwidth]{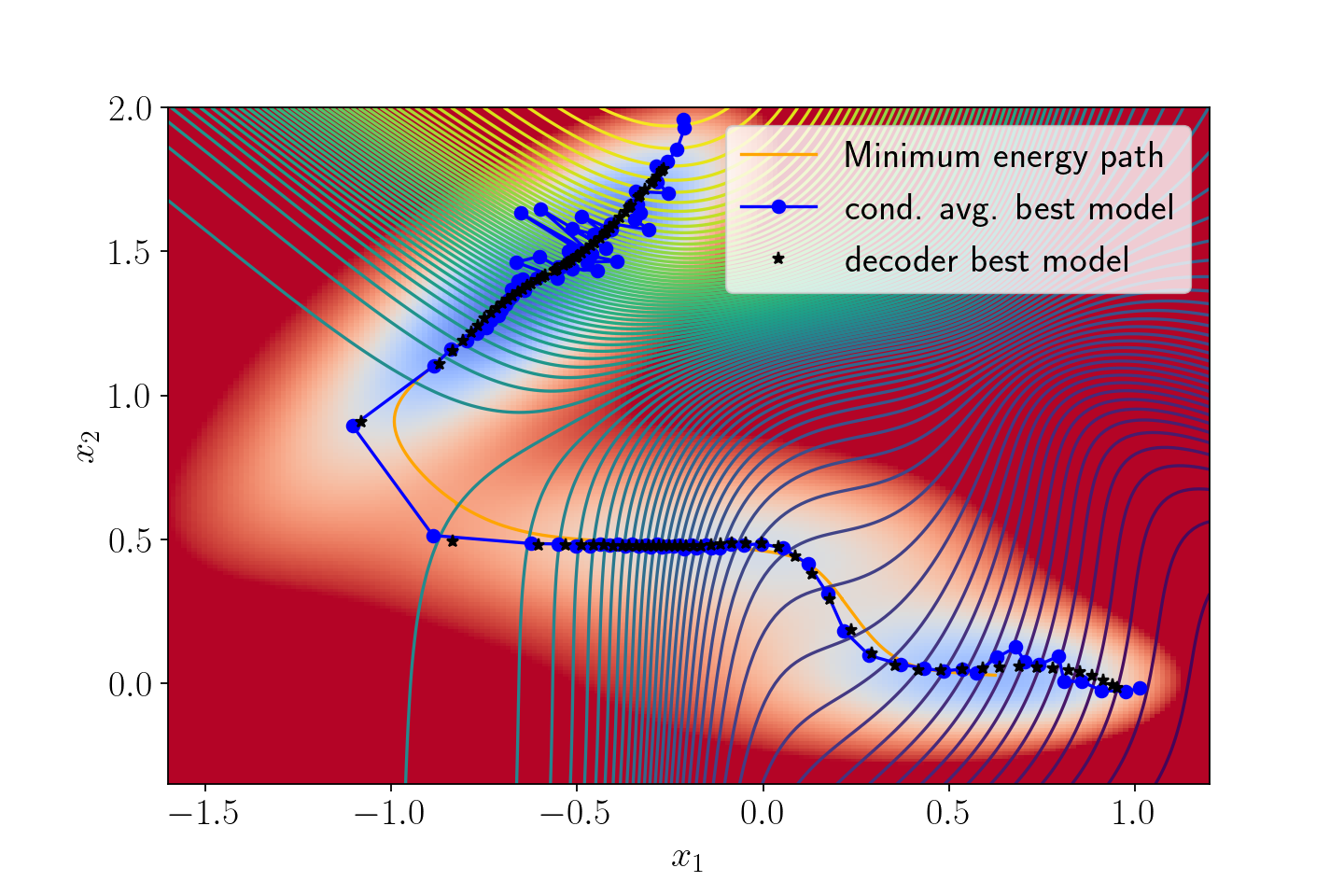}
		\caption{}
            \label{fig:bolz_plus_reac_cdt_avg_heatmap_b}
	\end{subfigure}
	\caption{Minimum energy path, conditional averages and decoder on the potential energy heatmap with isolevels of the encoder~$f_\mathrm{enc}$ for the AE models (2, 5, 5, 1, 20, 20, 2) trained on~\eqref{eq:sLh_lambda} with (a) $\lambda = 0.5$ and (b) $\lambda = \frac{1}{11}$. Conditional averages are consistently computed using a $\lambda$-mixture of both measures with 100 uniformly spaced bins in encoder space~$\cZ$.}
	\label{fig:bolz_plus_reac_cdt_avg_heatmap}
\end{figure}

Depending on the relative weight of the two reconstruction errors in the definition of~$\sLh_\lambda(\theta)$, the ``wrong solution'' can be avoided since configurations distributed according to the reactive measure emphasize reconstruction efforts in the transition region. In Figures~\ref{fig:bolz_vs_reac_cdt_avg_heatmap} and~\ref{fig:bolz_plus_reac_cdt_avg_heatmap}, the training results starting from the same initial conditions are shown for various values of~$\lambda$. For~$\lambda=1$ (training on the Boltzmann--Gibbs distribution only), the ``wrong solution'' is obtained. For $\lambda=0$ (only configurations from the reactive path measure), transition regions are well described but the two important modes in the probability measure are not well taken into account. Additional numerical tests (not reported here) indicate that the best results are obtained for~$\lambda$ small, as this is a situation where the addition of a significant fraction of points in the transition regions allows to avoid converging to the ``wrong solution''. The value of the parameter~$\lambda$ could be fine-tuned (using a $K$-fold cross-validation procedure for instance) to obtain the most satisfactory results in terms of the values of the alignment criterion~\eqref{eq:cos_angle_alignement} and distance between decoder values and conditional averages. 

\subsection{Describing multiple transition paths}
\label{sec:numerics_multiple_paths}

To illustrate the approach presented in Section~\ref{sec:multiple_path_2_state} to describe multiple transition paths, we first consider a double well potential on a circle. More precisely, we consider the potential energy function in $\mathbb{R}^2$
\begin{equation}
\label{eq:doublewell_circle}
\begin{aligned}
V(x_1,x_2) = 2 x_2^2 + 10 \left( x_1^2 + x_2^2 -1\right)^2.
\end{aligned}
\end{equation}
A dataset of $2 \times 10^5$ points distributed according to the Boltzmann--Gibbs measure was generated using a Euler--Maruyama discretization of the overdamped Langevin dynamics with $\beta = 2$ and $\Delta t = 0.01$. The early stopping waiting time $n_\mathrm{wait}=50$, learning rate~$0.005$ and minibatch size~$500$ were kept constant for all the subsequent training experiments with~\eqref{eq:loss_multiple_dec} and its variations. The weights of the regularization terms were not chosen using cross-validation, but chosen based on preliminary runs to identify appropriate values. 

\paragraph{Unregularized approach.}
We compare in Figure~\ref{fig:double_dec_circle} results obtained with $K=2$ decoders when minimizing the unregularized loss~\eqref{eq:loss_multiple_dec} to results with a single decoder. When a single decoder is considered, the decoder path describes accurately both modes, but only one transition path is correctly indexed, while on the second path, in the zone of lowest weight according to the Boltzmann--Gibbs measure, the description is less precise and the encoder function varies very fast. When~$K=2$ decoders are considered, each decoder takes care of a mode of the probability measure, with some interface more or less in the middle. The encoder has rather horizontal isolevels in the modes, constructed so that the decoder path is close to some principal curve. Of course, this does not allow to infer anything about the transition mechanism for configurations switching from one probability mode to another.

\begin{figure}[!ht]
	\centering
	\begin{subfigure}[b]{0.49\textwidth}
		\centering
		\includegraphics[width=\textwidth]{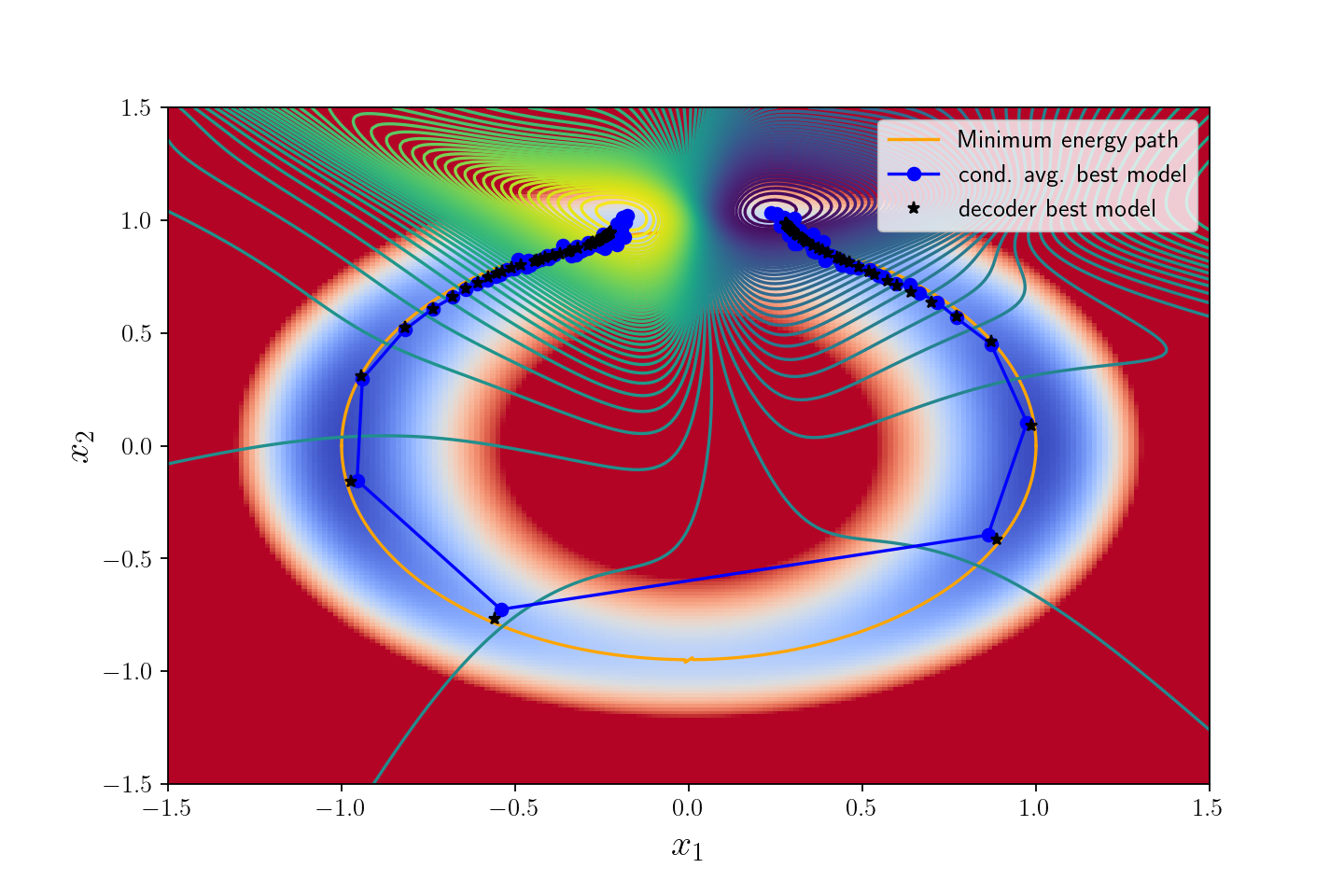}
		\caption{}
            \label{fig:double_dec_circle_a}
	\end{subfigure}
	\begin{subfigure}[b]{0.49\textwidth}
		\centering
		\includegraphics[width=\textwidth]{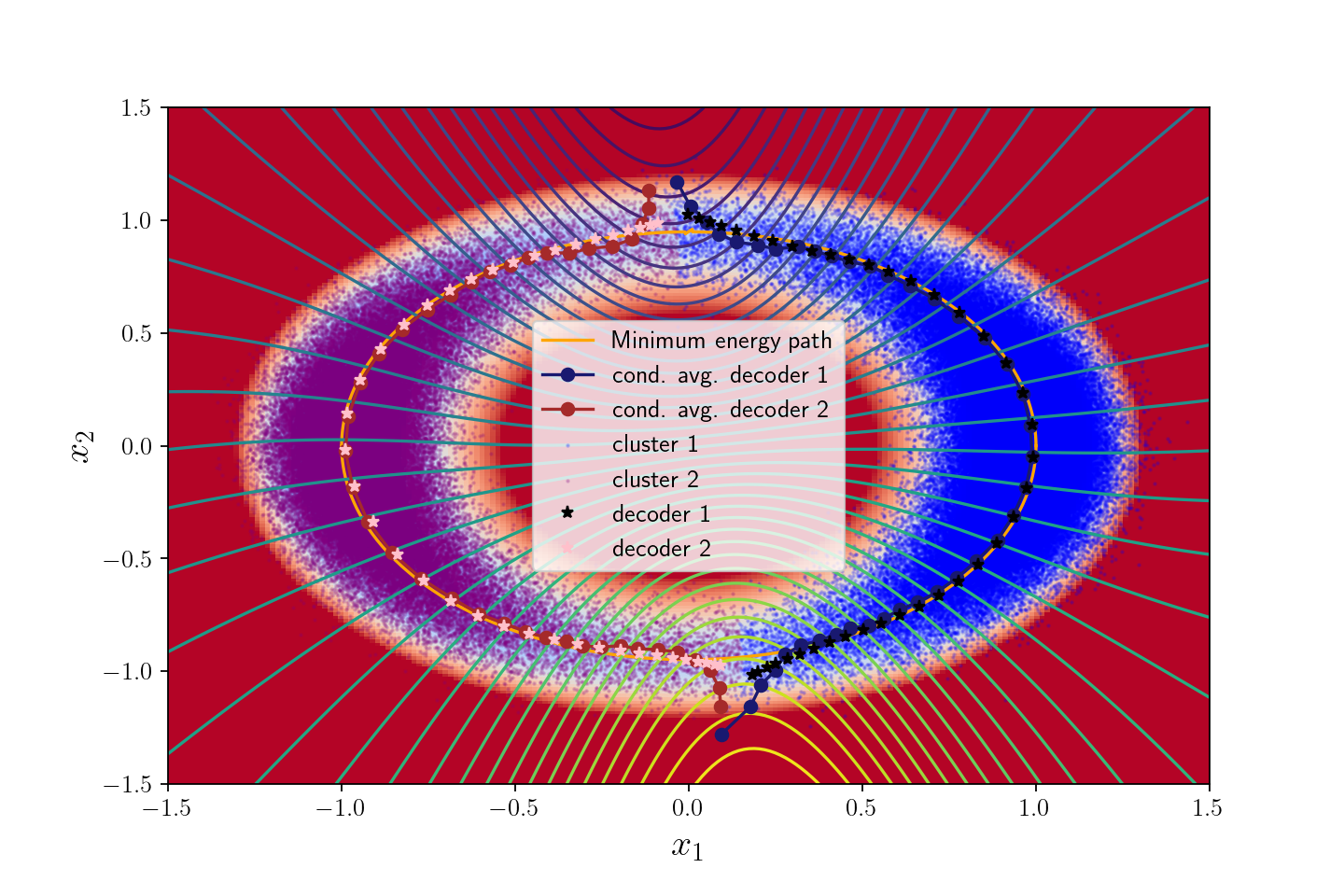}
		\caption{}
            \label{fig:double_dec_circle_b}
	\end{subfigure}
	\caption{(a): Minimum energy path, conditional averages and decoders on the potential energy heatmap with isolevels of the encoder $\fenc$ for models with a single decoder ($K=1$). (b): Same plots for a model with~$K=2$ decoders on top of which the training dataset points are colored according to the index of the decoder which leads to the minimal reconstruction error. Conditional averages are computed using 100 bins uniformly spaced in the encoded dimension.}
	\label{fig:double_dec_circle}
\end{figure}

Further physical insights need to be injected in the model in order to obtain a reconstruction of the system which is also relevant to describe the transition paths switching from one metastable region to another. One possible approach to this end would be to train the model on configurations sampled according to the reactive path measure, already considered in Section~\ref{sec:changing_ref_measure}, as each decoder would concentrate on one mode of this probability measure (located around the transition points in the vicinity of~$(0,-1)$ and~$(0,1)$).

\begin{figure}[!ht]
	\centering
	\begin{subfigure}[b]{0.45\textwidth}
		\centering
		\includegraphics[width=\textwidth]{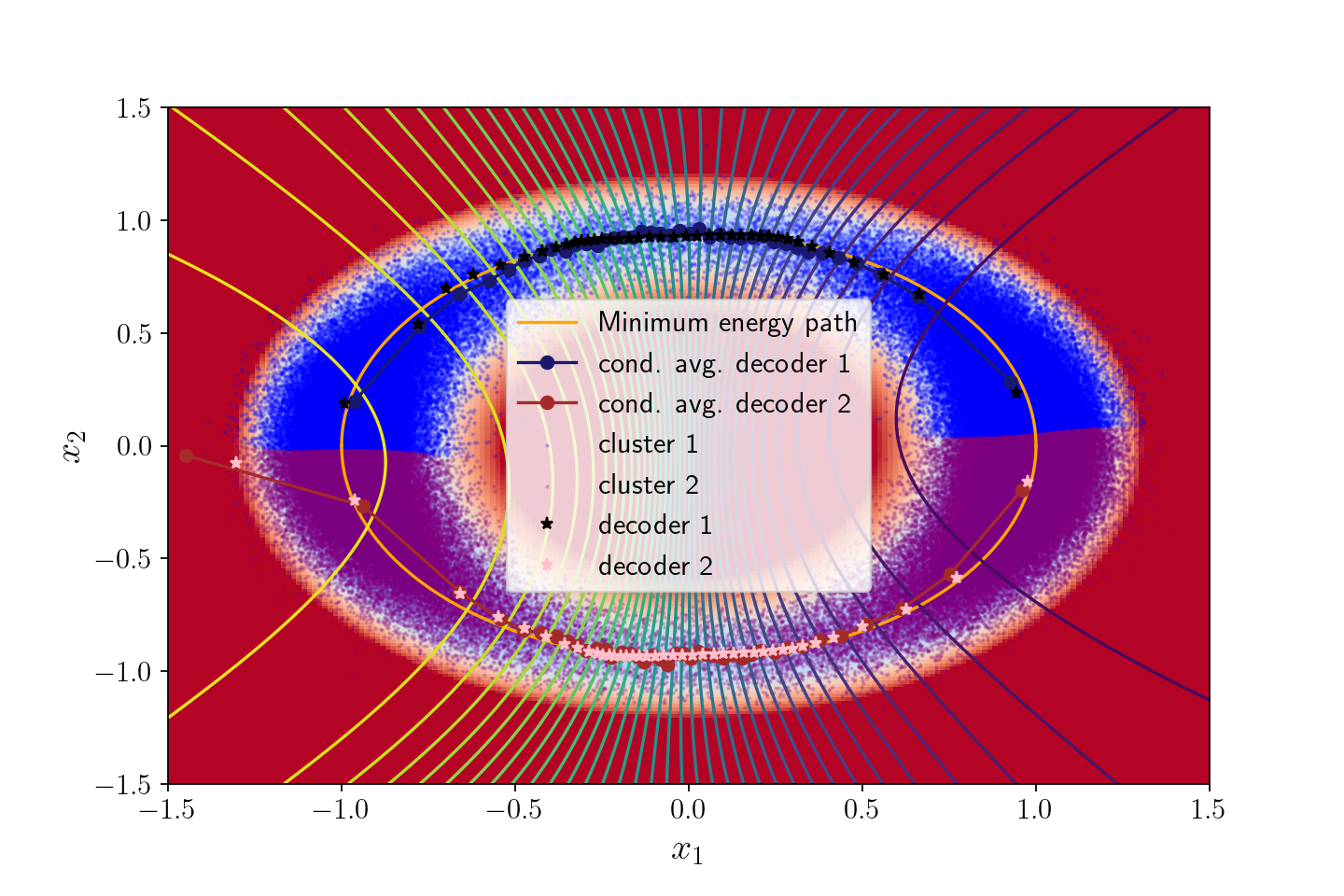}
		\caption{}
            \label{fig:double_decoder_b}
	\end{subfigure}
	\begin{subfigure}[b]{0.45\textwidth}
		\centering
		\includegraphics[width=\textwidth]{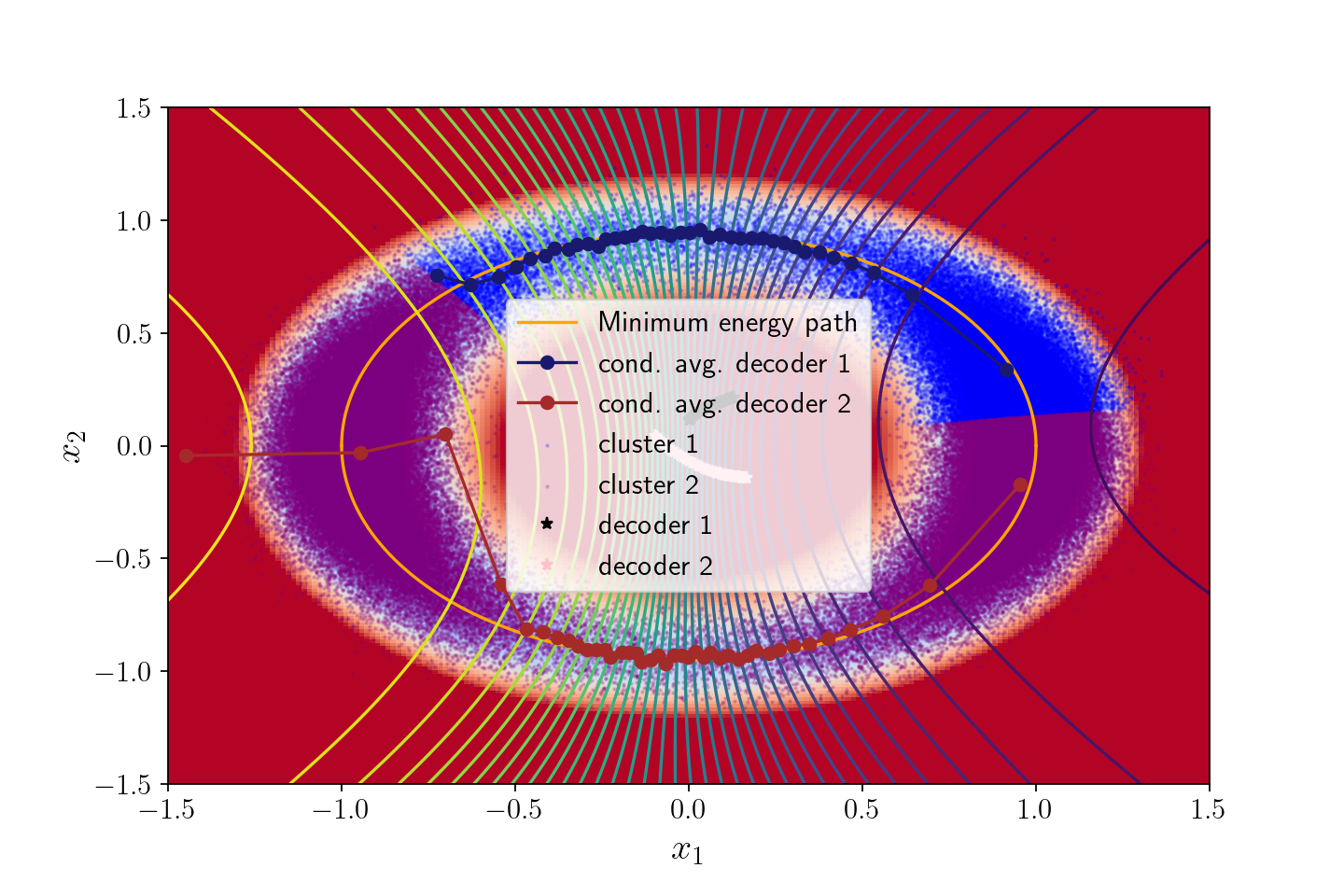}
		\caption{}
            \label{fig:double_decoder_c}
	\end{subfigure}
	\caption{Minimum energy path, conditional averages and decoders on the potential energy heatmap with isolevels of the encoder $f_\mathrm{enc}$ for models trained to minimize the loss~\eqref{eq:loss_multiple_dec_constrained} 
		(a) $\lambda_0 = 1$, $\lambda_1 = 10^4$ and $\lambda_2 = 1$ and (b) $\lambda_0 = 0$, $\lambda_1 = 10^4$ and $\lambda_2 = 1$. These results should be compared to the one obtained in Figure~\ref{fig:double_dec_circle_b}, which corresponds to the choice~$\lambda_0 = 1$ and $\lambda_1 = \lambda_2 = 0$.}
	\label{fig:double_decoder}
\end{figure}

\paragraph{Training with regularization.}
Including the regularization terms to the multiple decoder loss as in~\eqref{eq:loss_multiple_dec_constrained} allows to obtain a model for which each decoder correctly indexes a different transition path (see Figure~\ref{fig:double_decoder_b}). In fact using only the added penalization terms ($\lambda_0 = 0$) allows to obtain a collective variable able to index both paths (see Figure~\ref{fig:double_decoder_c}). The reconstruction error with the two decoders allows to perform a clustering of the dataset corresponding to the two transition paths. 

We illustrate the behavior of models obtained by minimizing the loss function~\eqref{eq:loss_multiple_dec_constrained_2} for the entropic switch potential~\cite{Park2003}:
\begin{equation}
\label{eq:entropic_switch_potential}
\begin{aligned}
V(x_1,x_2) & = 3 \, \mathrm{e}^{-x_1^2} \left( \mathrm{e}^{ - \left(x_2 - \frac{1}{3} \right)^2} - \mathrm{e}^{-\left(x_2 - \frac{5}{3}\right)^2} \right) - 5 \,  \mathrm{e}^{-x_2^2} \left( \mathrm{e}^{ - \left(x_1 - 1 \right)^2} - \mathrm{e}^{-\left(x_1 + 1 \right)^2} \right) \\ & \;+ 0.2 \, x_1^4 + 0.2 \left(x_2 - \frac{1}{3}\right)^4.
\end{aligned}
\end{equation}
The dataset of points distributed according to the associated Boltzmann--Gibbs measure was generated following the same procedure as for the M{\"u}ller--Brown potential (see Section~\ref{sec:interpretation_numerics}) with a time step~$\Delta t = 0.01$ and~$\beta = 2$. There are $N_\mathrm{pen} = 2$ points at which the values of the encoder are fixed, respectively to~0 for~$\wx^1 = (-1.05, 0.04)$ and to~1 for~$\wx^2 =(1.05, -0.04)$. The training results of an AE model (2, 5, 5, 1, 20, 20, 2) with two decoders using the loss~\eqref{eq:loss_multiple_dec_constrained_2} are presented in Figure~\ref{fig:ent_switch_doub_dec}. Using the semi--supervised approach to constrain both decoders in each well (Figure~\ref{fig:ent_switch_doub_dec_b}) allows to differentiate correctly the top and bottom paths linking the right and left well. Depending on the initial parameters in the model, when the variance term is used instead (Figure~\ref{fig:ent_switch_doub_dec_a}), these paths are not as nicely separated. Such a behavior cannot be detected by observing the convergence of the decoder to the conditional averages, or the angle between the gradient of the encoder and the derivative of the decoder. One needs to rely on the semi-supervised nature of the problem, and analyze decoder paths to see whether they go from one probability mode to the other. 

\begin{figure}[!ht]
	\centering
	\begin{subfigure}[b]{0.49\textwidth}
		\centering
		\includegraphics[width=\textwidth]{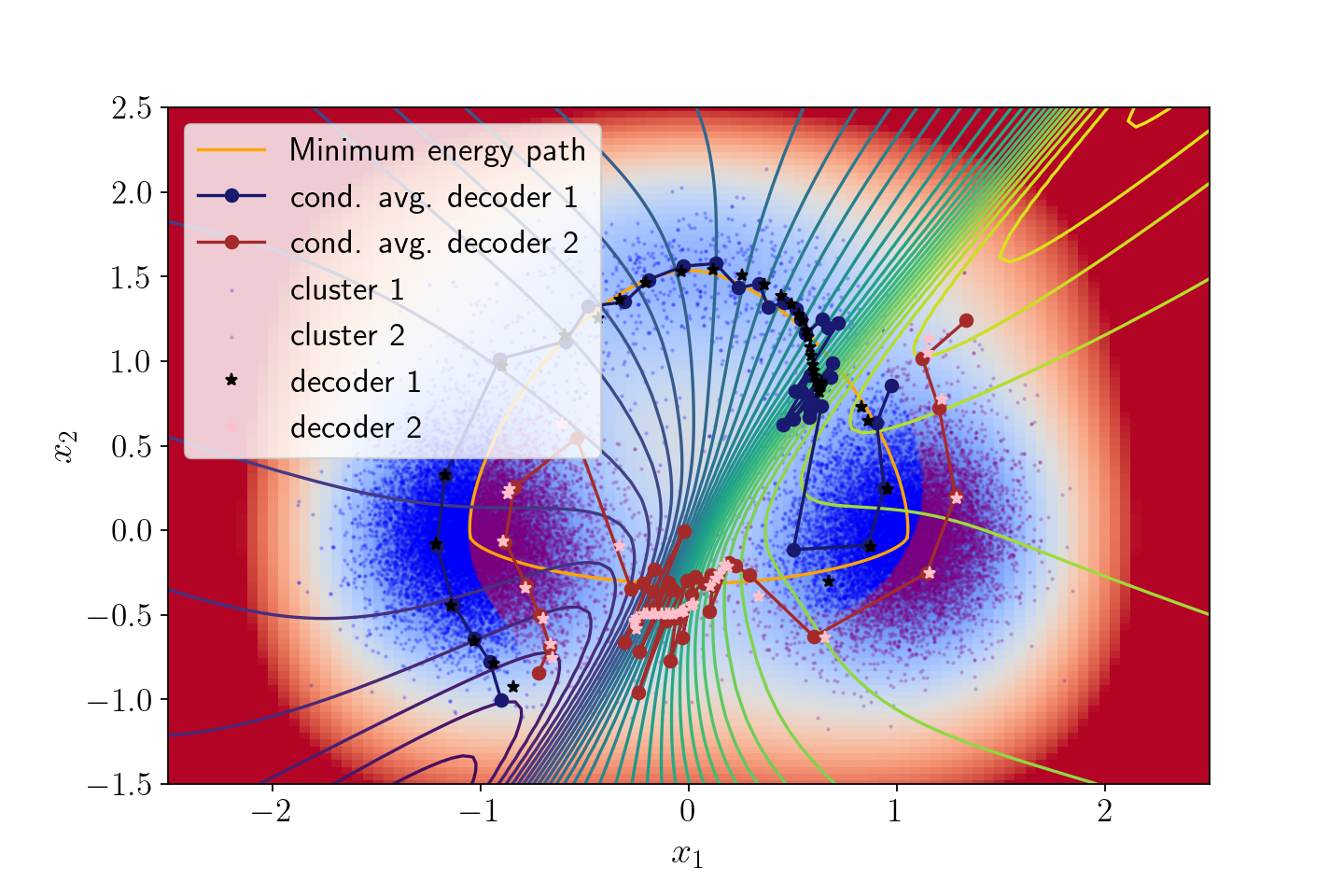}
		\caption{}
            \label{fig:ent_switch_doub_dec_a}
	\end{subfigure}
	\begin{subfigure}[b]{0.49\textwidth}
		\centering
		\includegraphics[width=\textwidth]{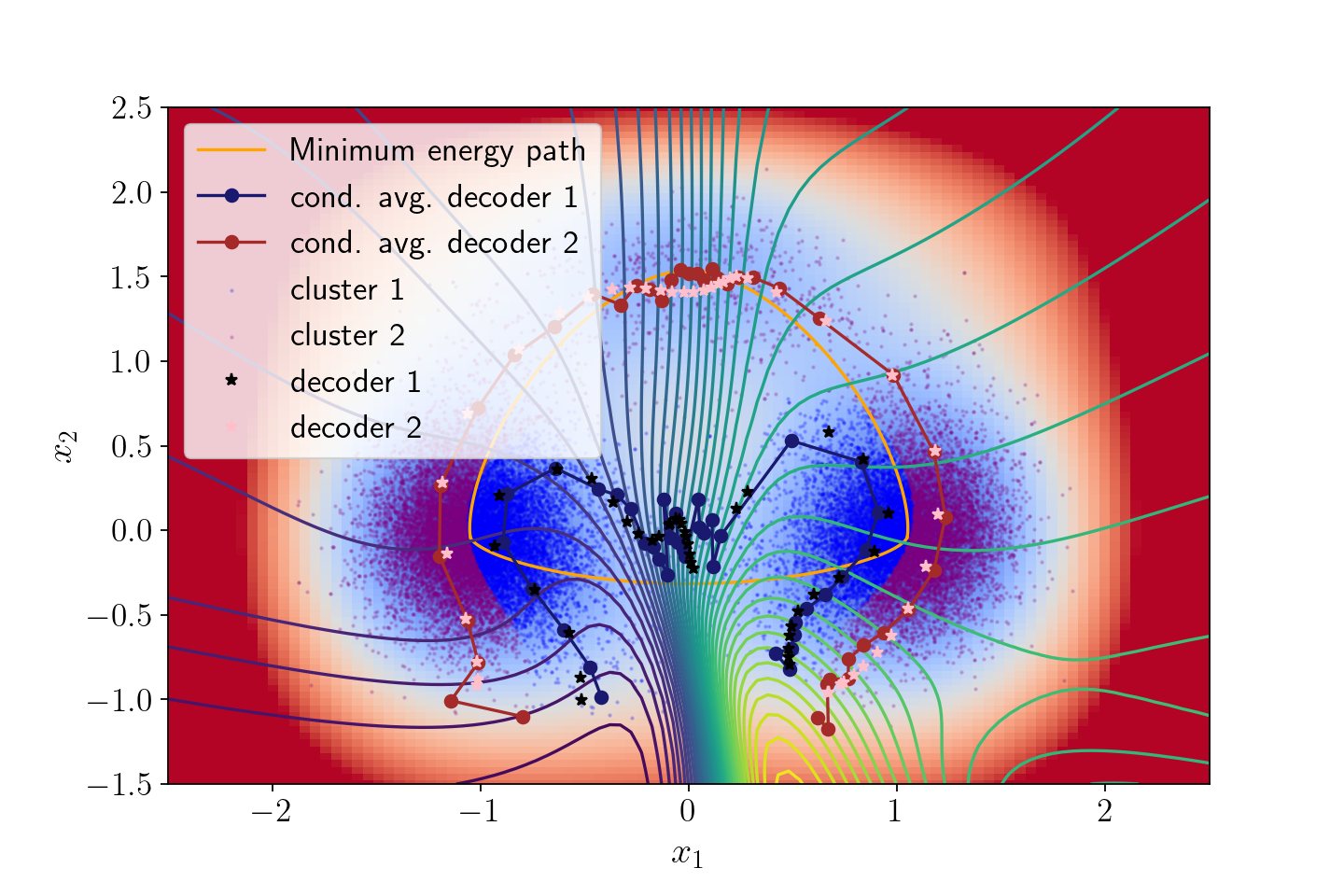}
		\caption{}
            \label{fig:ent_switch_doub_dec_b}
	\end{subfigure}
	\caption{Minimum energy path, conditional averages and decoders on the potential energy heatmap with isolevels of the encoder for models trained to minimize the loss~\eqref{eq:loss_multiple_dec_constrained} with (a) $\lambda_0 = 1$, $\lambda_1 = 1$, $\lambda_2 = 10^{-4}$ and $\lambda_3 = \lambda_4 = 0$; (b) $\lambda_1 = 0$, $\lambda_2 = 10^{-4}$ and $\lambda_3 = \lambda_4 = 0.1$. The dataset points are colored according to the decoders leading to the minimal reconstruction error. 
	} 
	\label{fig:ent_switch_doub_dec}
\end{figure}

\paragraph{Overparametrized setting.}
We finally consider a situation where there are more decoders than transition paths. We tested this on a system for which there is only one transition path, namely the M{\"u}ller--Brown potential (with the same dataset as the one used in Section~\ref{sec:interpretation_numerics}), for which we are looking for~$K=2$ paths corresponding to different decoders. In such a situation, the two paths build for each decoders are almost parallel except in the high energy region (see Figure~\ref{fig:double_dec_mullerbrown}). An analysis of the transition paths therefore reveals that there is in fact a single transition mechanism here. We expect that such an analysis can be similarly carried out in more complicated physical systems in order to detect redundant decoders.

\begin{figure}[!ht]
	\centering
	\includegraphics[width=14.5cm]{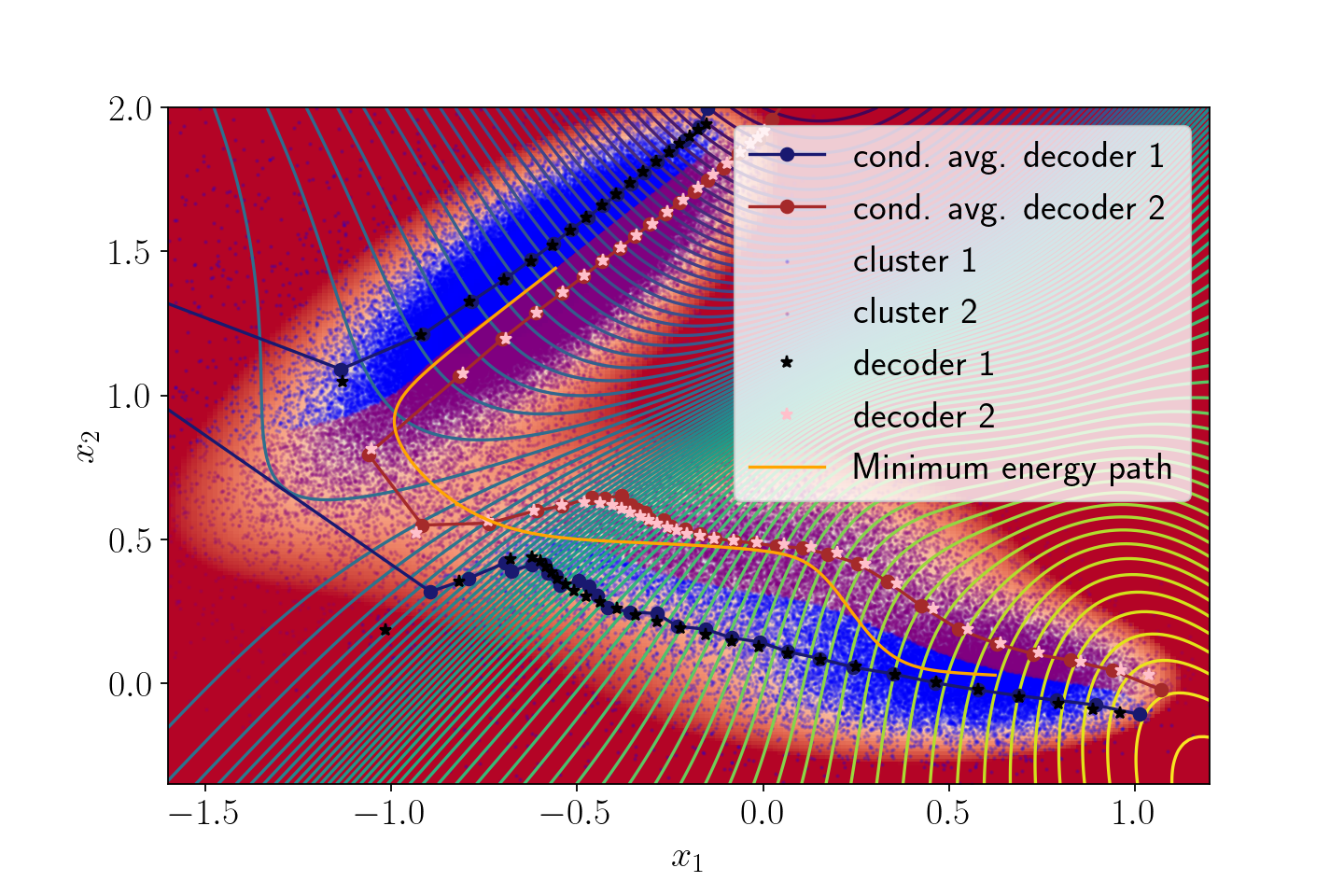}
	\caption{ Minimum energy path, conditional averages and decoders on the potential energy heatmap with isolevels of the encoder $f\mathrm{enc}$. The training dataset points are colored according to which decoder leads to the minimal reconstruction error for models trained to minimize the  loss~\eqref{eq:loss_multiple_dec_constrained} with $\lambda_0=1$, $\lambda_1=1$ and $\lambda_2=10^{-4}$.} 
	\label{fig:double_dec_mullerbrown}
\end{figure}

\subsection{Alanine dipeptide in water}
\label{sec:ad}

We finally study in this section the use of autoencoders on the molecular system alanine dipeptide. We compare the quality of unregularized and regularized autoencoders. For the latter ones, a regularization term is added to the reconstruction loss in order to require autoencoders to parametrize the leading eigenfunctions (corresponding to the largest non-trivial eigenvalues) of the transfer operator of the underlying dynamics; see \eqref{regularized-loss-in-practice} of Section~\ref{sec:parametrzing_eigenfunc_multiple_states}. 

\paragraph{Description of the system.}
Alanine dipeptide is composed of~$22$ atoms, $10$~of which are non-hydrogen atoms. It is well known that the conformations of the system, in particular the metastable states, can be well described by two backbone dihedral angles $\phi$ and $\psi$.
A training data consisting of~$1.5\times 10^5$ configurations sampled along a trajectory was obtained from MD simulations of the molecule solvated in water, using the GROMACS package~\cite{GROMACS}; see Appendix~\ref{app:MD} for further precisions. As Figure~\ref{fig-ad-traj_a} indicates, the metastable conformations of the system are well sampled by the training data.
Figure~\ref{fig-ad-traj_b} shows the conditional variances of the Cartesian coordinates of the~$10$ non-hydrogen atoms, given the values of the two dihedral angles, computed using  trajectory data binned on a uniform grid of size $100\times 100$ according to the values of the dihedral angles. 
One can observe that the conditional variance is small in the metastable region where the $\phi$ angle is negative, but it is larger when the configuration is
near the boundary of the metastable regions or in the metastable region where $\phi$ is positive. In essence, the conditional variance is smaller in regions of the phase-space which have a larger weight under the target Boltzmann--Gibbs measure. The total variance of the trajectory data is $3.988$~nm$^2$, whereas the average of the conditional variances (represented in Figure~\ref{fig-ad-traj_b}) with respect to the empirical distribution of the two dihedral angles (shown in Figure~\ref{fig-ad-traj_a}) is $0.167$~nm$^2$. The latter is an approximation over a grid in latent space of the quantity in~\eqref{eq:recon_error_by_condvar} (and therefore an approximation of the reconstruction error) in the case where the encoder~$\fenc$ is given by the two dihedral angles~$\phi,\psi$ and the decoder is the conditional mean.

\begin{figure}[h!]
	\centering
	\begin{subfigure}{0.32\textwidth}
		\includegraphics[width=1.0\textwidth]{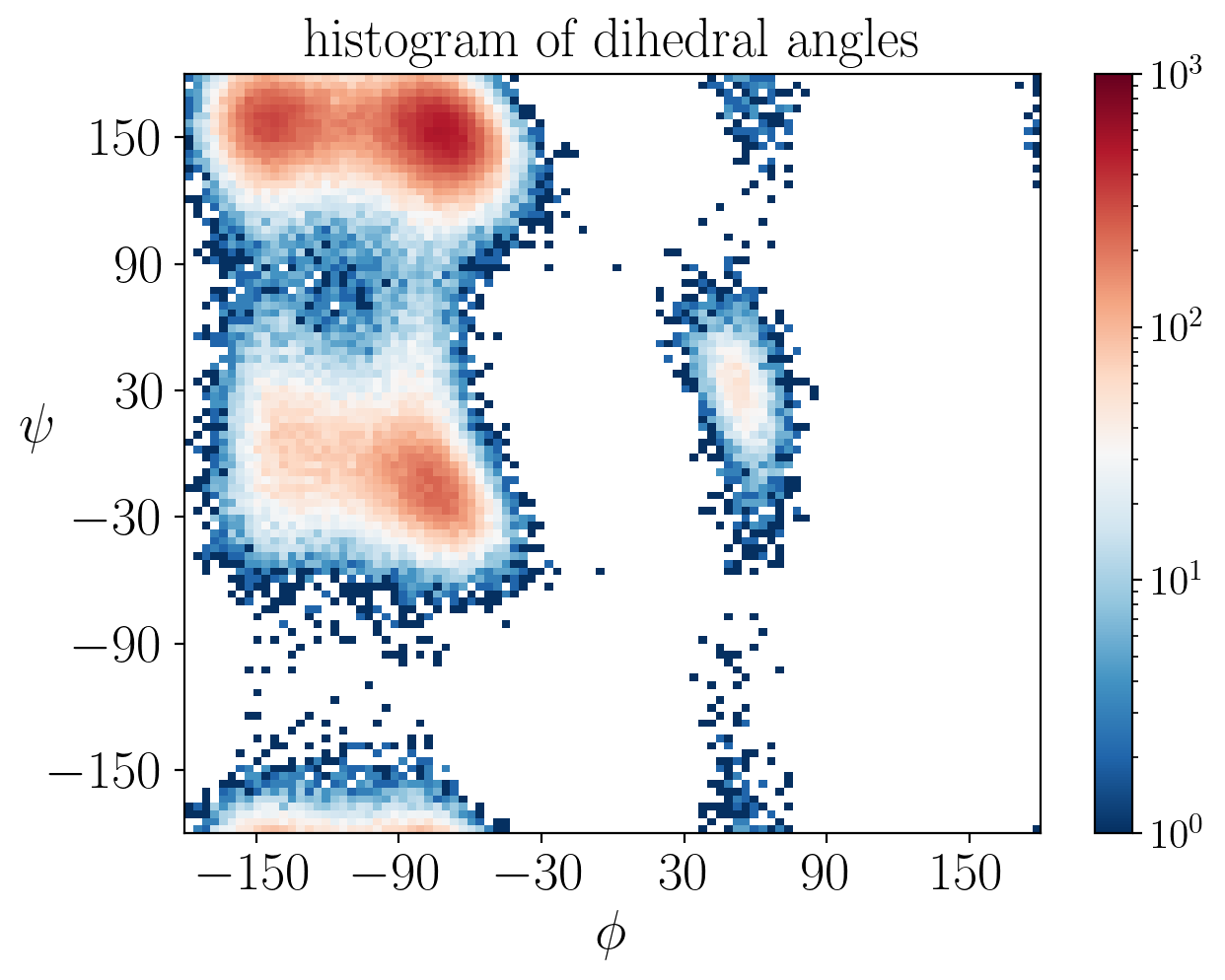}
		\caption{}
            \label{fig-ad-traj_a}
	\end{subfigure}
	\begin{subfigure}{0.32\textwidth}
		\includegraphics[width=1.0\textwidth]{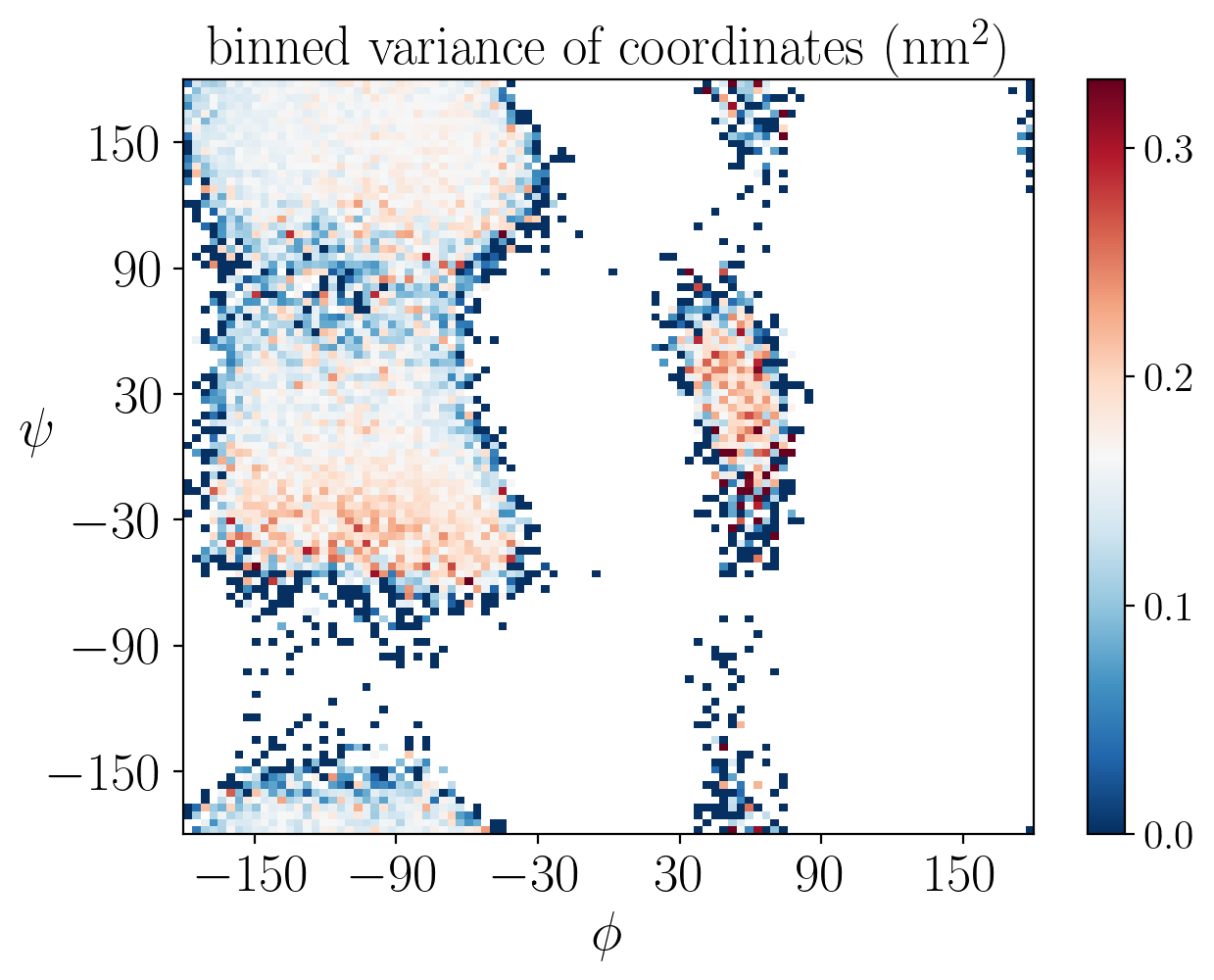}
		\caption{}
            \label{fig-ad-traj_b}
	\end{subfigure}
	\caption{Alanine dipeptide. (a) Histogram of dihedral angles along the $1.5~\mu$s-long trajectory. (b) Conditional variances of the Cartesian coordinates of the~$10$ non-hydrogen atoms given the value of two dihedral angles.}
        \label{fig-ad-traj} 
\end{figure}

\paragraph{Conditional variances for unregularized autoencoders.}
With the data prepared as discussed above, we first train an autoencoder with the standard reconstruction loss (\emph{i.e.}~\eqref{regularized-loss-in-practice} with~$\lambda_0 =1$ and~$\lambda_1=\lambda_2=0$). The dimension of the state space was chosen to be $D=30$, since we are only interested in the configurations of the $10$ non-hydrogen atoms in alanine dipeptide and the coordinate of each atom is in~$\mathbb{R}^3$. The encoded dimension was set to $d=2$. For the encoder, we employed
a neural network that consists of a transformation layer without training parameters (a map from $\mathbb{R}^{30}$ to $\mathbb{R}^{30}$ that aligns the configurations by minimizing the root mean squared deviation using Kabsch algorithm~\cite{Kabsch}; see Ref.~\citenum{ZLS22} for implementation details and also Ref.~\citenum{cv-using-invariant-rep} for an alternative approach) in order to guarantee both translation and rotation invariance, and a feedforward neural network of size $(30, 20, 15, 10, 2)$, \emph{i.e.}\ an input layer of size $30$, an output layer of size $2$, and three hidden layers of sizes~$20, 15, 10$, respectively. For the decoder, we employed a feedforward neural network of size~$(2, 10, 15, 20, 30)$. We used \textrm{tanh} as activation functions in both encoder and decoder (except in their output layers where there is no activation function). To train the neural network, the dataset of size $1.5 \times 10^5$ was split into a training set and a test set with ratio of sizes~$4:1$. We employed the Adam optimizer~\cite{Adam} in PyTorch with a (large) batch size $20,000$ and learning rate $2\times 10^{-3}$. The random seed was fixed in the training experiments below.

\begin{figure}[h!]
	\centering
	\begin{subfigure}{0.32\textwidth}
		\includegraphics[width=1.0\textwidth]{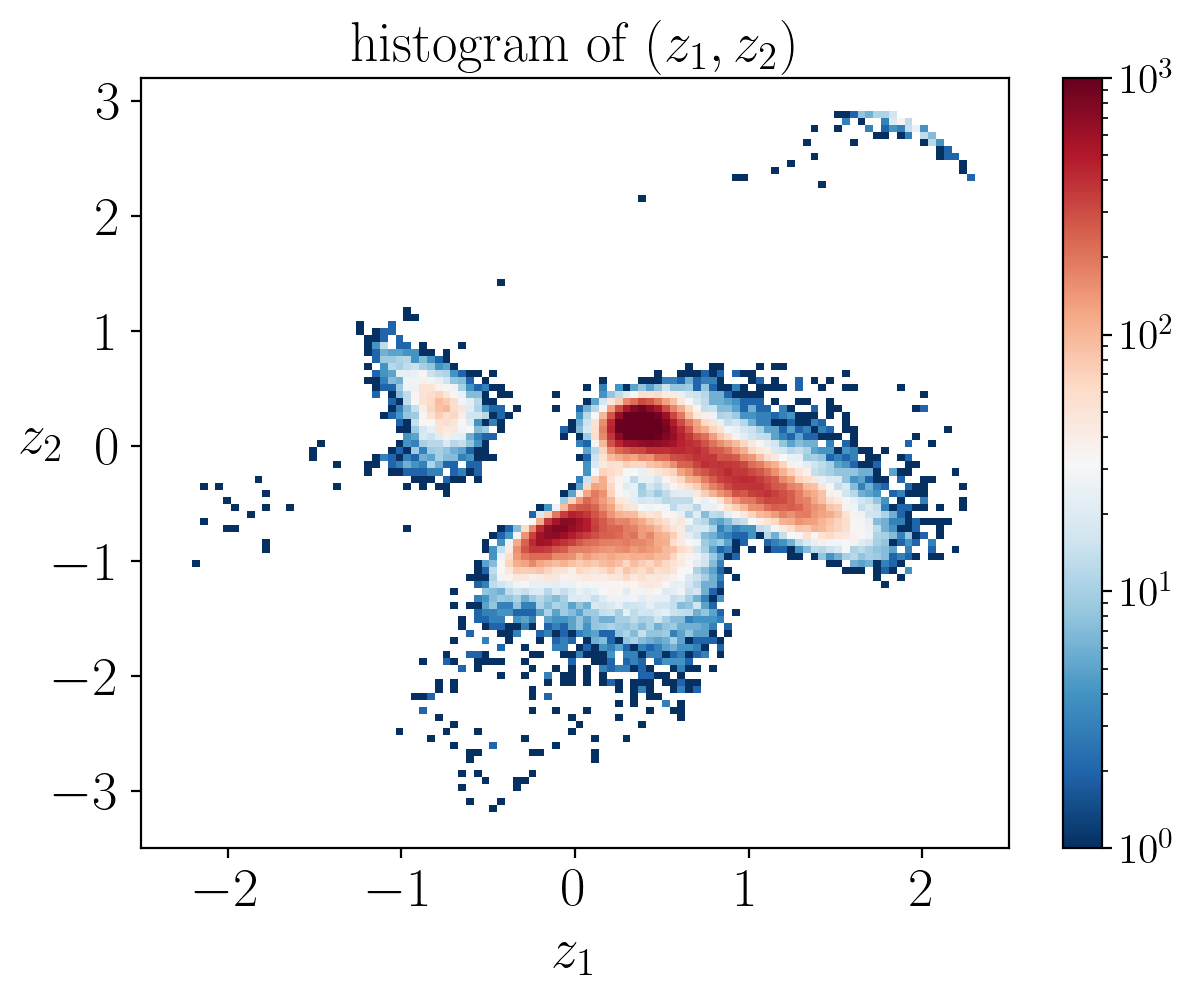}
		\caption{}
            \label{fig-ad-standard-ae_a}
	\end{subfigure}
	\begin{subfigure}{0.32\textwidth}
		\includegraphics[width=1.0\textwidth]{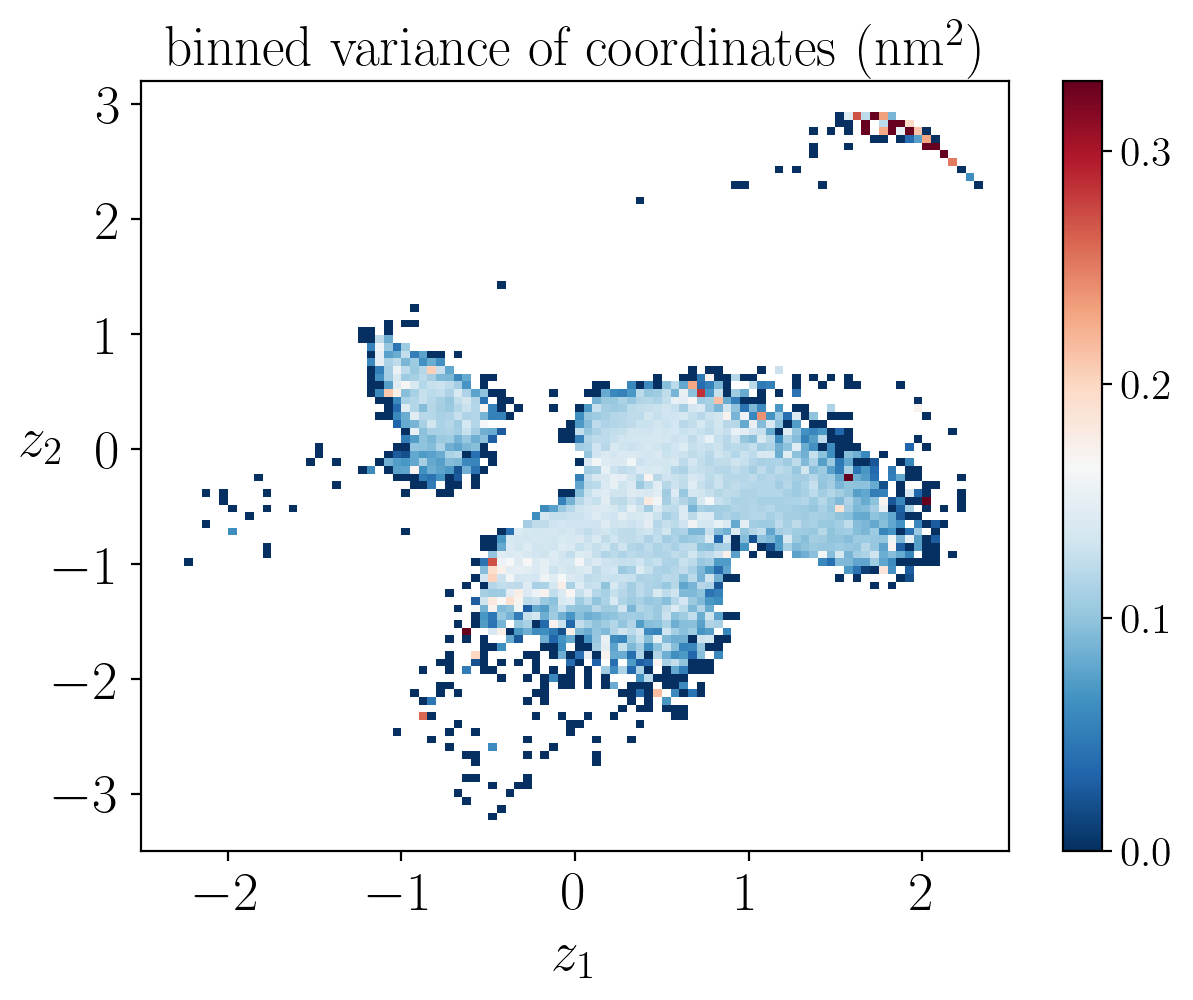}
		\caption{}
            \label{fig-ad-standard-ae_b}
	\end{subfigure}
	\begin{subfigure}{0.32\textwidth}
		\includegraphics[width=1.0\textwidth]{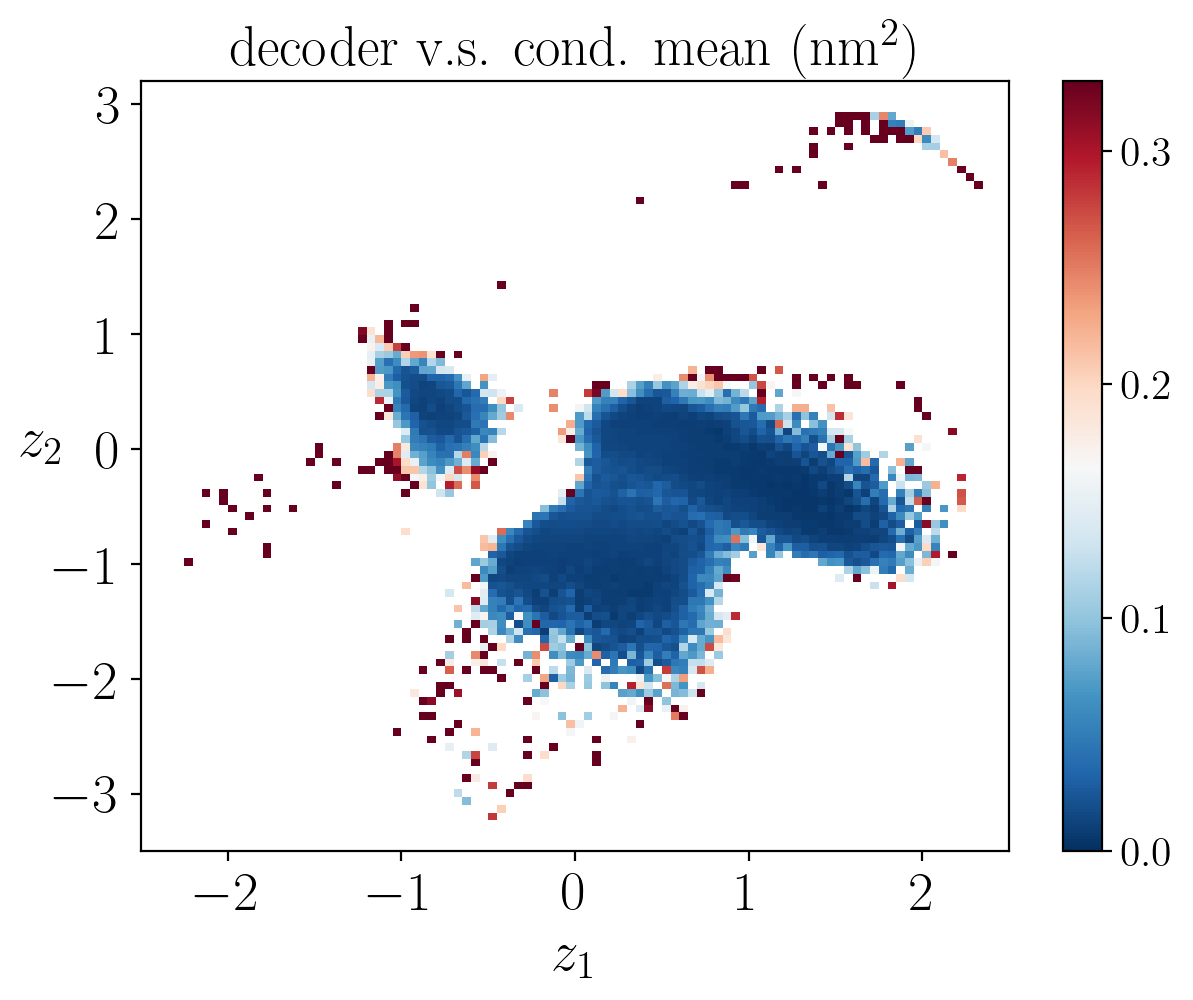}
		\caption{}
            \label{fig-ad-standard-ae_c}
	\end{subfigure}
	\caption{Autoencoder learned by training the standard autoencoder on alanine dipeptide with reconstruction loss. (a) Histogram of the encoded values along the $1.5~\mu$s-long trajectory.
		(b) Conditional variances of coordinates of the $10$ non-hydrogen atoms given
		the encoded values. (c) Mean squared differences between the decoded coordinates (by the autoencoder) and the conditional mean coordinates
		of the $10$ non-hydrogen atoms, given the encoded values.}
        \label{fig-ad-standard-ae}
\end{figure}

Figure~\ref{fig-ad-standard-ae} presents the results of the trained autoencoder after~$2,000$ training epochs. The histogram in Figure~\ref{fig-ad-standard-ae_a} shows that, similarly to the dihedral angles, the learned encoder is able to map the system's metastable conformations into different clusters in latent space. Figure~\ref{fig-ad-standard-ae_b} shows the conditional variances of the $10$ non-hydrogen atoms given the encoded values, computed using trajectory data binned on a uniform grid of size $100\times 100$ according to the encoded values. A comparison between Figure~\ref{fig-ad-traj_b} and Figure~\ref{fig-ad-standard-ae_b} reveals that the learned encoder indeed outperforms dihedral angles in reducing the conditional variances. In fact, as shown in Table~\ref{tab-ad-compare-cvs}, the average conditional variance (with respect to the empirical distribution of the encoder shown in Figure~\ref{fig-ad-standard-ae_a} is~$0.126$~nm$^2$, which is smaller than the average conditional variance in the case of dihedral angles. Figure~\ref{fig-ad-standard-ae_c} shows the mean square difference between the decoded coordinates and the conditional average coordinates of the $10$ non-hydrogen atoms given the encoded values. It can be seen that the decoded coordinates are close to the conditional averages in the metastable regions, while the difference is larger in certain low-density regions. 

\paragraph{Regularized autoencoders.}
We now consider the training loss~\eqref{regularized-loss-in-practice}. 
In our numerical experiments, we chose $K=3$, $(\omega_1,\omega_2,\omega_3)=(1.0, 0.4. 0.2)$, $(\lambda_0, \lambda_1, \lambda_2)=(1.0, 0.2, 10.0)$, and a lag-time $\tau=10\,\mathrm{ps}$. We adopted the same encoder, decoder and training parameters as for the unregularized networks considered later on in this section. For the regularizers, we employed three neural networks with identical sizes~$(2, 20, 20, 20, 1)$ in order to represent the three leading eigenfunctions.

Figure~\ref{fig-ad-reg-ae} shows the jointly learned encoder $f_{\mathrm{enc},\theta_1}$ and the three leading eigenfunctions~$\varphi_1, \varphi_2, \varphi_3$. As one can see in Figures~\ref{fig-ad-reg-ae_a}-\ref{fig-ad-reg-ae_b} and Table~\ref{tab-ad-compare-cvs}, similarly to the encoder trained with the standard reconstruction loss, the encoder trained with the loss~\eqref{regularized-loss-in-practice} is able to achieve dimensionality reduction with a smaller mean conditional variance (\emph{i.e.}\ $0.134$ nm$^2$) then the one obtained with dihedral angles.
Figure~\ref{fig-ad-reg-ae_c} shows the mean square difference between the decoded coordinates and the conditional mean coordinates of the $10$ non-hydrogen atoms given the encoded values. Similarly to the case of the standard reconstruction loss (Figure~\ref{fig-ad-standard-ae_c}),  the decoded coordinates are close to the conditional averages in the metastable regions, while the difference is visible in certain low--density regions. Figures~\ref{fig-ad-reg-ae_d} to~\ref{fig-ad-reg-ae_i} show the three leading eigenfunctions jointly learned in the training. Each of these eigenfunctions reveals information about certain transition events of the system on a large time-scale (related to the eigenvalues in Table~\ref{tab-ad-compare-cvs}, where larger eigenvalues correspond to transition events on longer time-scales\cite{msm_generation}). The eigenfunctions are close to being constant inside metastable regions, and exhibit a sign change in the transition regions between metastable states.
From these results, we see that 
training with the loss~\eqref{regularized-loss-in-practice} indeed yields an encoder that is capable of reconstructing conformations and at the same time parametrizing the first three eigenfunctions. We have also performed experiments with $K=1$ and $K=2$ in the loss~\eqref{regularized-loss-in-practice}, for which we have obtained similar results in terms of the encoder's quality for both conformation reconstruction and eigenfunction parametrization.

To summarize, training with the loss~\eqref{regularized-loss-in-practice} yields for this example results seemingly similar to the ones obtained with the standard reconstruction loss. This may be due to the simplicity of the alanine dipeptide system, for which the autoencoder trained to minimize the reconstruction error is already sufficient to characterize the dynamics, \emph{i.e.} the leading eigenfunctions. Nevertheless, the numerical experiments validate the capability of the training algorithm with the loss~\eqref{regularized-loss-in-practice}. We expect that this can provide more prominent regularizing effects in training autoencoders on MD systems with more complex kinetics. 

\begin{table}[htp]
  \begin{tabular}{lcccc}
    \hline 
    CV map & cond. var. & $\nu_1$ & $\nu_2$ & $\nu_3$ \\
    \hline 
    \hline 
    dihedrals   & $0.167$ & $0.99$ & $0.88$ & $0.74$ \\
    \hline
    standard AE & $0.126$ & $0.99$ & $0.88$ & $0.74$ \\
    \hline
    reg.\ AE& $0.134$ & $0.99$ & $0.88$ & $0.76$ \\
    \hline
  \end{tabular}
  \centering
  \caption{Comparison of different CV maps for alanine dipeptide. The CV maps are built using either the two dihedral angles (``dihedrals''), the encoder learned using the standard reconstruction loss (``standard AE''), and the encoder learned using the loss \eqref{regularized-loss-in-practice} (``reg.\ AE''). Column ``cond.  var.'' records the average conditional variances (unit: $\mathrm{nm}^2$) computed using different CV maps. Each row of the columns ``$\nu_1$'', ``$\nu_2$'' and ``$\nu_3$'' records the estimations of the first three non-trivial eigenvalues of the transfer operator with $\tau=10$ ps (see \eqref{eqn:raylaigh-ritz-quotient} and the follow-up discussion there). For the CV map ``reg. AE'', these eigenvalues and the eigenfunctions are estimated jointly during the training with the loss~\eqref{regularized-loss-in-practice}, whereas for the CV maps ``dihedrals'' and ``standard AE'' they are estimated by performing a subsequent training with the loss~\eqref{regularized-loss-in-practice} where $\lambda_0=0$ and the corresponding CV map is fixed.}
  \label{tab-ad-compare-cvs}
\end{table}

\begin{figure}[h!]
  \begin{subfigure}{0.32\textwidth}
    \includegraphics[width=1.0\textwidth]{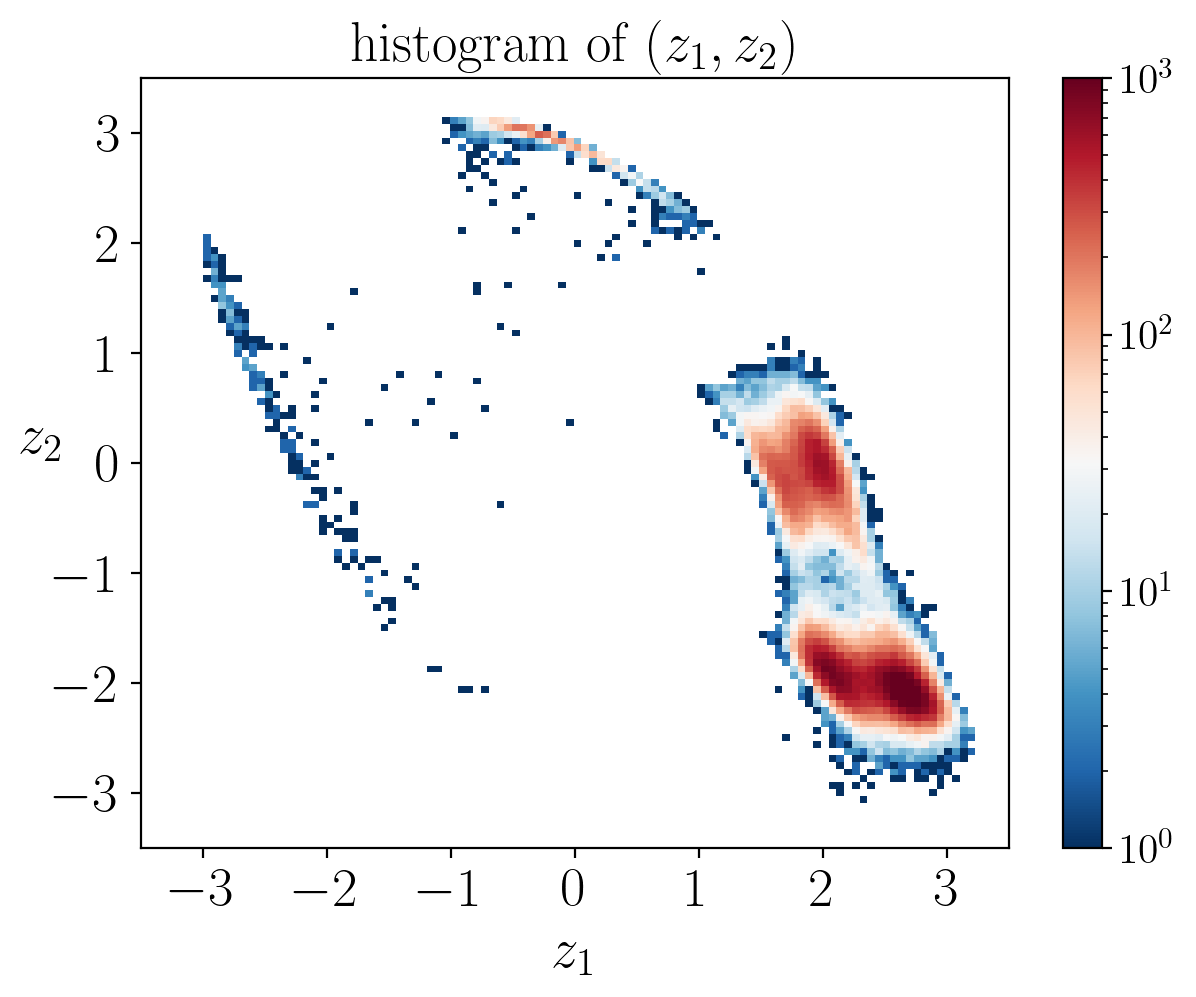}
    \caption{}
    \label{fig-ad-reg-ae_a}
  \end{subfigure}
  \begin{subfigure}{0.32\textwidth}
    \includegraphics[width=1.0\textwidth]{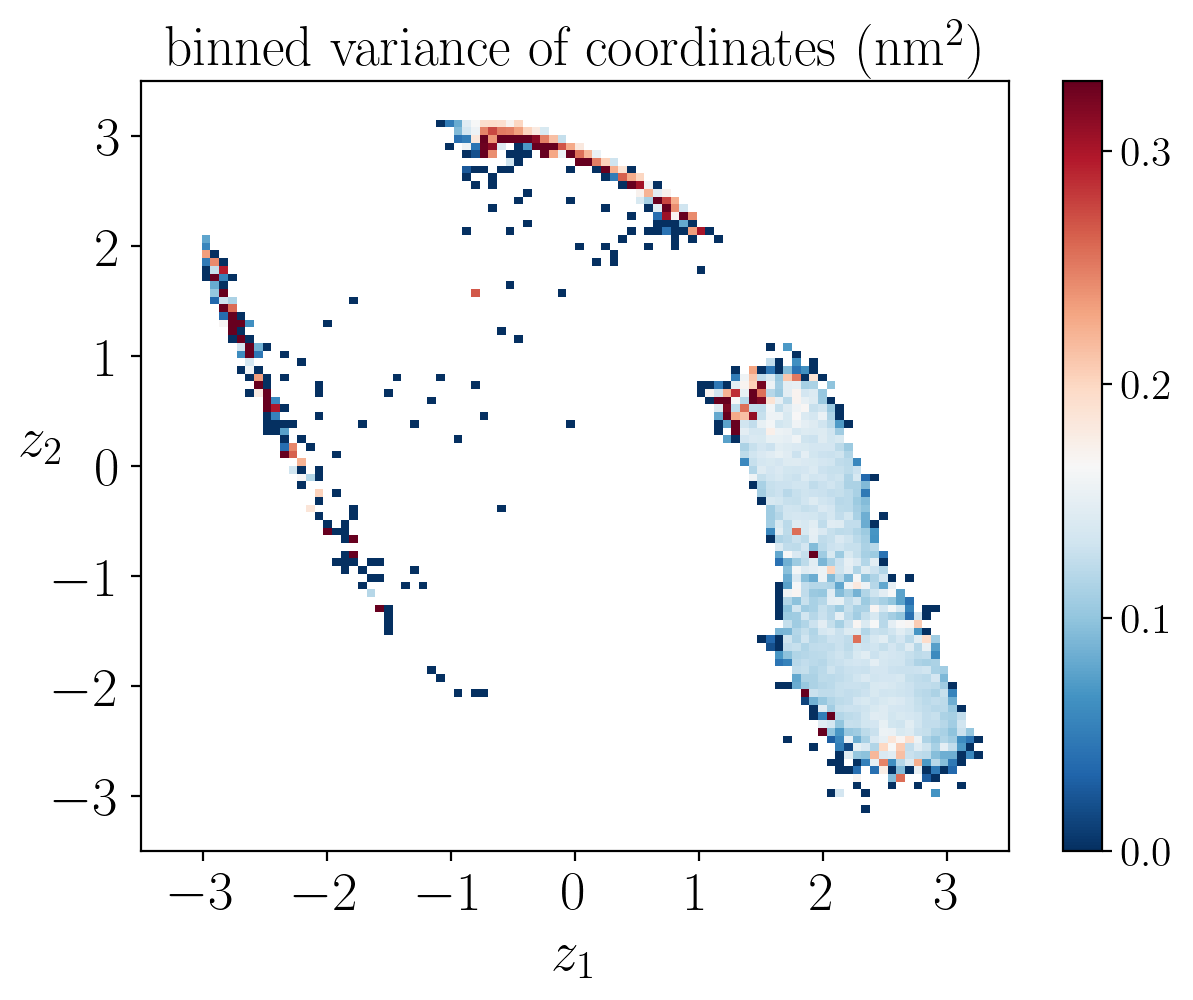}
    \caption{}
    \label{fig-ad-reg-ae_b}
  \end{subfigure}
  \begin{subfigure}{0.32\textwidth}
    \includegraphics[width=1.0\textwidth]{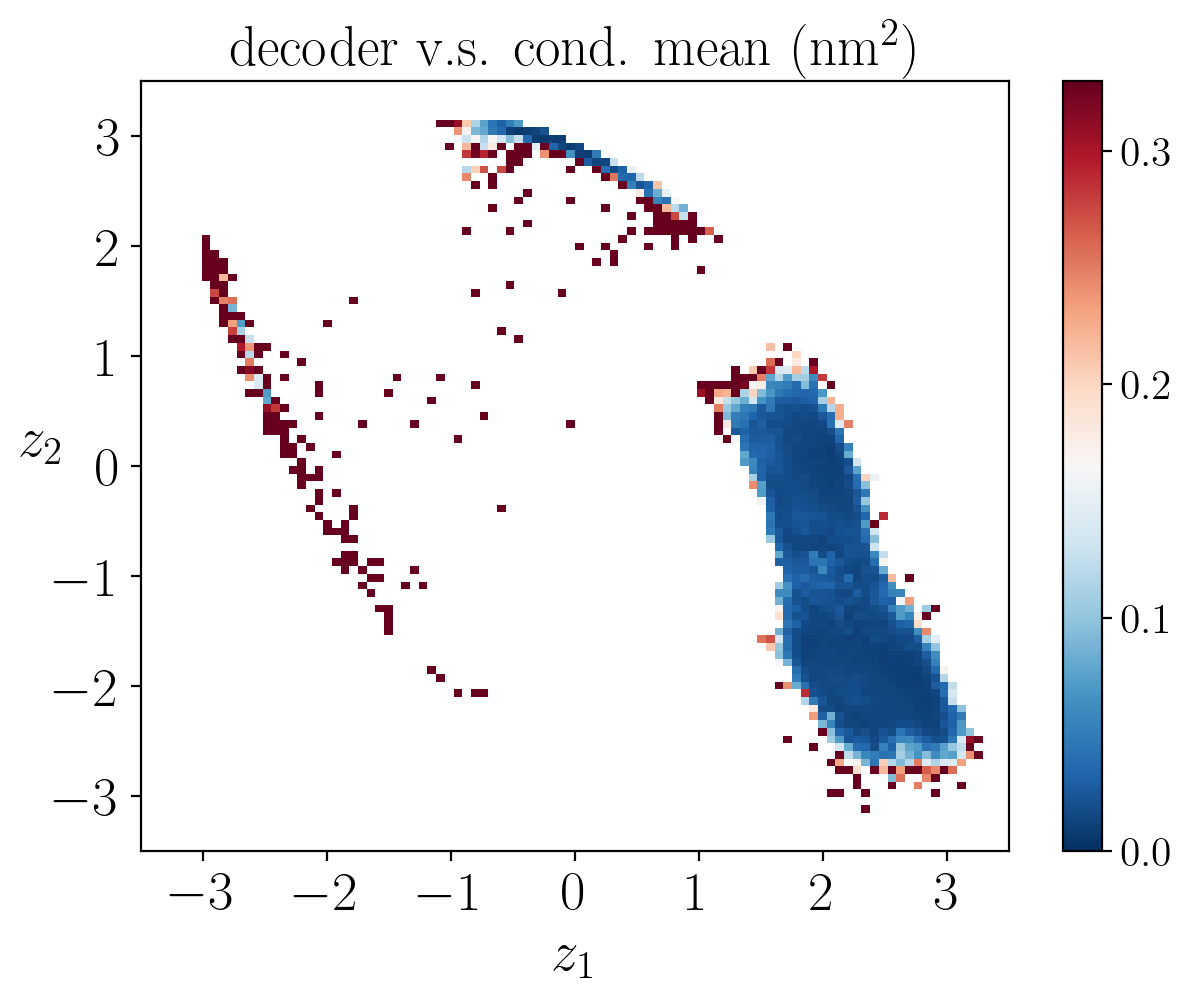}
    \caption{}
    \label{fig-ad-reg-ae_c}
  \end{subfigure}  
  \begin{subfigure}{0.32\textwidth}
    \includegraphics[width=0.95\textwidth]{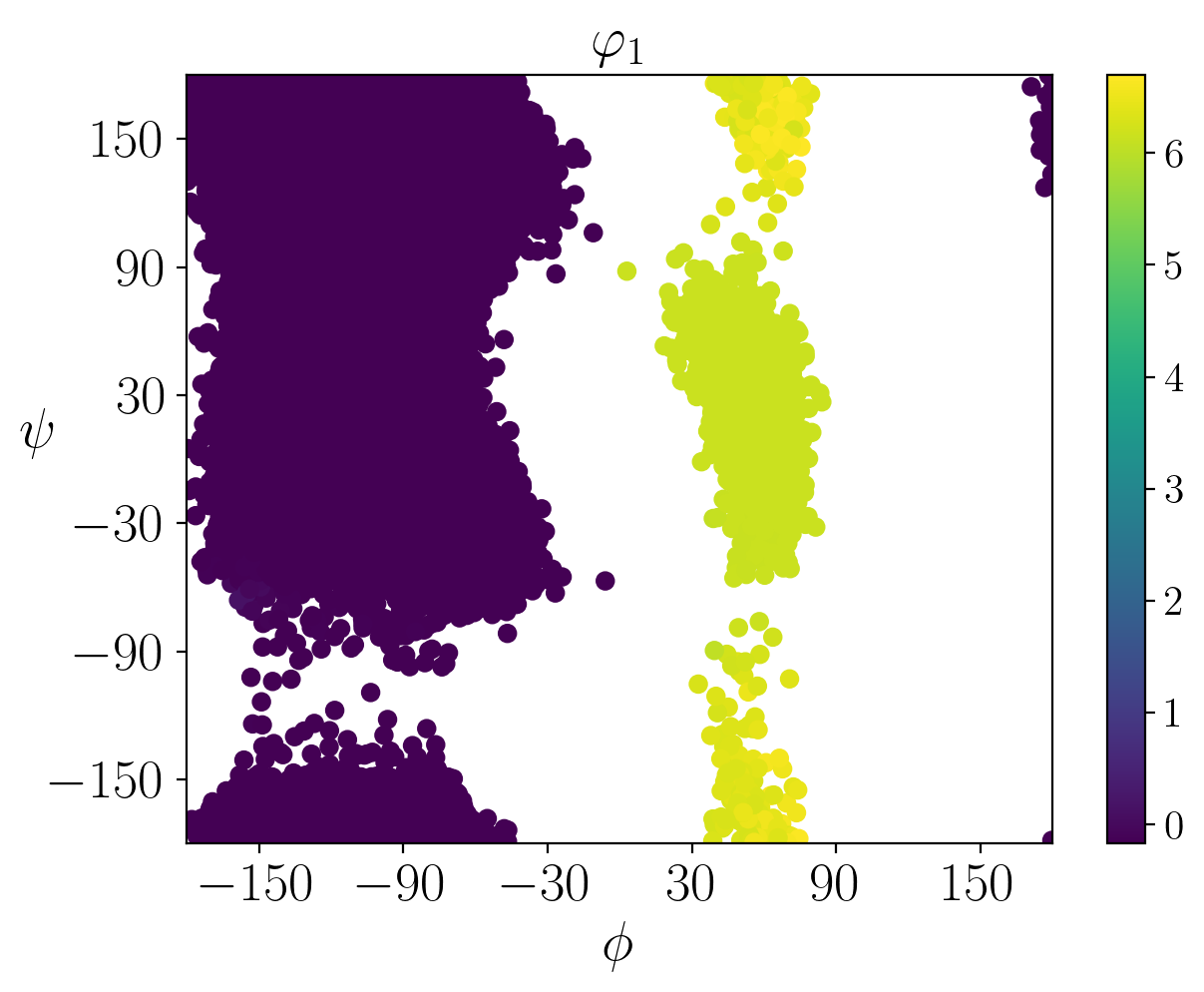}
    \caption{}
    \label{fig-ad-reg-ae_d}
  \end{subfigure}
  \begin{subfigure}{0.32\textwidth}
    \includegraphics[width=1.0\textwidth]{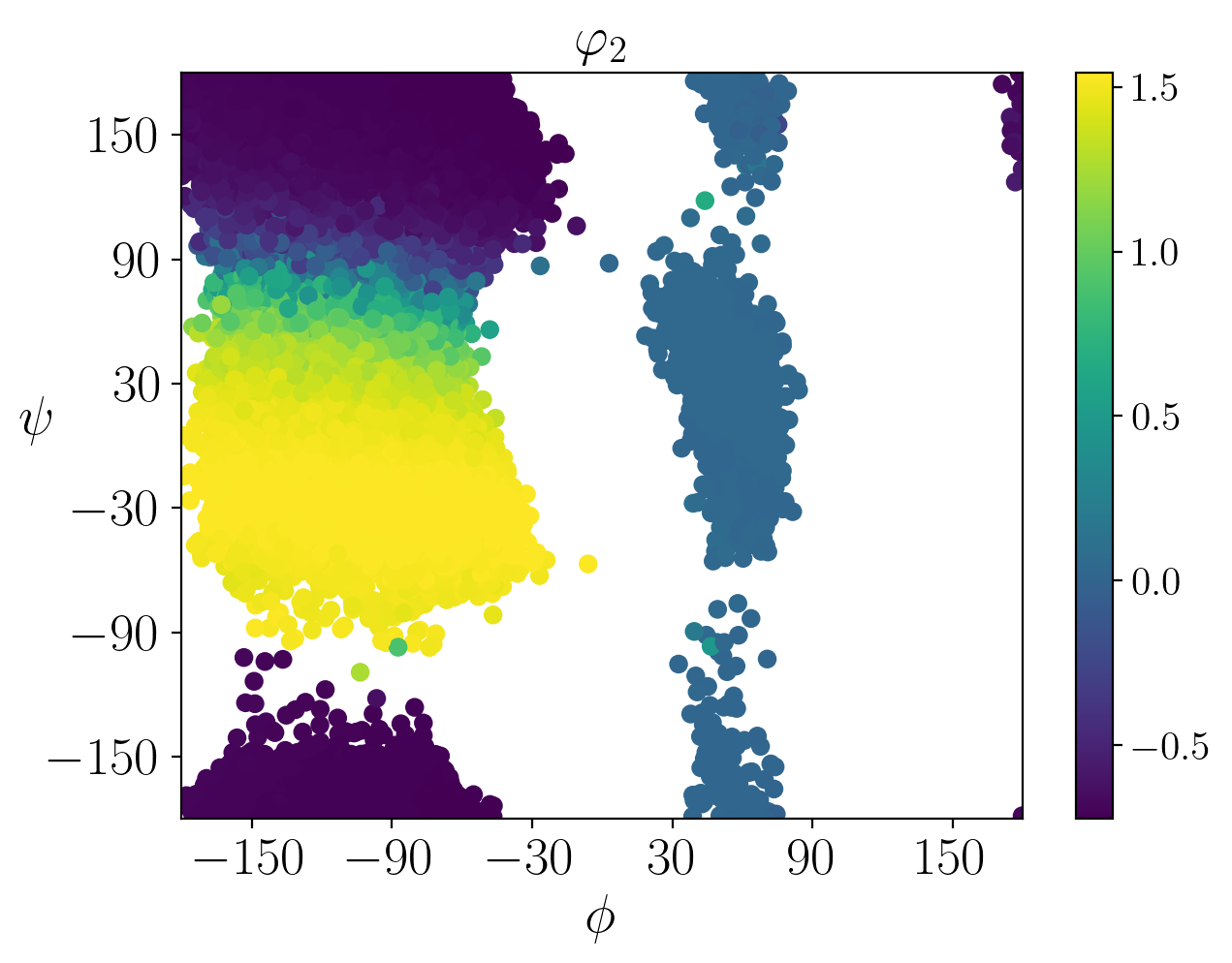}
    \caption{}
  \end{subfigure}
  \begin{subfigure}{0.32\textwidth}
    \includegraphics[width=1.0\textwidth]{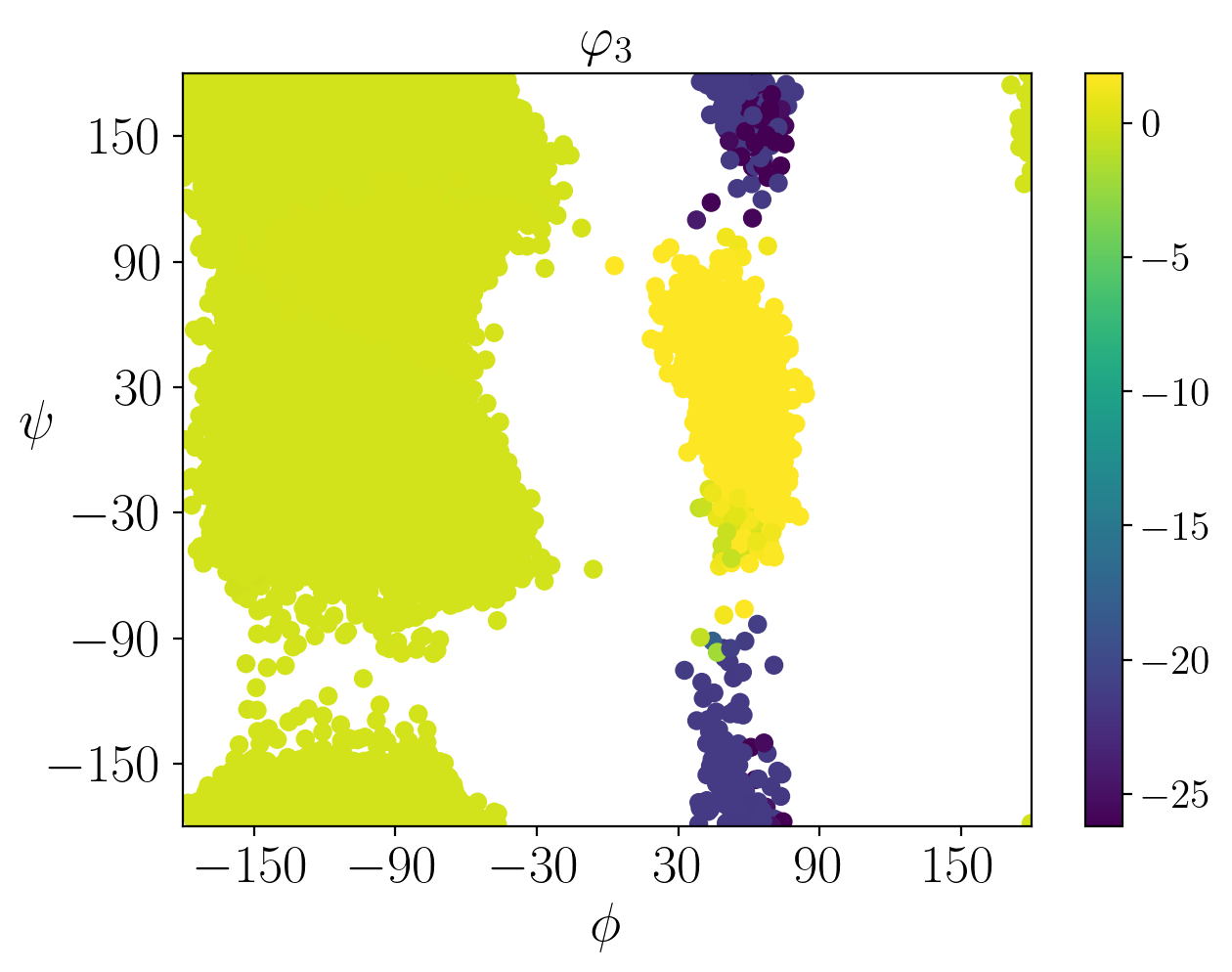}
    \caption{}
  \end{subfigure}
  \begin{subfigure}{0.32\textwidth}
    \includegraphics[width=0.95\textwidth]{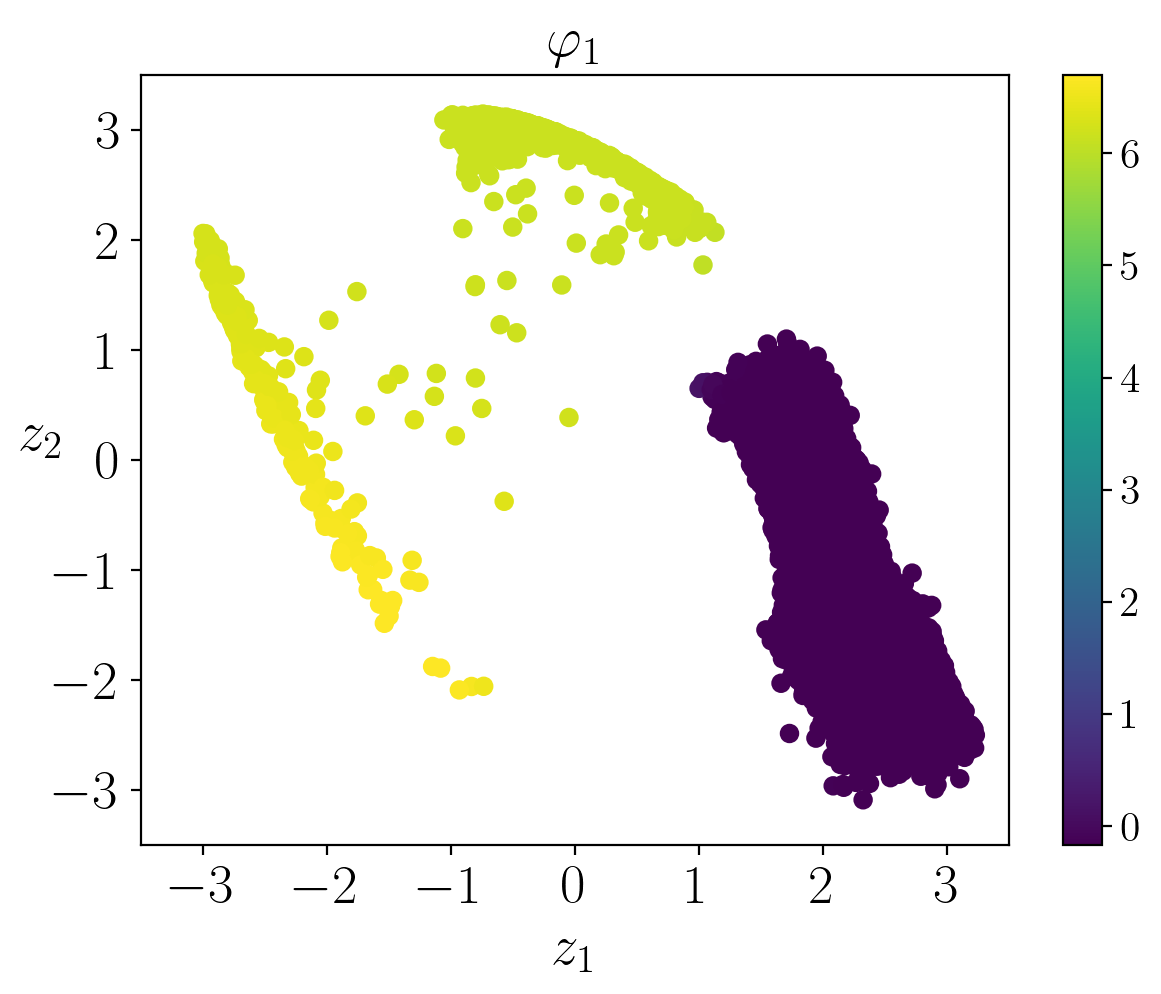}
    \caption{}
  \end{subfigure}
  \begin{subfigure}{0.32\textwidth}
    \includegraphics[width=1.0\textwidth]{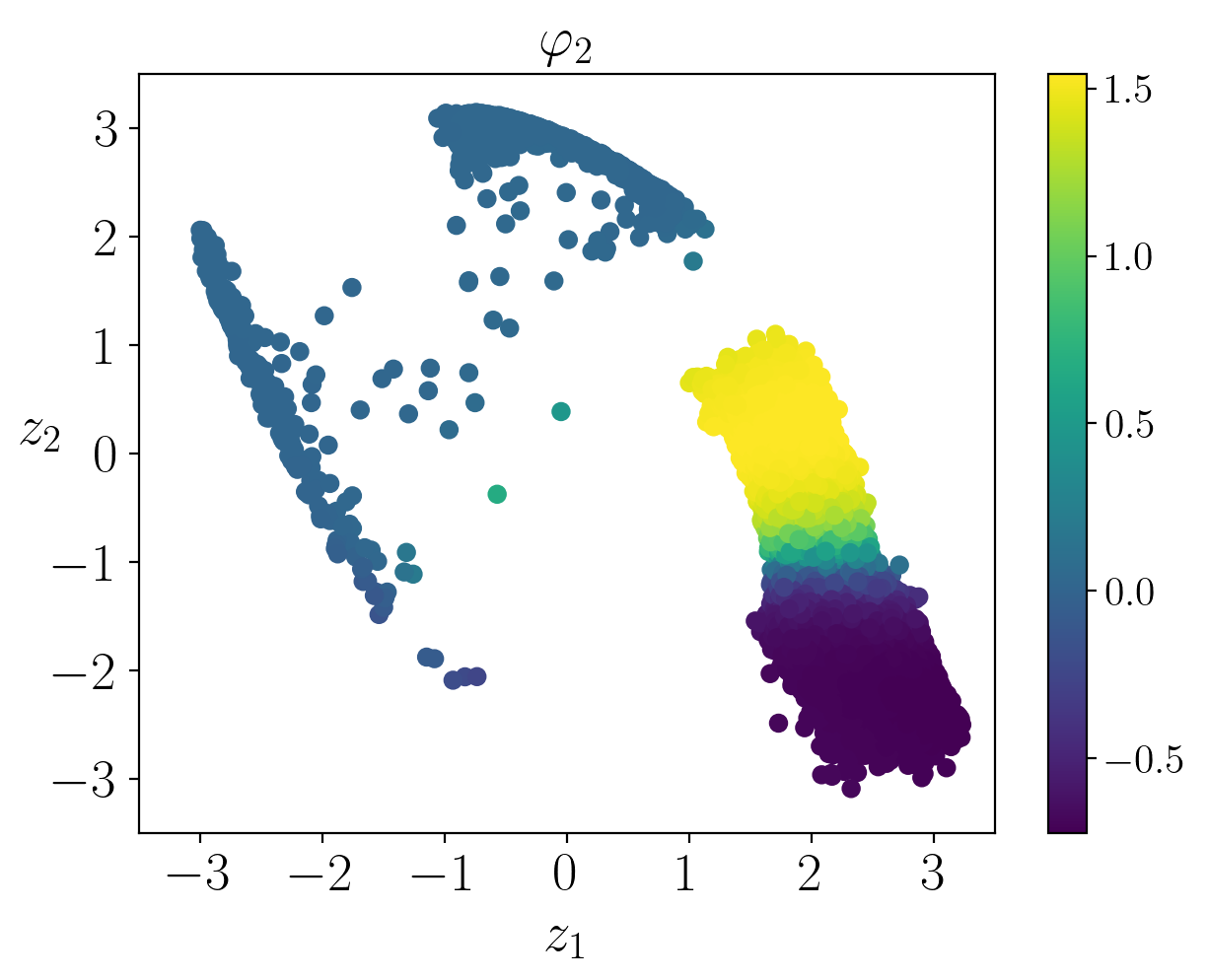}
    \caption{}
  \end{subfigure}
  \begin{subfigure}{0.32\textwidth}
    \includegraphics[width=1.0\textwidth]{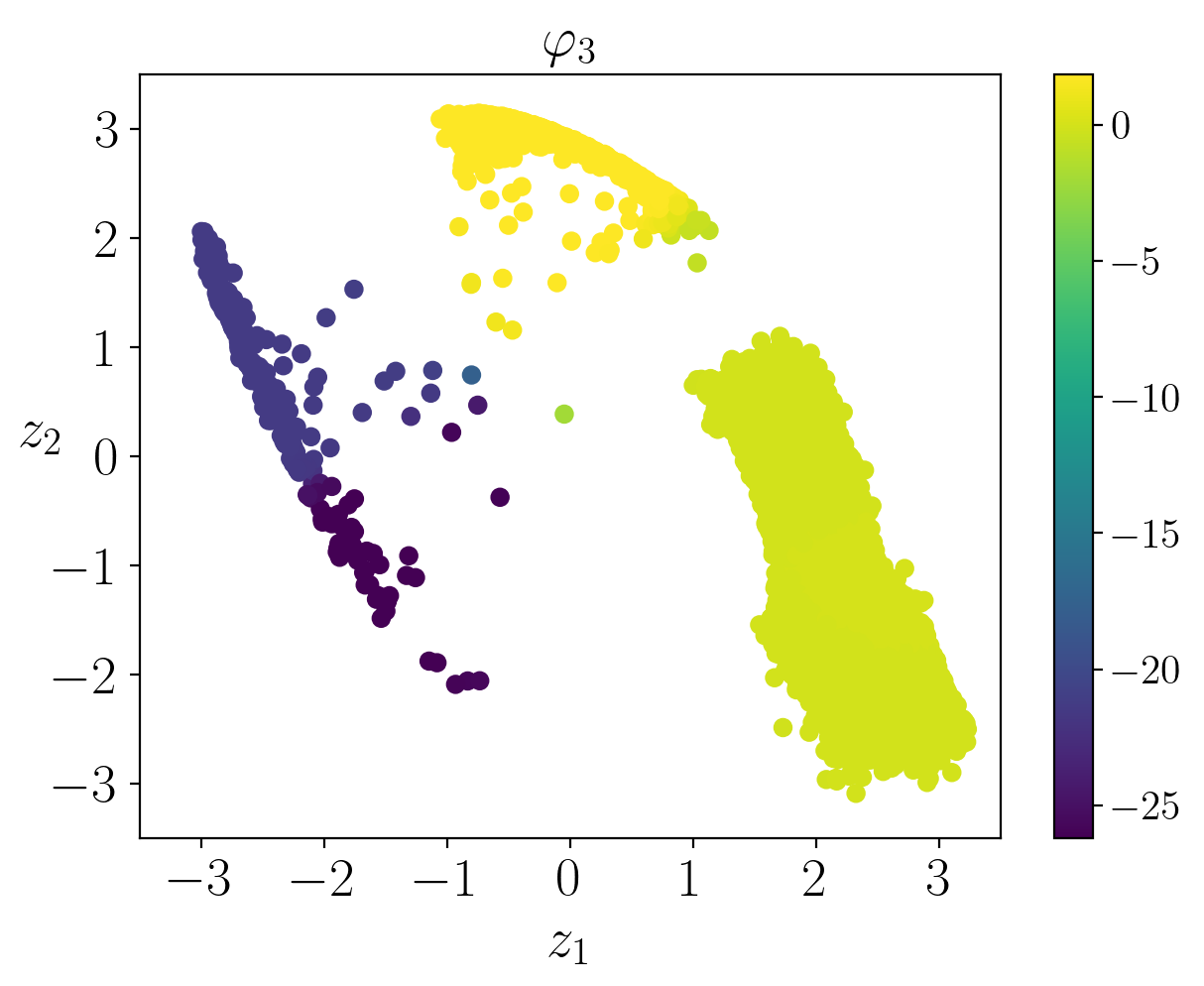}
    \caption{}
    \label{fig-ad-reg-ae_i}
  \end{subfigure}
  \centering
  \caption{Autoencoder and first three eigenfunctions learned by training the autoencoder (see Figure~\ref{fig-reg-ae}) with the loss~\eqref{regularized-loss-in-practice}. (a) Histogram of the encoded values along the $1.5~\mu$s-long trajectory. (b) Conditional variances of coordinates of the $10$ non-hydrogen atoms given the encoded values. (c) Mean squared differences between the decoded coordinates and the conditional mean coordinates of the 10 non-hydrogen atoms, given the encoded values. (d)-(f): Scatter plots of the first three eigenfunctions~$\varphi_1$, $\varphi_2$, and $\varphi_3$ (as compositions of the learned encoder and regularizer) evaluated on the trajectory data in the space of dihedral angles. (g)-(i): Scatter plots of the first three eigenfunctions evaluated on the trajectory data in the space of the learned encoder.}
  \label{fig-ad-reg-ae}
\end{figure}

\section{Conclusions and perspectives}
\label{sec:perspectives}

We have shown in this work that the usual setting of autoencoder training in terms of the reconstruction error can be reformulated in various ways. This allows to relate the problem to the well-studied class of techniques for computing principal curves and manifolds. Somewhat more interestingly in our view, it also allows to interpret the decoder as some approximation of the Bayes predictor associated with a given encoder, which corresponds to a conditional expectation. The quality of the decoding can be assessed by comparing actual conditional expectations and the output of the decoder for a given value of the latent variable corresponding to the value of the encoder. We also suggested various extensions beyond the usual training of autoencoders to capture some information on transitions from one metastable state to another, by (i) emphasizing transition points through some modified distribution of training points, (ii) allowing for multiple decoders to parametrize multiple transition paths, and (iii) adding extra terms to the loss function to take into account dynamical properties.

There are various natural follow-ups of this work. A first one is to combine numerical methods to sample reactive trajectories from one metastable state to another and the training of autoencoders. One instance of such an idea is to start from some (partial) data set of configurations, train an autoencoder with a one-dimensional bottleneck, and then use the corresponding encoder as a one-dimensional reaction coordinate in the adaptive multilevel splitting algorithm~\cite{CG07,LL19,Pigeon_al_2023} to sample reactive trajectories, in order to complement the initial data set, as done in Sections~\ref{sec:changing_ref_measure} and~\ref{sec:numerics_modifying_proba_dist}. A second follow-up is to extend our analysis to slow CVs, \emph{i.e.} genuinely consider dynamical aspects from the start. This can be done by working with time-lagged autoencoders~\cite{WN18,CSF19_TAE}, relying on some data set where couples of configurations separated by a fixed time lag are considered; see Ref.~\citenum{zhang2023understanding} for elements in this direction. An alternative option is to search for CVs satisfying some optimality conditions in terms of their effective dynamics~\cite{LL09,ZHS16}. Finally, let us recall Remark~\ref{rmk:measure_is_known}, which highlights that the data distributions at hand in molecular dynamics are not unknown, in contrast to the usual machine learning setting. This extra information can probably be used to improve off-the-shelf machine learning methods.

\begin{acknowledgement}
	
The work of T.L and G.S. was funded in part by the European Research Council (ERC) under the European Union's Horizon 2020 research and innovation programme (grant agreement No. 810367). 
The work of W.Z. was funded by the Deutsche Forschungsgemeinschaft (DFG, German
Research Foundation) through grant CRC 1114 ``Scaling Cascades in Complex Systems'' (project No. 235221301) and DFG Eigene Stelle (project No. 524086759).

\end{acknowledgement}

\appendix

\section{Derivation of~\eqref{eq:formal_EL_condition_optimal_fenc}}
\label{sec:derivation_formal_EL_condition_optimal_fenc}

We formally derive in this appendix the necessary condition~\eqref{eq:formal_EL_condition_optimal_fenc} satisfied by (local) maxima of~\eqref{eq:equivalent_maximization_on_fenc}, for~$\cZ \subset \mathbb{R}^d$ and~$\cX \subset \mathbb{R}^D$. We assume that~$\mu$ has a smooth density with respect to the Lebesgue measure, still denoted by~$\mu$ with some abuse of notation. We denote by~$\mathrm{Supp}(\mu)$ the support of this measure, namely the closure of the set of points at which the density is positive. A key assumption in our derivation is that the encoder is a smooth function, more precisely~$\fenc \in C^2(\cX,\cZ)$, and~$\nabla \fenc(x) \in \mathbb{R}^{D \times d}$ has full rank~$d$ for all~$x\in\mathrm{Supp}(\mu)$ (we use here the convention that the columns of~$\nabla \fenc$ are the gradients of the components of the function~$\fenc$). 

\paragraph{Analytical reformulation of the maximization problem.}
Our aim here is to rewrite the problem~\eqref{eq:equivalent_maximization_on_fenc} in an analytic form, where expectations are replaced by integrals. We start by writing it as the minimization of
\[
\sS(\fenc) = \int_\cX \left\|g^\star_{\fenc}(\fenc(x))\right\|^2 \mu(x) \, dx, 
\]
where we recall that~$g^\star_{\fenc}(z) = \E[X \, | \, \fenc(X) = z]$ is defined in~\eqref{eq:Bayes_predictor_dec}. From an analytical viewpoint, the conditional expectation can be rewritten as
\begin{equation}
\label{eq:analytical_formula_g_star_fenc}
g^\star_{\fenc}(z) = \frac{\dps \int_{\Sigma_z} x \, \mu(x) \, \delta_{\fenc(x)-z}(dx)}{\dps \int_{\Sigma_z} \mu(x) \, \delta_{\fenc(x)-z}(dx)},
\end{equation}
where~$\Sigma_z \subset \cX$ is the submanifold defined in~\eqref{eq:Sigma_z} (the definition makes sense in view of the assumptions on~$\fenc$), and the delta measure~$\delta_{\fenc(x)-z}(dx)$ on~$\Sigma_z$ is defined by the decomposition formula
\[
dx = \delta_{\fenc(x)-z}(dx) \, dz.
\]
The latter equality should be understood as follows: for any bounded measurable functions~$\varphi:\cX\to\mathbb{R}$ and~$u:\cZ \to \mathbb{R}$ with compact supports,
\begin{equation}
\label{eq:weak_formulation_decomposition}
\int_\cX \varphi(x) u(\fenc(x)) \, dx = \int_\cZ u(z) \left( \int_{\Sigma_z} \varphi(x) \, \delta_{\fenc(x)-z}(dx) \right) dz.
\end{equation}
As discussed in Section~3.2.1.1 of Ref.~\citenum{LRS10}, the delta measure~$\delta_{\fenc(x)-z}(dx)$ can be related to the surface measure~$\sigma_{\Sigma_z}(dx)$ induced by the Lebesgue measure in the ambient Euclidean space~$\mathbb{R}^D$ through the co-area formula~\cite{EG92,AFP00} as
\[
\delta_{\fenc(x)-z}(dx) = \left(\det G(x) \right)^{-1/2} \sigma_{\Sigma_z}(dx),
\]
where we introduced the Gram matrix~$G(x) = \nabla\fenc(x)^\top \nabla\fenc(x) \in \mathbb{R}^{d\times d}$ in order to alleviate the notation.

In order to make sense of derivatives of the conditional expectation, one possible strategy, considered for instance in Section~2 of Ref.~\citenum{GW13}, is to rely on some mollified version of the conditional expectation, namely
\[
g^\varepsilon_{\fenc}(z) = \frac{\dps \int_\cX x \chis(\fenc(x)-z) \, \mu(x) \, dx}{\dps \int_\cX \chis(\fenc(x)-z) \, \mu(x) \, dx},
\]
where
\[
\chis(z) = \frac{1}{\varepsilon^d} \chi\left(\frac{z}{\varepsilon}\right)
\]
is an approximation of the Dirac mass at~$0 \in \mathbb{R}^d$ for~$\varepsilon\to 0$ (the function~$\chi$ being a smooth nonnegative function of integral~1). We follow here an alternative route, where we do not consider regularization by convolutions. More precisely, we provide results ensuring that the functions under consideration are differentiable, by relying on a weak formulation to compute derivatives of integrals with respect to the delta measure; which then allows us to establish~\eqref{eq:formal_EL_condition_optimal_fenc}. For completeness, we also give an expression for the derivatives of these functions, although this will not be used in the proof of~\eqref{eq:formal_EL_condition_optimal_fenc}.

\paragraph{Differentiability results.}
A first differentiability result makes precise derivatives with respect to~$z$. 

\begin{lemma}
	\label{lem:first_lemma_coarea_derivatives}
	Fix a function~$\varphi \in L^1(\cX)$ such that~$\div_x\left(\varphi \nabla\fenc G^{-1} \right) \in L^1(\cX)$. Then, for almost every~$z \in \cZ$,
	\[
	\nabla_z \left( \int_{\Sigma_z} \varphi(x) \, \delta_{\fenc(x)-z}(dx) \right) = \int_{\Sigma_z} \div_x\left(\varphi(x) \nabla\fenc(x)G(x)^{-1} \right) \, \delta_{\fenc(x)-z}(dx), 
	\]
	where the divergence of the~$D \times d$ matrix~$\varphi(x)\nabla\fenc(x)G(x)^{-1}$ on the right-hand side of the previous equality is applied column by column.
\end{lemma}

The above equality should be understood as follows: for any~$\omega \in \mathbb{R}^d$,
\[
\omega^\top \nabla_z \left( \int_{\Sigma_z} \varphi(x) \, \delta_{\fenc(x)-z}(dx) \right) = \int_{\Sigma_z} \div_x\left(\varphi(x) \nabla\fenc(x)G(x)^{-1} \omega \right) \, \delta_{\fenc(x)-z}(dx).
\]
The formula established in Lemma~\ref{lem:first_lemma_coarea_derivatives} allows to compute the derivative of~\eqref{eq:analytical_formula_g_star_fenc} for a given function~$\fenc$, by considering~$\varphi(x) = x_i \, \mu(x)$ with~$1 \leq i \leq D$ for the integral in the numerator (where~$x_i$ is the $i$-th component of~$x \in \mathbb{R}^D$), and~$\varphi(x) = \mu(x)$ for the one in the denominator.

\begin{proof}
	We establish the result for a smooth bounded function~$\varphi$ with compact support, the general case following by density. We start from~\eqref{eq:weak_formulation_decomposition}, written for a smooth bounded function~$u:\cZ \to \mathbb{R}$ with compact support:
	\[
	\int_\cZ \nabla_z u(z) \left( \int_{\Sigma_z} \varphi(x) \, \delta_{\fenc(x)-z}(dx) \right) dz = \int_\cX \varphi(x) (\nabla_z u)(\fenc(x))  \, dx.
	\]
	Since~$\nabla_x (u \circ \fenc) = (\nabla_x \fenc) [(\nabla_z u) \circ \fenc]$ with our convention for~$\nabla_x \fenc$,
	we can rewrite the derivative in~$z$ as a derivative in~$x$ as
	\begin{equation}
	\label{eq:rewriting_nabla_z_u_with_x}
	(\nabla_z u)(\fenc(x)) = G(x)^{-1} \nabla\fenc(x)^\top \nabla_x \left[ u(\fenc(x)) \right].
	\end{equation}
	Therefore,
	\[
	\begin{aligned}
	& \int_\cZ \omega^\top \nabla u(z) \left( \int_{\Sigma_z} \varphi(x) \, \delta_{\fenc(x)-z}(dx) \right) dz \\
	& = \int_\cX \varphi(x) \, \omega^\top G(x)^{-1} \nabla\fenc(x)^\top \nabla_x \left[ u(\fenc(x)) \right] \, dx \\
	& = -\int_\cX \div_x\left(\varphi(x)  \nabla\fenc(x)G(x)^{-1}\omega \right) u(\fenc(x)) \, dx \\
	& = -\int_\cZ u(z) \left( \int_{\Sigma_z} \div_x\left(\varphi(x) \nabla\fenc(x)G(x)^{-1}\omega \right) \, \delta_{\fenc(x)-z}(dx) \right) dz.
	\end{aligned}
	\]
	This allows to identify the desired derivative. 
\end{proof}

In order to understand variations with respect to~$\fenc$, we consider a smooth and bounded perturbation~$\wfenc$ of~$\fenc$. We assume that there exists~$\eta_\star > 0$ (which depends on~$\fenc$ and~$\wfenc$) such that, for any~$\eta \in (-\eta_\star,\eta_\star)$, the smooth function~$\fenc+\eta\wfenc$ has a gradient of full rank~$d$ at every~$x \in \cX$. This can be ensured for instance when~$\wfenc$ has compact support. We can then introduce the functional~$I_\varphi:\cZ \times (-\eta_\star,\eta_\star) \to \mathbb{R}$ defined as
\[
I_\varphi(z,\eta) = \int_{\widetilde{\Sigma}_{z,\eta}} \varphi(x) \, \delta_{(\fenc+\eta\wfenc)(x)-z}(dx), \qquad \widetilde{\Sigma}_{z,\eta} = \left(\fenc+\eta\wfenc\right)^{-1}\{z\}.
\]
Note that
\begin{equation}
\label{eq:sS_eta}
\sS\left(\fenc+\eta\wfenc\right) = \int_\cX \sum_{i=1}^D \left|\frac{\dps I_{\phi_i}\left((\fenc+\eta\wfenc)(x),\eta\right)}{I_{\mu}\left((\fenc+\eta\wfenc)(x),\eta\right)}\right|^2 \mu(x) \, dx, \qquad \phi_i(x) = x_i \, \mu(x). 
\end{equation}
The next result establishes the differentiability of~$I_\varphi$ with respect to~$\eta$.

\begin{lemma}
	\label{lem:second_lemma_coarea_derivatives}
	Fix a function~$\varphi \in L^1(\cX)$ such that~$\div_x\left(\varphi \nabla\fenc G^{-1} \wfenc\right) \in L^1(\cX)$. Then, for almost every~$z \in \cZ$,
	\[
	\partial_\eta I_\varphi(z,0) = -\int_{\Sigma_z} \div_x\left( \varphi \nabla\fenc G^{-1}\wfenc \right)(x) \, \delta_{\fenc(x)-z}(dx).
	\]
\end{lemma}

\begin{proof}
	As for the proof of Lemma~\ref{lem:first_lemma_coarea_derivatives}, we establish the result for a smooth bounded function~$\varphi$ with compact support, the general case following by density. We introduce a smooth function~$u:\cZ \to \mathbb{R}$ with compact support and write
	\[
	\int_\cZ I_\varphi(z,\eta) u(z) \, dz = \int_\cX \varphi(x) u\left((\fenc+\eta\wfenc)(x)\right) dx.
	\]
	The right-hand side of the previous equality is a smooth function of~$\eta$, which can be differentiated. Using again~\eqref{eq:rewriting_nabla_z_u_with_x} to rewrite~$(\nabla_z u) \circ \fenc$ in terms of the gradient of~$u\circ\fenc$,
	\begin{align}
	\left. \frac{d}{d\eta} \left( \int_\cZ I_\varphi(z,\eta) u(z) \, dz\right) \right|_{\eta=0}
	& = \int_\cX \varphi(x) \wfenc(x)^\top \nabla_z u(\fenc(x)) \, dx \label{eq:derivative_eta_before_duality} \\
	& = \int_\cX \varphi(x)\wfenc(x)^\top \left( G(x)^{-1} \nabla\fenc(x)^\top \nabla_x [ u \circ \fenc](x) \right) dx \notag \\
	& = -\int_\cX \div_x\left( \varphi \nabla\fenc G^{-1}\wfenc \right) (u \circ \fenc)\, dx\notag \\
	& = -\int_\cZ u(z) \left(\int_{\Sigma_z} \div_x\left( \varphi \nabla\fenc G^{-1}\wfenc \right)(x) \, \delta_{\fenc(x)-z}(dx) \right) dz, \notag
	\end{align}
	which leads to the claimed result.
\end{proof}

\paragraph{Necessary condition for~$\fenc$ to be a maximizer.}
We finally establish~\eqref{eq:formal_EL_condition_optimal_fenc}. We rely on the differentiability results established in Lemmas~\ref{lem:first_lemma_coarea_derivatives} and~\ref{lem:second_lemma_coarea_derivatives}, and assume to this end that the functions~$\phi_i$ (for $1 \leq i \leq d$) all belong to~$L^1(\cX)$, and also that $\div_x\left(\mu (\nabla\fenc) G^{-1}\right)$, $\div_x\left(\mu \nabla\fenc G^{-1} \wfenc\right)$, $\div_x\left(\phi_i \nabla\fenc G^{-1}\right)$ and~$\div_x\left(\phi_i \nabla\fenc G^{-1} \wfenc\right)$ all belong to~$L^1(\cX)$ (recalling~\eqref{eq:sS_eta}). Denoting by~$\sI_\varphi(z) = I_\varphi(z,0)$, it holds
\[
\begin{aligned}
& \left. \frac{d}{d\eta}\left[\frac12 \sS\left(\fenc+\eta\wfenc\right) \right] \right|_{\eta=0} \\
& \qquad \qquad = \int_\cX \sum_{i=1}^D \left( \wfenc(x)^\top \nabla_z \left(\frac{\sI_{\phi_i}}{\sI_\mu}\right)(\fenc(x)) + \sD_{\phi_i,\wfenc}(\fenc(x))\right)\frac{\dps \sI_{\phi_i}(\fenc(x))}{\sI_{\mu}(\fenc(x))} \mu(x)\, dx, \\
\end{aligned}
\]
where
\[
\sD_{\phi_i,\wfenc}(z) = \frac{\partial_\eta I_{\phi_i}(z,0)}{\sI_\mu(z)}-\frac{\sI_{\phi_i}(z)\partial_\eta I_\mu(z,0)}{\sI_\mu(z)^2}.
\]
Note also that
\[
\frac{\dps \sI_{\phi_i}(z)}{\sI_{\mu}(z)} = g^\star_{\fenc,i}(z).
\]
The aim is to factor out the dependence on~$\wfenc$ in the integral involving~$\sD_{\phi_i,\wfenc}$. We consider the first term on the right hand side of the above equality, and use~\eqref{eq:derivative_eta_before_duality} to write
\begin{align*}
\int_\cX \partial_\eta I_{\phi_i}(\fenc(x),0) \frac{\dps \sI_{\phi_i}(\fenc(x))}{\sI_{\mu}(\fenc(x))^2} \mu(x)\, dx
& = \int_\cZ \partial_\eta I_{\phi_i}(z,0) \frac{\dps \sI_{\phi_i}(z)}{\sI_{\mu}(z)^2} \left(\int_{\Sigma_z} \mu(x) \, \delta_{\fenc(x)-z}(dx)\right) dz \\
& = \int_\cZ \partial_\eta I_{\phi_i}(z,0) \frac{\dps \sI_{\phi_i}(z)}{\sI_{\mu}(z)} \, dz \\
&= \int_\cZ \partial_\eta I_{\phi_i}(z,0) g^\star_{\fenc,i}(z) \, dz \\
& = \int_\cX \phi_i(x) \wfenc(x)^\top \nabla_z g^\star_{\fenc,i}(\fenc(x)) \, dx.
\end{align*}
Similarly,
\[
\int_\cX \partial_\eta I_\mu(\fenc(x),0) \frac{\dps \sI_{\phi_i}(\fenc(x))^2}{\sI_{\mu}(\fenc(x))^3} \mu(x)\, dx
= \int_\cX \mu(x) \wfenc(x)^\top \nabla_z \left[\left(g^\star_{\fenc,i}\right)^2\right](\fenc(x)) \, dx.
\]
A necessary condition for~$\fenc$ to be a critical point is therefore
\[
\begin{aligned}
& \int_\cX \wfenc(x)^\top \left(\sum_{i=1}^D \nabla_z g^\star_{\fenc,i}(\fenc(x)) g^\star_{\fenc,i}(\fenc(x))\right) \mu(x) \, dx \\
& \qquad \qquad + \sum_{i=1}^D \int_\cX \wfenc(x)^\top \left( \frac{\phi_i(x)}{\mu(x)}\nabla_z g^\star_{\fenc,i}(\fenc(x)) - \nabla_z \left[\left(g^\star_{\fenc,i}\right)^2\right](\fenc(x)) \right) \mu(x)\,dx = 0.
\end{aligned}
\]
Since~$\wfenc$ is arbitrary, this shows that the following equality must hold for all~$x \in \cX$ and~$1 \leq j \leq d$:
\[
\sum_{i=1}^D \left(\frac{\phi_i(x)}{\mu(x)}\partial_{z_j} g^\star_{\fenc,i}(\fenc(x)) - \partial_{z_j} g^\star_{\fenc,i}(\fenc(x)) g^\star_{\fenc,i}(\fenc(x))\right) = 0.
\]

This indeed gives~\eqref{eq:formal_EL_condition_optimal_fenc} since~$\phi_i(x) = x_i \mu(x)$ by~\eqref{eq:sS_eta}.

\section{Proof of Lemma~\ref{lem:cond_exp_PCA}}
\label{sec:proof_lem:cond_exp_PCA}

For a given a matrix~$W\in \mathbb{R}^{D\times K}$ such that $W^\top W= \mathrm{I}_K$, let us compute the conditional expectation~$\E(X\,|\,\fred(X))$ for the linear map $\fred(x)=W^\top x$. Introduce to this end a matrix~$U\in \mathbb{R}^{D\times (D-K)}$ such that $U^\top W=0$ and $U^\top U = \mathrm{I}_{D-K}$. 
Since the columns of $W$ and $U$ form an orthonormal basis of $\mathbb{R}^D$, the identities $(WW^\top + UU^\top)W=W$ and $(WW^\top + UU^\top)U=U$ imply that~$WW^\top + UU^\top = \mathrm{I}_D$. 

Given $z \in \mathbb{R}^K$, we can then write (with~$\mu$ the distribution of the centered Gaussian measure under consideration)
\begin{align}
  \E \left(X\,\middle|\,\fred(X)=z\right)
  &= \E \left(X\,\middle|\,W^\top X = z\right) \notag \\
  &= Z_z^{-1}\int_{\{x\in \mathbb{R}^D : W^\top x = z\}} x\, \mu(dx) \notag\\
  &= Z_z^{-1} \int_{\{x\in \mathbb{R}^D : W^\top x = z\}} \left(WW^\top x + UU^\top x \right) \mu(dx) \notag\\
  &= Wz + Z_z^{-1} U\int_{\{x\in \mathbb{R}^D : W^\top x = z\}} U^\top x\, \mu(dx), \label{eq:cond_Gaussian_PCA}
\end{align}
where
\[
Z_z = \int_{\{x\in \mathbb{R}^D : W^\top x = z\}} \mu(dx).
\]
Note that, since $X$ is a centered Gaussian random variable and $U^\top W=0$, the two vectors $U^\top X$ and $W^\top X$ are independent centered Gaussian random variables in~$\mathbb{R}^{D-K}$ and $\mathbb{R}^{K}$, respectively. Therefore, the second term in~\eqref{eq:cond_Gaussian_PCA} vanishes, and
\[
\E \left(X\,\middle|\,\fred(X)\right) = WW^\top X.
\]
This shows that the two minimization tasks~\eqref{pca-objective} and~\eqref{eq:dim-reduction-pca} are equivalent. 

\section{Details of the molecular dynamics simulations of Section~\ref{sec:ad}}
\label{app:MD}

Alanine dipeptide was put into a cubic simulation box with periodic boundary condition, with box sizes set to the diameter of the system plus~$2.4$~nm. Water molecules were added into the box, resulting in~$3,469$ atoms in total. The force field AMBER99SB-ILDN~\cite{amber99sb-ildn} and TIP3P water models were adopted for alanine dipeptide and water, respectively.
Energy minimization was performed using steepest descent minimization, with initial step-size set to~$0.01$~nm in the GROMACS input file and maximum number of minimization steps set to $500$.
The convergence was reached after~$289$ steps when the maximum of the absolute values of the force components were below $500$~kJ~mol$^{-1}$~nm$^{-1}$.

The MD simulation to obtain the training data was conducted using a leap-frog stochastic dynamics integrator with time step~$1$~fs. The linear constraint solver (LINCS) of order~$4$ was adopted to impose holonomic constraints (fixing bond lengths involving H-atoms). Verlet lists were used to compute interactions between neighboring atoms, the frequency to update the neighbor list being initially set to~$20$ steps and the cut-off radius to~$1.0$~nm. Both Coulomb and van der Waals potentials were shifted to zero at the cut-off distance $1.0$~nm.
Electrostatic interactions were computed using the fourth order fast smooth particle-mesh Ewald (SPME) method with grid of size $0.16$~nm. A modified Berendsen thermostat was used with time constant $1$~ps and reference temperature $300$~K.
After equilibrating the system by a NVT simulation for $100$~ps followed by a NPT simulation for $200$~ns, we simulated the system for $1.5~\mu$s and, by recording the state every $10$~ps, we obtained a dataset consisting of~$1.5\times 10^5$ states along the system's trajectory.


\bibliography{AE_biblio}

\end{document}